%% file: new_paper.tex
\documentclass[twocolumn, superscriptaddress,nofootinbib,floatfix]{revtex4-1} %
\usepackage[a4paper, left=0.75in, right=0.75in, top=0.7in, bottom=0.7in]{geometry}

\usepackage[utf8]{inputenc}
\usepackage[english]{babel}
\usepackage[T1]{fontenc}
\usepackage{amsmath}
\usepackage{bbm}
\usepackage{amsfonts}
\usepackage{hyperref}
\usepackage{xcolor}
\usepackage{times}

\usepackage[caption=false]{subfig}
\allowdisplaybreaks

\usepackage{tikz}
\usepackage{natbib}
\usepackage{amsthm}
\usepackage{physics}
\usepackage{amssymb}
\usepackage{braket}
\usepackage{algorithm}
\usepackage{algorithmicx}
\usepackage{algpseudocode}

\usepackage{comment}

\newtheorem{definition}{Definition}
\newtheorem{lemma}{Lemma}

\newtheorem{proposition}{Proposition}
\newtheorem*{proposition*}{Proposition}

\usepackage{tikz-cd} 
\usepackage{adjustbox}

\input{commands.tex}

\newcommand{\ie}{\textit{i.e.}\ }
\newcommand{\eg}{\textit{e.g.}\ }

\newcommand{\bbb}[1]{\textbf{#1}}

\newcommand{\tmom}{t}

\newcommand{\fu}{Dahlem Center for Complex Quantum Systems, Freie Universit\"{a}t Berlin, 14195 Berlin, Germany}

\newcommand{\hzb}{Helmholtz-Zentrum Berlin f{\"u}r Materialien und Energie, 14109 Berlin, Germany}
\newcommand{\hhi}{Fraunhofer Heinrich Hertz Institute, 10587 Berlin, Germany}

\begin{document}

\title{Learning fermionic correlations by evolving with random\\ translationally invariant Hamiltonians}
\date{\today}

\author{Janek Denzler}
\email[]{denzlerj@pm.me}
\affiliation{\fu}

\author{Antonio Anna Mele}
\affiliation{\fu}

\author{Ellen Derbyshire}
\affiliation{\fu}

\author{Tommaso Guaita}
\email[]{tommaso.guaita@fu-berlin.de}
\affiliation{\fu}

\author{Jens Eisert}
\affiliation{\fu}
\affiliation{\hzb}
\affiliation{\hhi}

\begin{abstract}
Schemes of classical shadows have been developed to facilitate the read-out of digital quantum devices, but similar tools for analog quantum simulators are scarce and experimentally impractical. In this work, 
we provide a measurement scheme for fermionic quantum devices that estimates
second and fourth order correlation functions by means of free fermionic, translationally invariant evolutions --- or quenches --- and measurements in the mode occupation number basis. We precisely characterize what correlation functions can be recovered and equip the estimates with rigorous bounds on sample complexities, a particularly important feature in light of the difficulty of getting good statistics in reasonable experimental platforms, with measurements being slow. Finally, we demonstrate how our procedure can be approximately implemented with just nearest-neighbour, translationally invariant hopping quenches, a very plausible procedure under current experimental requirements, and requiring only random time-evolution with respect to a single native Hamiltonian. On a conceptual level, this work brings 
the idea of classical shadows to the realm of large scale analog quantum simulators.
\end{abstract}

\maketitle
The last decade has seen tremendous theoretical and experimental development of quantum devices, with hopes of using such machines to infer properties of physical systems beyond the limits of classical computation.
Of particular interest is understanding properties of fermionic systems, given their applications in condensed matter physics, quantum chemistry and high-energy physics. Specifically the field of fermionic \emph{analog quantum simulation} has received considerable attention~\cite{CiracZollerSimulation,Trotzky}. Platforms based on ultracold atoms in optical lattices have advanced significantly, reaching a high degree of control and single site resolution in occupation number measurements, making it relevant to determine whether fermionic quantum devices can be expected to outperfom qubit-based approaches in simulating interacting fermions. \cite{RoschTransport,GrossFermions,Mazurenko,EsslingerReview,AspuruGuzikFermions}.

That said, any mature experimental effort in this direction that hopes to have strong predictive power will have to adhere to rigorous certification demands~\cite{eisert_quantum_2020,kliesch_theory_2021}, as well as operate with resource-efficient \emph{read-out methods}, which are unfortunately lacking.  
While in the realm of digital quantum computing methods of randomized measurements
\cite{Toolbox}, of benchmarking \cite{MagGamEmer,KnillBenchmarking,eisert_quantum_2020}, and of classical shadows \cite{huang_predicting_2020,huang_quantum_2021,chen_robust_2021,ShallowShadows,ohliger_efficient_2013} have risen to prominence due to their versatility and practicability, the same cannot be said for 
\emph{analog quantum simulators}, where the need for good read-out methods is 
equally pressing. Although classical shadow approaches have begun to be adapted to fermionic systems~\cite{zhao_fermionic_2021,Wan_2022,GuangHao+22}, current protocols require the implementation of mode transformations that are out of scope for state-of-the-art optical lattice experiments. These experiments, for instance, natively implement only nearest-neighbour hopping evolutions with hopping terms that cannot easily be tuned on a single site level but are essentially translationally invariant~\cite{Kuhr}. Complex-valued hopping terms are also challenging to implement, although progress has been made in some two-dimensional settings~\cite{aidelsburger_artificial_2016}.
What, then, is the most reasonable read-out protocol for fermionic systems with \emph{translationally invariant, free evolution} as their natural mode of operation?
\begin{figure}[ht]
\centering
\includegraphics[width=0.48\textwidth]{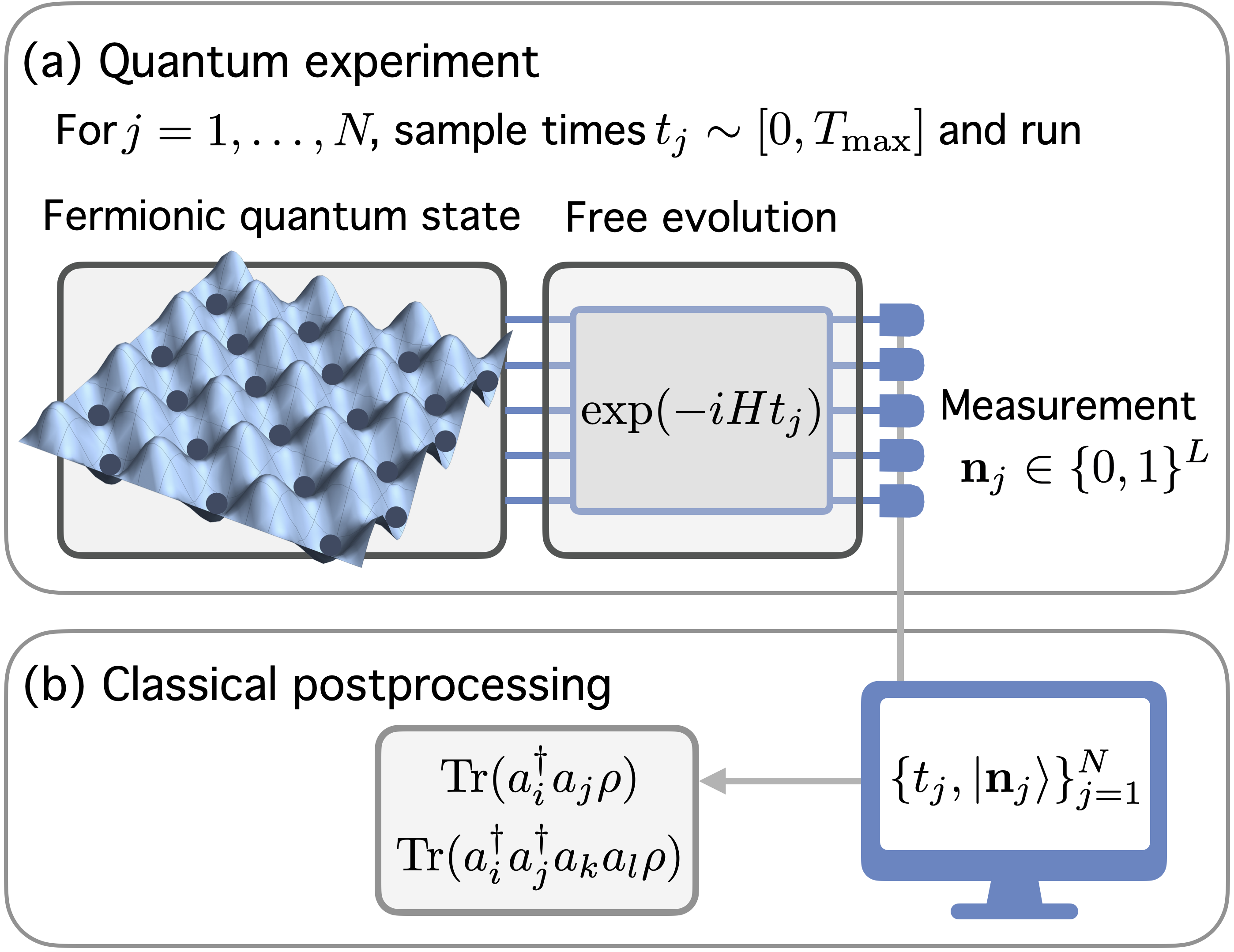}
\caption{Scheme of our protocol, where $H$ is a translationally invariant, non-interacting fermionic Hamiltonian.}
\label{fig:fermiqp_scheme}
\vspace{-0.6cm}
\end{figure}

\vspace{-0.1cm}
In this work, we present a compelling answer to this question. We do so by developing a machinery of randomized measurements for analog quantum simulators, building on classical shadows~\cite{huang_predicting_2020}, that requires only free fermionic (particle number preserving) evolutions and mode occupation number measurements. Crucially, unlike previous analyses~\cite{zhao_fermionic_2021,Wan_2022,GuangHao+22}, we \emph{do} restrict the allowed evolutions to be \emph{translationally invariant}, while the measurements 
are site-resolved. Both these elements are readily available experimentally, the latter famously in quantum gas microscopes~\cite{GreinerMicroscope,Microscope,Kuhr}. The core quantities that we estimate are arbitrary $2$-point and $4$-point correlation functions, which are essential for determining the underlying properties of strongly correlated fermionic systems, as they fully characterise most meaningful fermionic many-body Hamiltonians. Notice that the protocol as a whole is not translationally invariant (due to the site-resolved measurements), so the recovery procedure does not need to be fundamentally limited to translationally invariant correlation functions or states. 

We analyse this protocol and its the experimental feasibility on several levels. Firstly, we characterise entirely the subset of observables that can be recovered if we simply impose translational invariance of the unitary evolutions. Next, we show that in this setting it is actually possible to estimate all 2- and 4-point correlation functions, provided one adds an $\mathcal{O}(1)$ number of auxiliary lattice sites, and we further discuss sample complexity guarantees. This is of particular importance in a setting in which single shots can take seconds \cite{PhysRevLett.115.263001}, rendering them a precious resource. Finally, we provide analytical and numerical evidence that nearest-neighbour interactions with real hopping amplitudes are enough to approximate the previously characterised protocol, essentially only requiring quenches with respect to a fixed, simple nearest-neighbour Hamiltonian. This is in stark contrast to existing shadow estimation schemes, where the system Hamiltonian must be constantly reprogrammed.

Our novel adaptation of shadow estimation to the realm of experimentally feasible analog quantum simulation builds upon a body of prior work, much of it very recent. The presumably first work offering a method for efficient read-out of
quantum many-body systems 
resorts to random measurements and 
tensor network recovery methods \cite{ohliger_efficient_2013}.
Refs.~\cite{QuantumReadout,gluza_recovering_2021} provide a practical
compressed sensing inspired procedure to 
estimate fermionic 2-point correlation functions 
by means of free 
fermionic evolutions and mode occupation number 
measurements,
while Ref.~\cite{NaldesiZoller} 
builds on this, introducing randomized evolution for more stringent guarantees, as well as allowing for the estimation of 4-point correlation functions.
Only beginning with Ref.~\cite{zhao_fermionic_2021} has the power of the full classical shadows framework \cite{huang_predicting_2020} been 
employed in the context of fermions, 
allowing for the estimation of fermionic correlation functions of arbitrary order. Ref.~\cite{GuangHao+22} builds on this work, improving the sample complexity and restricting to more feasible particle number preserving operations. 

However, all these methods rely on the application of complex random mode transformations to the system, which are generated by complex, long-range, 
non-translationally invariant Hamiltonians. In a digital setting, such evolution can in principle be decomposed into local gates, but this is currently hard for analog quantum simulators, for reasons stated above.

Other works \cite{Tran+22} consider settings where a system, possibly extended by a number of auxiliary modes, undergoes a single quench by a Hamiltonian chosen among those that can be naturally implemented in the experimental platform. These studies, however, do not characterise in detail what can be recovered if the quench Hamiltonian is restricted to be translationally invariant or local.
Ref. \cite{vankirk2022hardwareefficient} also considers settings with limited experimental control. They study in detail two unitary ensembles, one of which is translationally invariant. However, in the context of optical lattice experiments involving fermions, these quenches are not readily available. We, on the other hand, introduce a scheme
of \emph{classical shadows suitable
for analog quantum simulation}. 
We achieve this substantial step towards experimental feasibility while not giving up
any mathematical rigour.

\textit{Preliminaries}.
We consider a fermionic one-dimensional
lattice system with periodic boundary conditions, composed of $L$ sites (modes) with creation and annihilation operators $a_j^\dagger$  and $a_j$ for $j \in \{0,\dots, L-1\}$, respectively, subject to the canonical anti-commutation relations. The fundamental elements of our randomized measurements protocol will be particle number preserving, free unitary evolutions, that is unitaries of the form $U={\rm exp}(it\sum^{L-1}_{i,j=0} h_{i,j} a^\dagger_i a_j)$ for some Hermitian matrix $h$. These quenches correspond to $L \times L$ unitary matrices by their action on the mode operators. Indeed we have that $U a_i U^\dagger=\sum_{j=0}^{L-1} u_{i,j}a_j$ for some matrix $u\in \mathrm{U}(L)$, which in turn uniquely determines $U$, up to a phase.

We will consider only translationally invariant evolutions, that is $U$ such that $[U,T]=0$, where $T$ is the shift operator that translates all the modes by one site along the one-dimensional chain. We further impose that the coefficients $h_{i,j}$ of the Hamiltonian generating this evolution are real. It follows that the relevant unitary transformations in mode space constitute the subgroup of symmetric circulant unitaries $CU_{\text{Sym}}(L)\subset U(L)$, \ie unitaries whose matrix elements satisfy $u_{i,j}=u_{i+k,j+k}$ for all $k$, and $u_{i,j}=u_{j,i}$, where all indices are to be understood modulo $L$. We will see later that unitaries in this subgroup can be suitably approximated by evolving with respect to a fixed nearest-neighbour hopping Hamiltonians with real coefficients. 

Let us now introduce the classical shadows formalism, tailored to fermionic systems~\cite{zhao_fermionic_2021,Wan_2022,GuangHao+22,NaldesiZoller}. It relies on first defining the \textit{measurement channel} $\mathcal{M}$, which acts on fermionic operators as
\begin{align}
    \mathcal{M}(\cdot):=\sum_{\boldsymbol{n}\in \{0,1\}^L}\mathbb{E}_U\,\bra{\boldsymbol{n}}U (\cdot) U^\dagger\ket{\boldsymbol{n}} U^\dagger\ketbra{\boldsymbol{n}}{\boldsymbol{n}}U \,,
    \label{measurementop}
\end{align}
where the sum is over all Fock basis vectors $\ket{\boldsymbol{n}}$  
of the fermionic Hilbert space and
$\mathbb{E}_{U  }$ is the expectation value with respect to sampling $U$ from some unitary ensemble $\mathcal{U}$.
The channel $\mathcal{M}$ is self-adjoint but in general not invertible. We can, however, always define its pseudo-inverse $\mathcal{M}^+$, which annihilates the channel's kernel and acts as an inverse on its image.

The randomized measurement protocol is as follows: First, we evolve our state $\rho$ according to a unitary $U$ drawn uniformly from $\mathcal{U}$, resulting in $U\rho U^\dagger$. We then measure the evolved state in the occupation number basis, giving some outcome $\mathbf{n}$. We repeat this $N$ times, storing the outcomes $(U, \boldsymbol{n})$ efficiently in a classical memory. In post-processing, to estimate an operator $O$, we compute the single-shot estimates $X_{U,\boldsymbol{n}} := \braket{\boldsymbol{n}| U \mathcal{M}^+(O) U^\dag |\boldsymbol{n}}$. Averaging over repeated measurements converges to
\begin{equation} 
        \underset{U,\boldsymbol{n}}{\mathbb{E}} X_{U,\boldsymbol{n}} = \Tr\left[\mathcal{P}_{\Im(\mathcal{M})}(O) \rho \right] \,,
\end{equation}
where $\mathcal{P}_{\Im(\mathcal{M})}$ is the orthogonal projector onto the image of $\mathcal{M}$. This can be seen by using $\mathcal{M}\circ \mathcal{M^+}=\mathcal{P}_{\Im(\mathcal{M})}$. Thus, it follows that $X_{U,\boldsymbol{n}}$ is an unbiased estimator for the expectation value $\Tr(O\rho )$, 
if $O$ is in the image of $\mathcal {M}$.

\textit{Measurement channel for translationally invariant unitaries}. 
This sketch of the protocol highlights how, to successfully perform the post-processing step, we need a complete characterisation of the \emph{channel} $\mathcal{M}$, its \emph{pseudo-inverse}, and its \emph{image}.
Previous works \cite{GuangHao+22, NaldesiZoller} rely on sampling the unitaries $U$ uniformly from the group of free fermionic evolution, with at most the constraint of particle number preservation. This means $\mathcal{U}\equiv U(L)$, for which $\mathcal{M}$ has a particularly simple form. For our purposes, we instead consider the unitary mode transformations $u$ being sampled uniformly (according to the Haar measure) from $CU_{\rm Sym}(L)$.

Our first observation is that for any ensemble of free fermionic, particle number preserving unitaries, the channel is block diagonal with respect to the subspaces of particle number preserving $2k$-point operators, that is operators that can be written as linear combinations of terms of the form $a^\dag_{i_1}\cdots a^\dag_{i_k} a_{j_1}\cdots a_{j_k}$. The action of $\mathcal{M}$ on a $2k$-point operator gives a linear combination
\begin{align}
    \mathcal{M}&(a^\dag_{i_1}\cdots a^\dag_{i_k} a_{j_1}\cdots a_{j_k})= \nonumber\\
    &\hspace{15mm}\sum_{\boldsymbol{l},\boldsymbol{m}} M^{(2k)}_{\boldsymbol{i} \boldsymbol{j},\boldsymbol{l} \boldsymbol{m}} a^\dag_{l_1}\cdots \: a^\dag_{l_k} a_{m_1}\cdots a_{m_k}
\end{align}
of other $2k$-point operators of the same order, where bold letters indicate a collection of $k$ indices, \eg, $\boldsymbol{i}\equiv(i_1,\dots ,i_k)$.
We can, therefore, focus on each $k$ individually. This is a fundamental step in reducing the complexity of the problem, as we move from a channel acting on all Fock space operators to maps acting on spaces of dimension of order $L^{2k}$. We will focus in particular on the 2- and 4-point sectors, as these include essentially all observables relevant for fermionic quantum simulation tasks.

\begin{proposition}[Expression for channel]
For $L$ odd, the tensors $M^{(2k)}_{\boldsymbol{i} \boldsymbol{j},\boldsymbol{l} \boldsymbol{m}}$ 
specifying
each of the blocks of $\mathcal{M}$ are given by 
\begin{equation}
    M^{(2k)}_{\boldsymbol{i} \boldsymbol{j},\boldsymbol{l} \boldsymbol{m}}= \!\!\!\sum_{\substack{\boldsymbol{p}\\ \pi\in\mathcal{S}_k}} \!\mathrm{sgn}(\pi) \: \Phi_{2k}(\boldsymbol{i},\boldsymbol{m},\pi\boldsymbol{j},\boldsymbol{l};\boldsymbol{p},\boldsymbol{p},\boldsymbol{p},\boldsymbol{p}), \label{eq:Mchannel}
\end{equation}
where $\mathcal{S}_k$ is the permutation group of $k$ elements and
\begin{align}
    \Phi_{2k}&(\boldsymbol{i},\boldsymbol{j},\boldsymbol{l},\boldsymbol{m};\boldsymbol{p},\boldsymbol{q},\boldsymbol{r},\boldsymbol{s})=\nonumber\\
    &\hspace{10mm}\mathbb{E}_u  \prod_{a=1}^k\overline{u}_{i_a,p_a} \, \overline{u}_{j_a,q_a} \: u_{l_a, r_a} \, u_{m_a, s_a}
\end{align}
are matrix elements of the $2k$-th moment operator of the unitary ensemble $\mathcal{U}$, that is in our case $CU_{\rm{Sym}}(L)$. 
\end{proposition}
Our main results rely on representation theoretic concepts to compute expressions for these moment operators. Indeed, in general the moment operators can be expressed in terms of the elements of the commutant of $u^{\otimes 2k}$, which leads to the expression 
\begin{align}\label{eq:moment_op}
    \Phi_{2k}&(\boldsymbol{i},\boldsymbol{j},\boldsymbol{l},\boldsymbol{m};\boldsymbol{p},\boldsymbol{q},\boldsymbol{r},\boldsymbol{s})=\nonumber\\
    &\hspace{1mm} \sum_{\substack{\boldsymbol{x},\boldsymbol{y}:\\\boldsymbol{x}\sim\boldsymbol{y}}}\:\prod_{a=1}^k \overline{v}^{(x_a)}_{i_a} \,\overline{v}^{(x_{k+a})}_{j_a} v^{(y_a)}_{l_a} v^{(y_{k+a})}_{m_a} \cdot \\
    &\hspace{20mm} \cdot v^{(x_a)}_{p_a} v^{(x_{k+a})}_{q_a} \,\overline{v}^{(y_a)}_{r_a} \,\overline{v}^{(y_{k+a})}_{s_a}\,,\nonumber
\end{align}
where the sum runs over all pairs of $2k$-dimensional indices $\boldsymbol{x}\equiv(x_1,\dots ,x_{2k})$ and $\boldsymbol{y}\equiv(y_1,\dots ,y_{2k})$ equivalent under the relation $\sim $, defined such that $\boldsymbol{x}\sim \boldsymbol{y}$ if and only if there exist a permutation $\pi \in \mathcal{S}_{2k}$ and a vector of signs $s\in \{+1,-1\}^{2k}$ such that $\boldsymbol{y}=\pi (s_1 x_1, \dots, s_{2k} x_{2k})$. The vectors $v^{(a)}$ are elements of the Fourier basis $v^{(a)}_b=(1/\sqrt{L} ) \exp(2\pi i \,ab/L)$.

In the 2-point sector ($k=1$), combining relations~\eqref{eq:Mchannel} and~\eqref{eq:moment_op} leads to an analytical expression 
\begin{equation}
        M^{(2)}_{i j, lm} = \frac{1}{L}\left( \delta_{i,j}\delta_{l,m}+\delta_{i,l}\delta_{j,m}-\dfrac{1}{L}\delta_{i-j,l-m}\right)\,. \label{eq:M_2}
\end{equation}
A straightforward analysis of~\eqref{eq:M_2} shows that the image of this map is given by the span of all operators of the form 

\noindent
\begin{align}
    a_i^\dagger a_j - a_{L-1}^\dagger a_{j-i-1} &\text{ for all } i\neq j \label{eq:2-point-eigenstates},\\
    a_i^\dagger a_i \hspace{22mm}&\text{ for all } i \,.
\end{align}
These are all the 2-point correlation functions that are possible to recover using the randomized measurement scheme discussed above (and all other incoherent measurement protocols without ancillas that adhere to our experimental restrictions). In particular, the operators~\eqref{eq:2-point-eigenstates} are also eigenvectors of $\mathcal{M}$ with eigenvalue $1/L$.

In the 4-point sector ($k=2$), enumerating all the terms appearing in equation~\eqref{eq:moment_op} is a less trivial combinatorial task. The result can however be used to evaluate the tensor $M^{(4)}_{\boldsymbol{i} \boldsymbol{j},\boldsymbol{l} \boldsymbol{m}}$, which will then need to be pseudo-inverted numerically (see the Supplemental Material \cite{supp} for details).

\textit{Recovery of arbitrary $2$- and $4$-point correlation functions}.
This analysis of the measurement channel shows that there are fundamental constraints on the 2- and 4-point observables that we can recover using a randomized measurement scheme with only evolutions in $CU_{\rm{Sym}}(L)$. 
However, we are able to circumvent these limitations if we have the possibility to extend the measured system by adding a small number of auxiliary modes. 

In particular, we consider the case in which it is possible to prepare the system in the state $\rho'=\rho\otimes |\boldsymbol{0}\rangle\langle\boldsymbol{0}|$, where $\rho$ is the state that we would like to measure (defined on $L$ modes) and $|\boldsymbol{0}\rangle\langle\boldsymbol{0}|$ represents some further $L_{\rm anc}$ auxiliary modes, initially prepared in the unoccupied state (geometrically speaking, the auxiliary modes are arranged within the same one-dimensional ring as the system modes). This should be a relatively  natural procedure in a cold atom set-up, where one can introduce a large local chemical potential on a subset of sites to keep them unoccupied during the state preparation process. Subsequently, one evolves the system with the randomly drawn translationally invariant unitaries, letting it expand also into the auxiliary modes, and finally measures the occupation numbers of the whole system, including the auxiliary modes, with an approach
somewhat reminiscent of ``time of flight'' experiments.
We show that with the data collected in this way one can recover arbitrary 2-point correlations ${\Tr} (a^\dag_i a_j \rho)$, or 4-point correlations ${\Tr}(a^\dag_i a^\dag_j a_k a_l \rho)$. To do this, we construct an operator $O'$ that belongs to the image of the measurement channel on the extended system, according to our previous characterisation, but also satisfies $\Tr(O'\rho')=\Tr(O \rho)$, where $O$ is an arbitrary 2- or 4-point correlation operator. It follows that the single-shot estimator for $O'$ on the full system gives a suitable post-processing function for recovering $\Tr(O \rho)$.

For the 2-point sector, we see that it is in fact sufficient to add just a single auxiliary mode to recover arbitrary correlation functions, provided $L_{\rm tot} = L + L_{\rm anc}$ is odd: Consider our $L$ mode system to be embedded into a $L+1$ mode lattice, where the modes labelled by $0,\dots,L-1$ correspond to the system and the mode labelled by $L$ is an initially unoccupied auxiliary mode. It is easy to check that the observable $O'=a^\dag_i a_j-a^\dag_L a_{j-i-1}$ is in the image of $\mathcal{M}$ on the full $L+1$ mode lattice, and that ${\rm Tr}(O'\rho')={\rm Tr}(a^\dag_i a_j \rho)$. Therefore, $\braket{\boldsymbol{n}| U \mathcal{M}^+(O') U^\dag |\boldsymbol{n}}$ is an unbiased estimator for ${\rm Tr}(a^\dag_i a_j \rho)$, evaluating to
\begin{align}
\label{eq:2pointest}
    \!\!\!X^{(i,j)}_{u,\boldsymbol{n}} \!:=  \!(L+1)  \!\!\left[   \sum^{L}_{k=0} \left(\overline{u}_{i,k}u_{j,k}\! - \overline{u}_{L,k}u_{j-i-1,k}\right)\!n_k\!\right]\!\!.
\end{align}
For the 4-point sector, a suitable operator $O'$ can be constructed through a numerical procedure. We observe that this gives successful outcomes with a small number of auxiliary modes
(see Supplemental Material \cite{supp}). 

Let us now address the number of samples (or measurement shots) that are needed to estimate correlations functions in this way.
An application of Hoeffding's inequality shows that to estimate all 2-point correlation functions to within an additive error $\epsilon$ with success probability $1-\delta$, \begin{equation} N\ge ({16}/{\epsilon^{ 2}})\left(L+1\right)^2\log\left({4L^2}/{\delta}\right).\end{equation}
samples are sufficient.

\textit{Implementation with nearest-neighbour Hamiltonians.} 
Up until now, our recovery procedure has been for ensembles of free fermionic unitaries drawn from the Haar measure on
$CU_{\rm Sym}(L)$. However, our goal is to recover correlation functions with simple Hamiltonian quenches. 
Consider therefore an ensemble of unitaries generated by the Hamiltonian
\begin{align}
    H_{\text{n.n}} = -J \left(\sum^{L-1}_{i=0} a^{\dag}_{i}a_{i+1} + a_{i+1}^\dag a_i \enspace\right).
\end{align}
Each evolution lasts some time $t$, and is characterised by the real number $\alpha=Jt$. In the Supplemental Material \cite{supp}, we present analytical and numerical evidence, on top of the numerical experiments in the subsequent section, that for arbitrary $L$ odd, the Haar measure on $CU_{\rm Sym}(L)$ can be approximated by this simple ensemble of quenches, provided that $\alpha$ is sampled according to suitable distributions. We take into account that the hopping coupling $-J$ is negative in most architectures~\cite{Kuhr}. For example, to  approximate second order moments up to an error $\epsilon$ (in the $2$-norm), it suffices to sample $\alpha$ uniformly from $[0, \alpha_{\text{max}}]$, where $\alpha_{\text{max}} = \mathcal{O}(L^4/\epsilon^2)$.
Note that one can fix $J$ and vary the quench time $t$, or use a product of $J t$ to reach the desired $\alpha$ values.

\emph{Numerical simulations.}  
To demonstrate the viability of our estimation procedure, we present a motivating example using experimentally feasible values for $\alpha_{\rm max}$. We prepare the initial state vector $|\psi\rangle = \exp\left(-i H_{\text{hub}}t_0\right)|\psi_0\rangle$, where $|\psi_0\rangle = |0, 1, \dots , 1, 0\rangle$ is the charge-density wave state vector and $H_{\text{hub}}$ is the Fermi-Hubbard Hamiltonian 
\begin{eqnarray}
H_{\text{Hub}} \!= \!\sum\limits_{i=0}^{L-2} (a_i^\dagger a_{i+1} + h.c.) + \! \sum\limits_{i=0}^{L-2} \! n_i n_{i+1} + \! \sum\limits_{i=0}^{L-1} \!g_i n_i,  \enspace
\end{eqnarray}
 with a random-field $g_i \sim[0.2, 0.7]$ and $t_0 = 1.5$.

The protocol is illustrated in Figure~\ref{fig:fermiqp_scheme}. We embed the above state vector $|\psi\rangle$ into a system of size $L_{\text{tot}} = L + L_{\text{anc}}$, where $L_{\text{tot}}$ is odd. The full state vector is therefore $|\psi^\prime\rangle = |\psi\rangle \otimes |\boldsymbol{0}\rangle$. Then, we sample $\alpha \sim [0,\alpha_{\text{max}}]$, evolve the state vector $\ket{\psi'}$ with the corresponding unitary $\exp(-iH_{\text{n.n}}t)$, and measure in the occupation number basis, yielding $\boldsymbol{n}$. We repeat this $N$ times and post-process the data as detailed above.
Figure~\ref{fig:colourplots} provides illustrative examples of our estimation procedure for $2$- and $4$-point functions: A compelling visual agreement of our estimation procedure is evident from the colour plots.
\begin{figure}[ht]
\centering
\includegraphics[width=0.49\textwidth]{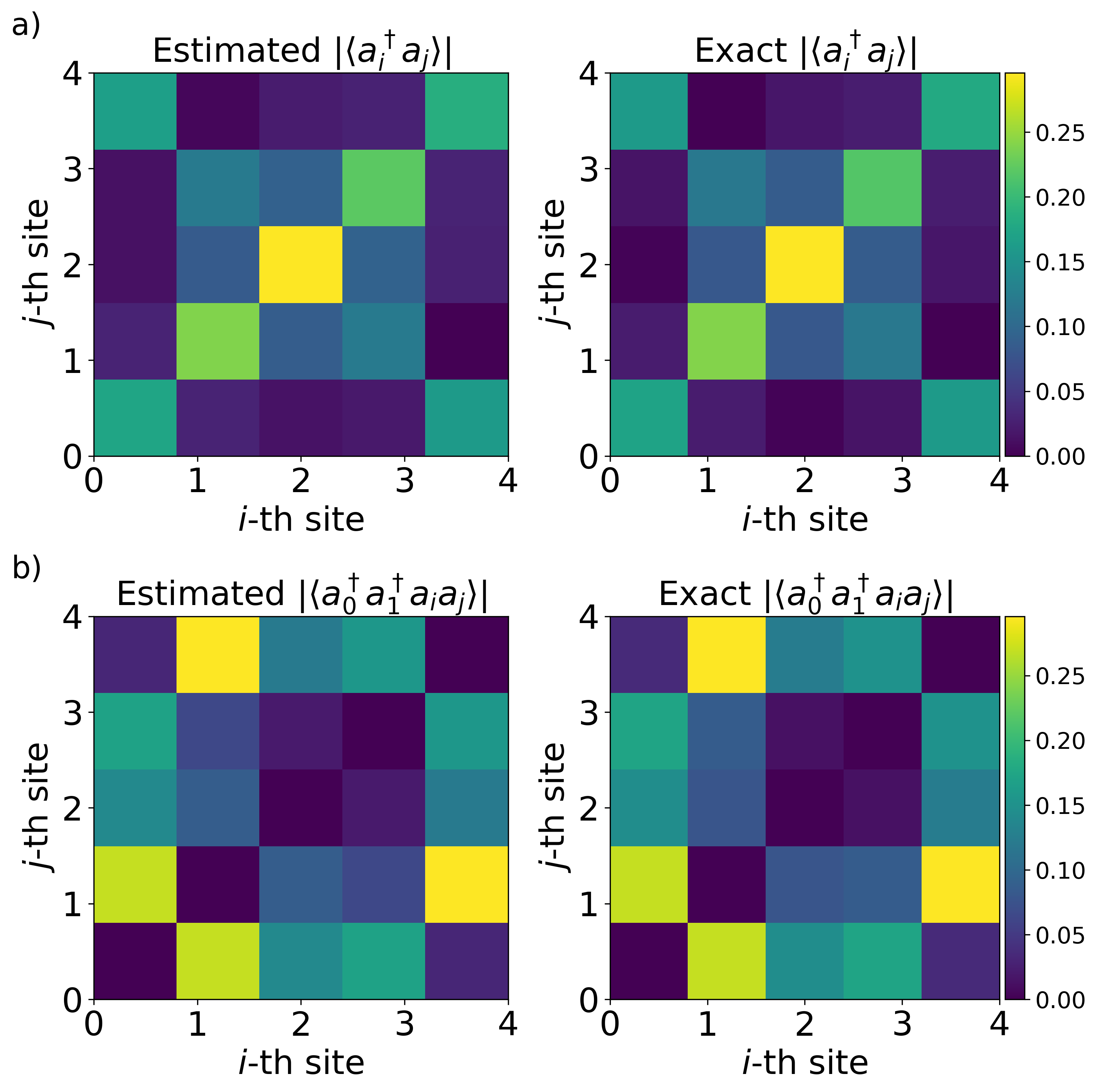}
\caption{Colour plots showing estimates and exact values for (a) $|\langle a_i^\dagger a_j\rangle|$, and (b) $|\langle a_0^\dagger a_1^\dagger a_i a_j\rangle|$. The data is derived from $5000$ quench samples, with $\alpha$ uniformly distributed in $[0,\alpha_{\text{max}}]$, where $\alpha_{\text{max}} = 120$. The system size is $L = 5$, with $L_{\rm tot}=7$. 
}
\label{fig:colourplots}
    \vspace{-0.3cm}
\end{figure}

To study the behaviour of our procedure in more detail, we implement it for various values of $\alpha_{\text{max}}$. The inset of Figure~\ref{fig:VarScaling} indicates that the average variance of our estimators (over all $2$-point functions) is independent of the value of $\alpha_{\text{max}}$ and appears to grows polynomially with $L$. This implies that the number of samples required for a given accuracy will grow correspondingly. More concretely, we find that, for $\alpha_{\text{max}} = 120$ and $L\leq 8$, the number of samples required to estimate 2-point functions within accuracy $\mathcal{O}(10^{-3})$ is $N \sim 70000$, or around $N \sim 20000$ for an accuracy $\mathcal{O}(10^{-2})$, see additional figures in the Supplemental Material \cite{supp}. 

Figure \ref{fig:VarScaling} also shows the average recovery error $\text{Ave} (\Delta O_{i,j}) = {L^{-2}}\sum_{i,j}|\langle a_i^\dag a_j\rangle - \langle a_i^\dag a_j\rangle_{\rm est}|$ with $\alpha_{\text{max}}$ and $L$, where $\langle a_i^\dag a_j\rangle_{\rm est}$ is the recovered expectation value. We see that the lowest $\alpha_{\text{max}} = 20$ is insufficient to produce estimates within an average error of $10^{-2}$, and this only gets worse with increasing $L$. We observe that the larger $\alpha_{\text{max}}$ gets, the less steeply the average error $\text{Ave} (\Delta O_{i,j})$ increases with $L$. For the analysed system sizes, our combined results indicate that the recovery of 2-point functions is possible with reasonable values of $\alpha_{\text{max}}\sim10^2$ and $N\sim10^4$. The analogous plots for 4-point correlation functions are discussed in the Supplemental Material \cite{supp}.
\begin{figure}[ht]
    \centering
    \includegraphics[width=0.45\textwidth]{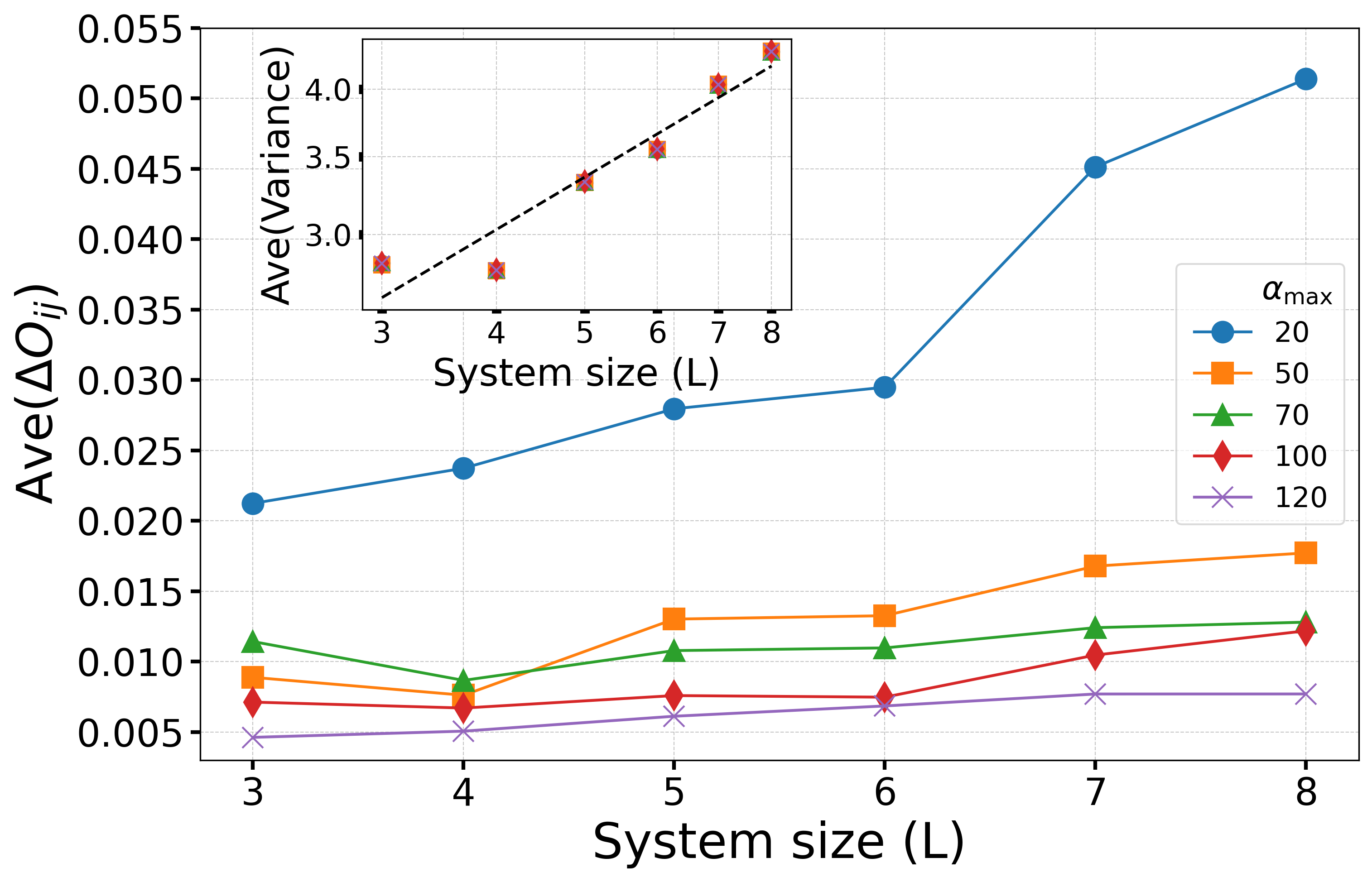}
    \caption{Average recovery error $\text{Ave}(\Delta O_{i,j})$ as a function of $L \leq 8$ for varying $\alpha_{\text{max}}$. Inset: Log-log plot illustrating the average variance of 2-point estimators for the same $L$, displaying a linear fit with a slope of $\sim 0.47$. Estimates are based on $N = 150000$ samples, chosen for high accuracy in estimating the true variance of the protocol. In practice, a significantly smaller sample size suffices for recovery, as we provide evidence for.}
    \label{fig:VarScaling}
    \vspace{-0.3cm}
\end{figure}

\emph{Conclusions.} In this work, we have presented an experimentally feasible method for measuring  $2$- and $4$-point fermionic correlation functions 
by innovating what could be called an
\emph{analog classical shadows protocol}. Unlike other protocols, ours requires only translationally invariant nearest-neighbour hopping quenches, with respect to a \emph{single} Hamiltonian, and is suited for fermionic analog quantum simulators. 
Such read-out methods are of great importance to \emph{variational quantum algorithms}~\cite{Variational}, for assessing \emph{quantum simulations}~\cite{Trotzky}, as well as for constructing improved \emph{entanglement witnesses}~\cite{quant-ph/0607167, Audenaert06}.
One of the main open directions left in our work is its generalisation to higher-dimensional systems. A straightforward generalisation is possible using Hamiltonians with complex-valued couplings, but more thought is required to determine whether this can also be achieved with only real couplings.
In conclusion, we believe that protocols of this type could become cornerstone methods in practical implementations of quantum simulation.

{\it Acknowlegdements.} We thank Steven Thompson for help with our numerical calculations and Lennart Bittel for useful discussions. We also thank the HPC services of ZEDAT, Freie Universität Berlin~\cite{bennett_curta_2020} and of the Physics Department, Freie Universität Berlin, for computing time. This work has been supported by the BMBF (primarily FermiQP, but also MuniQC-Atoms and Hybrid), the Munich Quantum Valley (K-8),
 the Quantum Flagship (PasQuans2), the Einstein Foundation (Einstein Research Unit on 
 Quantum Devices), and the DFG (CRC 183, FOR 2724).
 \newpage 
\bibliographystyle{apsrev4-2}

\bibliography{references}

\onecolumngrid
\clearpage
\appendix

\section{Preliminaries on fermions}

We consider a one-dimensional spinless fermionic system with periodic boundary conditions, composed of $L$ modes with creation and annihilation operators $a_{j}^{\dagger}$ and $a_{j}$, $j\in\{0,\dots , L-1\}$, satisfying the canonical anti-commutation relations \cite{tagliacozzo} 
\begin{align}
    \{a_{k}, a_{l}\}=0 \quad \{a_{k}, a_{l}^{\dagger}\}=\delta_{k,l} .
\end{align}
The Hilbert space $\mathcal{H}\cong{(\mathbb{C}^2)}^{\otimes L}$ is spanned by the orthonormal Fock basis vectors
\begin{equation}
    \ket{\boldsymbol{n}}:= 
    \ket{n_{0}, \ldots, n_{L-1}}:=\left(a_{0}^{\dagger}\right)^{n_{0}} \ldots\left(a_{L-1}^{\dagger}\right)^{n_{L-1}}\ket{\boldsymbol{0}} \,,
\end{equation}
where the vacuum state vector $\ket{\boldsymbol{0}}$ satisfies $a_{j}\ket{\boldsymbol{0}}=0$ for all $j$. The $n_{j} \in\{0,1\}$ are the eigenvalues of the respective occupation number operators $a_{j}^{\dagger} a_{j}$, which we collect into the vector $\boldsymbol{n} = (n_{0}, \ldots, n_{L-1})$.
The subspace spanned by all Fock basis vectors with $\sum^{L-1}_{i=0} n_{i}=l$ is called the $l$-particle subspace $\mathcal{H}_l$, and $\mathcal{H}=\bigoplus^{L-1}_{l=0} \mathcal{H}_l$. 
\begin{definition}[Free, particle number preserving Hamiltonians]\ 
A Hamiltonian $H$ acting on $L$ fermionic modes is a free, particle number preserving Hamiltonian if it is of the form
\begin{align}
    H=\sum^{L-1}_{i,j=0} h_{i,j} a^\dagger_i a_j,
    \label{freeHapp}
\end{align}
where $h\in \mathbb{C}^{L\times L}$ is Hermitian.

\end{definition}
Simulating a fermionic system that evolves according to such a Hamiltonian can be done efficiently in the number of modes \cite{Jozsa_2008}.
In particular, in the Heisenberg picture, the fermionic creation and annihilation operators transform according to an $L\times L$ unitary acting on the "mode space" $V:=\mathbb{C}^L$ (and not an exponentially large matrix). This insight is formalized in the following proposition:
\begin{proposition}[Fermionic operators in the Heisenberg picture] 
\label{Prop:freeFerm}
Let $H$ be a \emph{free, particle number preserving Hamiltonian} \eqref{freeHapp} and $U:=e^{iHt}$ the corresponding unitary evolution on $\mathcal{H}$, then
\begin{align}
    U a_i U^\dagger=\sum_{j=0}^{L-1} u_{i,j}a_j, \quad \text{or equivalently}\quad U^\dagger a_i U=\sum_{j=0}^{L-1} \overline{u}_{j,i}a_j,
    \label{evolutionfreeHapp}
\end{align}
where $u:=e^{-iht}$ is an $L\times L$ unitary acting on the mode space.
\end{proposition}
\begin{proof}
Let $a_i(t):=e^{iHt}a_ie^{-iHt}$ and $\boldsymbol{a}(t):=\left(a_1(t),\dots,a_L(t)\right)$.
We have
\begin{align}
    \left(\frac{d}{dt}\boldsymbol{a}(t)\right)_i&=\frac{d}{dt}(a_i(t))=i\left[H,a_i(t)\right]=ie^{iHt}\left[H,a_i\right]e^{-iHt}\\&=i\sum^{L-1}_{l,j=0}h_{l,j}e^{iHt}\left[a^\dagger_l a_j,a_i\right]e^{-iHt}=-i\sum^{L-1}_{j=0}h_{i,j}a_j(t)=-i\left(h \boldsymbol{a}(t)\right)_i \nonumber
\end{align}
where in the last line we have used the identity $[AB,C]=A\{B,C\}-\{A,C\}B$.
We conclude the proof by noting that the first order linear (matrix) differential equation $\frac{d}{dt}\boldsymbol{a}(t)=-i h \boldsymbol{a}(t)$ admits the unique solution $\boldsymbol{a}(t)=\exp\left(-i h t\right)\boldsymbol{a}(0)$.
\end{proof}

 For experimental feasibility, we will only consider evolution according to free, particle number preserving Hamiltonians which are also \emph{translationally invariant}. We now make precise what is meant by translational invariance, and subsequently fully characterize the corresponding unitaries on the mode space.

\begin{definition}[Translationally invariant operators]
A Hilbert space operator $O$ is \emph{translationally invariant} if it commutes with the unitary shift operator $T$. The shift operator is defined by $T\ket{n_0,n_1,\dots,n_{L-1}}=(-1)^{n_{L-1}\sum_{j=0}^{L-2} n_j}\ket{n_{L-1},n_0,\dots,n_{L-2}}$ and satisfies  $Ta_jT^\dagger=a_{j+1}$, where indices are considered modulo $L$.
\end{definition}

Due to linear independence of the operators $a_i^{\dagger}a_j$, 
a translationally invariant, free Hamiltonian $H=\sum_{i,j} h_{i,j} a^\dagger_i a_j$ must satisfy $h_{i,j} = h_{i+k,j+k}$ for all $k\in\mathbb{Z}$. We call such a matrix $h$ circulant (see Definition \ref{def:circulant}).
Unitaries on the mode space corresponding to a free, translationally invariant Hamiltonian inherit the circulant property:
\begin{proposition}
    Let $h\in \mathbb{C}^{L\times L}$ be circulant and Hermitian. Then $e^{-iht}$ is a circulant unitary.
\end{proposition} 
\begin{proof}
Unitarity follows from Hermicity of $h$. Further,
\begin{equation}
    \left(e^{-iht}\right)_{i+k,j+k}=\sum^\infty_{l=0}\frac{(-it)^l \left(h^l\right)_{i+k,j+k}}{l!}=\sum^\infty_{l=0}\frac{(-it)^l \left(h^l\right)_{i,j}}{l!}=\left(e^{-iht}\right)_{i,j},
\end{equation}
since if $h$ is circulant, then so is $h^l$ (because $(h^l)_{i+k,j+k}=\sum_{s_1,\dots,s_l}h_{i+k,s_1+k} h_{s_1+k,s_2+k}\cdots h_{s_{l}+k,j+k} =(h^l)_{i,j}$). 
\end{proof}

We end this section with two short Definitions, hinting at the quantities we ultimately want to estimate.

\begin{definition}[Correlation operators]
Let $k\in \mathbb{N}$.
An operator $O$ is a $2k$-point correlation operator if it is of the form $O=a_{i_1}^\dagger\dots a_{i_k}^\dagger a_{j_1}\dots a_{j_k}$, where $\left(i_1,\dots,i_k, j_1, \dots, j_k \right)\in\{0,\dots, L-1\}^{2k}$.
\label{def_CorrOp}
\end{definition}

\begin{definition}[$2k$-point subspace]
Let $k\in \mathbb{N}$. We define $W_{2k}$ to be the subspace spanned by all the $2k$-point correlation operators.
\label{def_CorrOpSubspace}
\end{definition}
Note that a basis for $W_{2k}$ is given by 
\begin{equation}
    \Big\{a_{i_1}^\dagger\dots a_{i_k}^\dagger a_{j_1}\dots a_{j_k} \, : \, 0 \le i_1 < \dots < i_k\le L-1 \text{ and } 0\le j_1< \dots < j_k\le L-1 \Big\}.
\end{equation}
We have 
\begin{align}
    \mathrm{dim}\left(W_{2k}\right)=\binom{L}{k}^2\le \frac{L^{2k}}{(k!)^2}.
\end{align}

\section{Properties of circulant matrices}
In the previous section we introduced the concept of circulant matrices. In this section we discuss some of their useful properties. We denote with bold vectors the elements of the mode space $V=\mathbb{C}^L$ and in particular the canonical basis $\{\boldsymbol{e}^{(i)}\}^{L-1}_{i=0}$.
\begin{definition}[Circulant matrices]
    \label{def:circulant}
    A matrix $M\in \mathbb{C}^{L\times L}$ is circulant if it is shift-invariant, \ie if it satisfies
\begin{align}
         M_{i,j}=M_{i+k,j+k}, 
\end{align}
for all $i,j,k \in \{0,\dots,L-1\}$. Indices should be considered modulo $L$.
As such, a circulant matrix is specified by a single row (column), as it is of the form
 \begin{equation}
\left(\begin{array}{ccccc}
M_{0} & M_{1} & M_{2} & \ldots & M_{L-1} \\
M_{L-1} & M_{0} & M_{1} & \ldots & M_{L-2} \\
& & \ddots & & \\
M_{2} & \ldots & M_{L-1} & M_{0} & M_{1} \\
M_{1} & \ldots & M_{L-2} & M_{L-1} & M_{0}
\end{array}\right) ,
 \end{equation}
where we have defined $M_{i}:=M_{0,i}$\, for $i\in \{0,\dots,L-1\}$. 
\end{definition}

We denote the set of $L\times L$ circulant matrices by $C(L)$. Two subsets of $C(L)$ will be of particular importance to us. They are defined as follows:

\begin{definition}[Circulant and symmetric circulant unitary subgroups]\ 
The set of circulant unitary matrices is defined as \begin{equation}
CU(L)=C(L)\cap U(L),\end{equation} where $U(L)$ is the full unitary group.
The set of symmetric circulant unitary matrices is defined as \begin{equation}
CU_{\rm Sym}(L)=CU(L)\cap C_{\rm Sym},    
\end{equation} where $C_{\rm Sym} = \{h\in \mathbb{C}^{L\times L}\, \vert \, h_{i,j} = h_{j,i} \, \forall i,j\}$ is the group of symmetric matrices.
\end{definition}

\begin{proposition}[Properties of circulant matrices] \ 
    \label{circulant_prop}
    \begin{enumerate}
        \item A matrix $M\in \mathbb{C}^{L\times L}$ is circulant if and only if it is diagonal in the orthonormal Fourier basis $\boldsymbol{v}^{(k)}=1/\sqrt{L}\sum_j \omega^{kj} \boldsymbol{e}^{(j)}$, where $k\in\{0,\dots,L-1\}$, and $\omega=\exp\left(\frac{2\pi i }{L}\right)$. The corresponding eigenvalues are $\lambda_k=\sum_{j} M_j \omega^{kj}$. Therefore, a convenient way to represent $M$ is its spectral decomposition
        \begin{equation}
            M=\sum^{L-1}_{k=0} \lambda_k \boldsymbol{v}^{(k)}\boldsymbol{v}^{(k)\dagger}\,. \label{eq:spec_decomp_circ}
        \end{equation}
        
        \item The set of circulant matrices $C(L)$ forms an Abelian group.

        \item Matrices $U \in CU(L)$ are of the form \eqref{eq:spec_decomp_circ}, with $\lambda_k=e^{i\varphi_k}$ for $\varphi_k \in [0,2\pi)$. For matrices of $CU_{\rm Sym}(L)$ we further have $\varphi_{L-k}=\varphi_{k}$.
        
        \item The sets of circulant unitary matrices $CU(L)$ and symmetric circulant unitary matrices $CU_{\rm Sym}(L)$ are non-trivial subgroups of $U(L)$. They are compact with respect to the Hilbert-Schmidt metric.
    \end{enumerate}
\end{proposition}
\begin{proof}

\noindent
\begin{enumerate}
\item For a circulant $M$ we have
\begin{equation}
M \boldsymbol{v}^{(k)}=\sum^{L-1}_{i,j=0}M_{i,j}\left(\boldsymbol{v}^{(k)}\right)_j \boldsymbol{e}^{(i)}=\frac{1}{\sqrt{L}}\sum^{L-1}_{i,j=0}M_{j-i}\omega^{kj} \boldsymbol{e}^{(i)}=\frac{1}{\sqrt{L}}\sum^{L-1}_{i=0}\left(\sum^{L-1}_{j=0}M_{j}\omega^{kj}\right) \omega^{ki} \boldsymbol{e}^{(i)}=\lambda_k \boldsymbol{v}^{(k)}\,,\end{equation}
where indices are to be understood modulo $L$ and $M_i$ is defined like in Definition \ref{def:circulant}.
Conversely, if $M$ is of the form \eqref{eq:spec_decomp_circ}, we have
\begin{align}
M_{i+k,j+k}&=\sum^{L-1}_{s=0} \lambda_s \left(\boldsymbol{v}^{(s)}\boldsymbol{v}^{(s)\dagger}\right)_{i+k,j+k}=\frac{1}{L}\sum^{L-1}_{s=0}\lambda_s \omega^{s(i+k)}\omega^{-s(j+k)}\\
&=\frac{1}{L}\sum^{L-1}_{s=0}\lambda_s\omega^{si}\omega^{-sj}=\sum^{L-1}_{s=0} \lambda_s \left(\boldsymbol{v}^{(s)}\boldsymbol{v}^{(s)\dagger}\right)_{i,j}=M_{i,j}
\end{align}

\item This is immediate, given the fact that circulant matrices are diagonal in the same (orthonormal) basis.

\item Since $U$ is unitary, eigenvalues must have absolute value 1, \ie $\lambda_k=e^{i \varphi_k}$. Moreover, transposing $U$ yields $U^T = \sum_k e^{i \varphi_k}\boldsymbol{v}^{(-k)}\boldsymbol{v}^{(-k)\dagger}$, since $\overline{\boldsymbol{v}}^{(k)}=\boldsymbol{v}^{(-k)}$. As such, $U^T=U$ is equivalent to $e^{i \varphi_{L-k}}=e^{i \varphi_k}$.

\item It is easy to check that the conditions on eigenvalues in point 3, which equivalently define $CU(L)$ and $CU_{\rm Sym}(L)$, are closed under matrix multiplication. 
Given the metric space $(X,d)$ consisting of $L\times L$ matrices with complex entries, equipped with the Hilbert-Schmidt metric, compactness of $CU(L)$ follows from the fact that it is the intersection of $U(L)$ with the pre-image of the compact (discrete) set $\{\mathbf{0}\}$, under the continuous map
\begin{equation}
f:X\to X,\, A\mapsto (TAT^\dagger - A),\end{equation}
where $T$ is the shift operator mapping $\boldsymbol{e}^{(i)}$ to $\boldsymbol{e}^{(i+1)}$.
Similarly, $CU_{\rm Sym}(L)$ is compact because it is the intersection of $CU(L)$ and the pre-image of $\{\mathbf{0}\}$ under the continuous map
\begin{equation}
g:X\to X,\, A\mapsto (A^{T}-A).\end{equation}

\end{enumerate}
\end{proof}

\section{Integrating over compact groups and the commutant theorem}

The Haar measure on a compact group $G$ formalizes the notion of a uniform probability measure on $G$, notably also for non-discrete $G$. In all protocols we propose, an experimentalist samples a uniformly random element (at least approximately, see Appendix~\ref{sec:nn-hamiltonians}) from a suitable $G$ and proceeds accordingly.  Thus, computing expectation values of single-shot estimates entails evaluating Haar integrals.

\begin{proposition}[Haar measure]
    Every compact group $G$ has a unique Haar measure $\mu_H:\mathcal{B}\to \mathbb{R}$, defined on the Borel $\sigma$-algebra $\mathcal{B}$ of $G$, satisfying
    \begin{enumerate}
        \item right-invariance: $\mu(Sg)=\mu(S)$, for all $g\in G$ and Borel sets $S$
        \item normalization: $\int_{G} d\mu_H(g) = 1 $,
    \end{enumerate}
and some additional regularization conditions that will be unimportant for the following (see Ref.~\cite{Simon1995RepresentationsOF} for a comprehensive treatment and proof). A compact group's right-invariant Haar measure is simultaneously left-invariant, and vice versa. In this case, the choice of definition is a matter of convention.

\end{proposition}

In particular, the average of any polynomial in the unitary entries can be deduced from the so-called moment operator of sufficient order.

\begin{definition}[Moment operator \cite{kliesch_theory_2021}]
\label{def:momop}
Let $G$ be a compact group and $\{u_g\}_{g\in G}\subset U(d)$ a $d$-dimensional unitary representation (i.e., $u_g u_h = u_{gh}$). 
The $\tmom$-th moment operator (with respect to the Haar measure on $G$) is the map 
\begin{align}
\begin{split}
\Phi_{\tmom}: \mathbb{C}^{\tmom d\times \tmom d} &\to \mathbb{C}^{\tmom d\times \tmom d},\\
A  \quad  &\mapsto \int_G u_g^{\dag \otimes \tmom } \: A \: u_g^{\otimes \tmom } \; d\mu_H(g),
\end{split}
\end{align}
where integration is carried out with respect to the unique Haar measure on $G$, $\mu_H$.
\end{definition}

The invariance of the Haar measure turns out to be of great help in characterizing moment operators. In fact, we will now prove a well-known, but nonetheless elegant characterization in terms of commutants:

\begin{definition}[Commutant]  \label{def:comm}
The commutant of a subset $\mathcal{S}\subset \mathbb{C}^{d\times d}$ is the subspace \begin{equation}
\operatorname{comm}(\mathcal{S}) = \{T\in \mathbb{C}^{d \times d} \, \vert \, TS = ST \, ,\forall S\in \mathcal{S}\}.\end{equation}
\end{definition}

\begin{lemma}[Commutant theorem \cite{kliesch_theory_2021}]\label{lemma:comm_theorem} \ 
Let $G$ and $u$ be defined like in Definition \ref{def:momop}.
The $\tmom $-th moment operator is the orthogonal projector (with respect to the Hilbert-Schmidt inner product) onto $\operatorname{comm}(\{u_g^{\otimes \tmom }\}_{g\in G})$.
\begin{proof}
Due to the right-invariance of the Haar measure, it immediately follows that the image of the $\tmom $-th moment operator is contained in $\operatorname{comm}(\{u_g^{\otimes \tmom }\}_{g\in G})$, since
\begin{align}
 u_h^{\otimes \tmom }\int_G u_g^{\dag \otimes \tmom } &A u_g^{ \otimes \tmom }d\mu_H(g) \\
  \nonumber
 = & \int_G u_{g'}^{\dag \otimes \tmom }A u_{g'h}^{\otimes \tmom }\,d\mu_H(g') \\
  \nonumber
 = & \int_G u_{g'}^{\dag \otimes \tmom }A u_{g'}^{\otimes \tmom } u_{h}^{\otimes \tmom }\, d\mu_H(g') \\
  \nonumber
 = & \int_G u_{g'}^{\dag \otimes \tmom }A u_{g'}^{\otimes \tmom }d\mu_H(g') \, u_h^{\otimes \tmom },
  \nonumber
\end{align}
where $g' = gh^{-1}$. It is also apparent that all $A\in \operatorname{comm}(\{u_g^{\otimes \tmom }\}_{g\in G})\setminus\{\mathbf{0}\}$ are $+1$-eigenvectors of the moment operator. What is left to show is that the orthogonal complement with respect to the Hilbert-Schmidt inner product, $\operatorname{comm}(\{u_g^{\otimes \tmom }\}_{g\in G})^{\perp}$, is an invariant subspace. To that effect, take two arbitrary elements $B\in \operatorname{comm}(\{u_g^{\otimes \tmom }\}_{g\in G})$ and $C\in \operatorname{comm}(\{u_g^{\otimes \tmom }\}_{g\in G})^{\perp}$. Then
\begin{align}
    \operatorname{tr}\left(B^{\dag}\int_G u_g^{\dag \otimes \tmom }C u_g^{ \otimes \tmom }d\mu_H(g)\right)  = \operatorname{tr}(B^{\dag}C)= 0\,,
\end{align}
where we have used that $\operatorname{comm}(\{u_g^{\otimes \tmom }\}_{g\in G})$ is closed under taking adjoints, because of unitarity, and the cyclicity of the trace.
\end{proof}
\end{lemma}

In the spirit of sampling from such unitary ensembles and averaging single-shot estimates, we will from here on denote integration with respect to the Haar measure by the expectation value $\mathbb{E}_{u\sim\{u_g\}_{g\in G}}$.

\subsection{Commutants and moment operators of the symmetric circulant unitary subgroups}

Since it is a compact groups, $CU_{\rm Sym}(L)$ carries a unique invariant measure as explained in the previous section, allowing us to apply the \emph{commutant theorem}. In order to determine the commutant of the unitary representation $\{u^{\otimes \tmom }\}_{u\in CU_{\rm Sym}(L)}$, we begin by decomposing it into irreducibles. As $CU_{\rm Sym}(L)$ is Abelian, the irreducible representations are one-dimensional. In what follows, bold indices indicate a collection of $\tmom$ indices, such as $\boldsymbol{i}\equiv(i_1,\dots ,i_\tmom)$

\begin{proposition}[Irreducible representations of $\{u^{\otimes \tmom }\}_{u\in CU_{\rm Sym}(L)}$]
\label{prop:irreps_of_ut}\ 
    The unitary representation $\{u^{\otimes \tmom }\}_{u \in CU_{\rm Sym}(L)}$ on $V^{\otimes t }=(\mathbb{C}^L)^{\otimes t}$ decomposes into a direct sum of one-dimensional irreducible representations. The representation spaces, which we denote as $V_{\boldsymbol{i} }$, are spanned by the vectors $\boldsymbol{v}^{(i_1)}\otimes\cdots\otimes \boldsymbol{v}^{(i_\tmom )}$, where $\boldsymbol{v}^{(i)}$ are the Fourier vectors defined in Proposition \ref{circulant_prop}. Two such irreps labeled by $\boldsymbol{i}$ and $\boldsymbol{j}$ are equivalent iff $\boldsymbol{j}=\pi(s_1 i_1,\dots,s_\tmom  i_\tmom )$, for some choice of $\pi$ and s, where $\pi\in S_\tmom $ is a permutation and $s$ is a collection of signs $s\in \{-1,+1\}^\tmom  $ (multiplication of indices by a sign should be understood modulo $L$). In such case we write $\boldsymbol{i}\sim \boldsymbol{j}$
\end{proposition}
\begin{proof}
    The vectors $\boldsymbol{v}^{(i_1)}\!\otimes\!\cdots\!\otimes\! \boldsymbol{v}^{(i_\tmom )}$ are eigenvectors of all elements $u^{\otimes \tmom }$. Thus, they span a representation space which is one-dimensional and hence irreducible. More explicitly, $u^{\otimes \tmom }\, \boldsymbol{v}^{(i_1)}\!\otimes\cdots\otimes\!\boldsymbol{v}^{(i_\tmom )} = e^{i\sum_{a=1}^\tmom  \varphi_{i_a}(u)} \, \boldsymbol{v}^{(i_1)}\!\otimes\cdots\!\otimes\! \boldsymbol{v}^{(i_\tmom )}$. Here, $e^{i\varphi_{i_a}\!(u)}$ is the eigenvalue of $u$ corresponding to the eigenvector $\boldsymbol{v}^{(i_a)}$ (see Proposition \ref{circulant_prop}). Two such representation are equivalent iff $\sum_{a=1}^\tmom  \varphi_{i_a}(u)=\sum_{a=1}^\tmom  \varphi_{j_a}(u)$ for all $u$, modulo $2\pi$. Due to the symmetries of $CU_{\rm Sym}(L)$, this is for sure the case if $(j_1,\dots,j_\tmom )=\pi(s_1 i_1,\dots,s_\tmom i_\tmom )$, while otherwise it is easy to construct a $u$ such that it is not the case.
\end{proof}

Using this observation and applying Shur's lemma, one can write down the commutant of $\{ u^{\otimes \tmom} \} _{u\in CU_{\rm Sym}(L)}$:
\begin{proposition}[Commutant of $\{ u^{\otimes \tmom} \} _{u\in CU_{\rm Sym}(L)}$]\label{prop:comm_ut}
    The commutant of $\{ u^{\otimes \tmom } \} _{u\in CU_{\rm Sym}(L)}$ is spanned by the operators 
    \begin{equation}
        \left( \boldsymbol{v}^{(i_1)}\!\otimes\!\cdots\!\otimes\!\boldsymbol{v}^{(i_\tmom )}\right) {\left( \boldsymbol{v}^{(j_1)}\!\otimes\!\cdots\!\otimes\!\boldsymbol{v}^{(j_\tmom )}\right)}^\dag \,\mbox{,      for all    } \boldsymbol{i}, \boldsymbol{j} \mbox{  such that  }\boldsymbol{i}\sim \boldsymbol{j}\,,
    \end{equation}
    where the equivalence $\boldsymbol{i}\sim \boldsymbol{j}$ is defined as in Proposition \ref{prop:irreps_of_ut}. Notice that all distinct operators of this form give rise to an orthonormal basis w.r.t. the Hilbert-Schmidt inner product. We indicate this set of basis operators as $\mathcal{C}_{\tmom}$.
\end{proposition}
\begin{proof}
    We want to constrain the vector space of linear operators $A$ on $V^{\otimes \tmom} $ such that $A u^{\otimes \tmom } = u^{\otimes \tmom } A$. Decomposing $u^{\otimes \tmom }$ into the $m = L^\tmom $ irreducible representations on $V_{\boldsymbol{i}}$ introduced in Proposition \ref{prop:irreps_of_ut}, we get
    \begin{equation}
        \begin{pmatrix} A_{\boldsymbol{i}_0,\boldsymbol{i}_0} & A_{\boldsymbol{i}_0,\boldsymbol{i}_1} & \cdots & A_{\boldsymbol{i}_0,\boldsymbol{i}_m}\\ \vdots & \vdots & & \vdots\\ A_{\boldsymbol{i}_m,\boldsymbol{i}_0} & A_{\boldsymbol{i}_m,\boldsymbol{i}_1} & \cdots & A_{\boldsymbol{i}_m,\boldsymbol{i}_m} \end{pmatrix} \begin{pmatrix} \rho_{\boldsymbol{i}_0}(u) &  &\\  & \ddots & \\& & \rho_{\boldsymbol{i}_m}(u) \end{pmatrix} = \begin{pmatrix} \rho_{\boldsymbol{i}_0}(u) &  &\\  & \ddots & \\& & \rho_{\boldsymbol{i}_m}(u) \end{pmatrix} \begin{pmatrix} A_{\boldsymbol{i}_0,\boldsymbol{i}_0} & A_{\boldsymbol{i}_0,\boldsymbol{i}_1} & \cdots & A_{\boldsymbol{i}_0,\boldsymbol{i}_m}\\ \vdots & \vdots & & \vdots\\ A_{\boldsymbol{i}_m,\boldsymbol{i}_0} & A_{\boldsymbol{i}_m,\boldsymbol{i}_1} & \cdots & A_{\boldsymbol{i}_m,\boldsymbol{i}_m} \end{pmatrix}\,,
    \end{equation}
    where $\rho_{\boldsymbol{i}}(u)$ is the one-dimensional irreducible action of $u^{\otimes \tmom }$ on $V_{\boldsymbol{i}}$. Therefore, $A_{\boldsymbol{i,j}} \rho_{\boldsymbol{j}}(u)= \rho_{\boldsymbol{i}}(u) A_{\boldsymbol{i,j}}$ for every $u\in CU_{\rm Sym}(L)$. As $\rho_{\boldsymbol{i}}$ and $\rho_{\boldsymbol{j}}$ are irreps, due to Shur's lemma $A_{\boldsymbol{i,j}}$ can be non-vanishing only if $\rho_{\boldsymbol{i}}$ is equivalent to  $\rho_{\boldsymbol{j}}$, that is $\boldsymbol{i}\sim\boldsymbol{j}$. This constrains $A$ to be in the span of $\left( \boldsymbol{v}^{(i_1)}\!\otimes\!\cdots\!\otimes\! \boldsymbol{v}^{(i_\tmom )}\right) {\left( \boldsymbol{v}^{(j_1)}\!\otimes\!\cdots\!\otimes\! \boldsymbol{v}^{(j_\tmom )}\right)}^\dag$, for all $\boldsymbol{i}$ and $\boldsymbol{j}$ such that $\boldsymbol{i}\sim\boldsymbol{j})$. Conversely, each of these orthonormal operators independently commutes with $u^{\otimes \tmom }$, thus they form a basis of the commutant.
\end{proof}

Thanks to Lemma \ref{lemma:comm_theorem}, we have that the $\tmom$-th moment operator can be written as the orthogonal projector onto the commutant of $\{ u^{\otimes \tmom} \} _{u\in CU_{\rm Sym}(L)}$, that is
\begin{align}
    \Phi_\tmom(\cdot)&\equiv{\mathbb{E}}_{u \sim CU_{\rm Sym}(L)}  {u}^{\dag \otimes \tmom} \:(\,\cdot \,)\: u^{\otimes \tmom} =\nonumber \\
    &= \sum_{C\in\mathcal{C}_\tmom} \Tr [C^\dag (\cdot)] C \nonumber \\
    &=\sum_{\substack{\boldsymbol{i}, \boldsymbol{j}\\ s.t.\, \boldsymbol{i}\sim \boldsymbol{j}}} \Tr[\left( \boldsymbol{v}^{(j_1)}\!\otimes\!\cdots\!\otimes\!\boldsymbol{v}^{(j_\tmom )}\right) {\left( \boldsymbol{v}^{(i_1)}\!\otimes\!\cdots\!\otimes\!\boldsymbol{v}^{(i_\tmom )}\right)}^\dag(\cdot)] \left( \boldsymbol{v}^{(i_1)}\!\otimes\!\cdots\!\otimes\!\boldsymbol{v}^{(i_\tmom )}\right) {\left( \boldsymbol{v}^{(j_1)}\!\otimes\!\cdots\!\otimes\!\boldsymbol{v}^{(j_\tmom )}\right)}^\dag\,. \label{eq:t_moment_op}
\end{align}
Evaluating individual matrix elements of this expression, we have
\begin{align}
    \Phi_{\tmom}(\boldsymbol{k},\boldsymbol{l};\boldsymbol{p},\boldsymbol{q}) \equiv \underset{u}{\mathbb{E}} \; \overline{u}_{k_1,p_1} \cdots \overline{u}_{k_\tmom,p_\tmom} 
    \: u_{l_1, q_1} \cdots u_{l_\tmom, q_\tmom}
    = \sum_{\substack{\boldsymbol{i},\boldsymbol{j}:\\\boldsymbol{i}\sim\boldsymbol{j}}}\: 
    \overline{v}^{(i_1)}_{k_1}\cdots \overline{v}^{(i_\tmom)}_{k_\tmom}
    \: v^{(j_1)}_{l_1} \cdots v^{(j_\tmom)}_{l_\tmom}  
    \: v^{(i_1)}_{p_1}\cdots v^{(i_\tmom)}_{p_\tmom} 
    \:\overline{v}^{(j_1)}_{q_1} \cdots \overline{v}^{(j_\tmom)}_{q_\tmom} \,. \label{eq:t_moment_op_elements}
\end{align}

\subsubsection{Second moments}
Specialising to $t=2$ and $L$ odd, we have that an explicit orthonormal basis of the commutant of $\{ u^{\otimes 2} \} _{u\in CU_{\rm Sym}(L)}$ is given by the set $\mathcal{C}_2$ of operators listed below:
\begin{alignat}{2}
\label{eq:setW}
\mathcal{C}_2 =
     \Big\{& (\boldsymbol{v}^{(n)} \otimes \boldsymbol{v}^{(m)})  \, (\boldsymbol{v}^{(n)} \otimes \boldsymbol{v}^{(m)})^\dag  \hspace{4em} && \forall n,m\,,   \nonumber \\
    &(\boldsymbol{v}^{(n)} \otimes \boldsymbol{v}^{(m)})  \, (\boldsymbol{v}^{(-n)} \otimes \boldsymbol{v}^{(m)})^\dag  && \forall m,\, n\neq 0\,, \nonumber\\
    &  (\boldsymbol{v}^{(n)} \otimes \boldsymbol{v}^{(m)})  \, (\boldsymbol{v}^{(n)} \otimes \boldsymbol{v}^{(-m)})^\dag  && \forall n,\, m\neq 0\,, \nonumber \\ 
    &  (\boldsymbol{v}^{(n)} \otimes \boldsymbol{v}^{(m)})  \, (\boldsymbol{v}^{(-n)} \otimes \boldsymbol{v}^{(-m)})^\dag  && \forall n\neq 0 ,\, m\neq 0\,, \nonumber\\
    &  (\boldsymbol{v}^{(n)} \otimes \boldsymbol{v}^{(m)})  \, (\boldsymbol{v}^{(m)} \otimes \boldsymbol{v}^{(n)})^\dag  && \forall n\neq \pm m \,, \nonumber \\
    &(\boldsymbol{v}^{(n)} \otimes \boldsymbol{v}^{(m)})  \, (\boldsymbol{v}^{(m)} \otimes \boldsymbol{v}^{(-n)})^\dag  && \forall n\neq \pm m ,\, n\neq 0 \,, \nonumber \\
    &  (\boldsymbol{v}^{(n)} \otimes \boldsymbol{v}^{(m)})  \, (\boldsymbol{v}^{(-m)} \otimes \boldsymbol{v}^{(n)})^\dag  && \forall n\neq \pm m,\, m\neq 0 \,, \nonumber\\ 
    & (\boldsymbol{v}^{(n)} \otimes \boldsymbol{v}^{(m)})  \, (\boldsymbol{v}^{(-m)} \otimes \boldsymbol{v}^{(-n)})^\dag  && \forall n\neq \pm m,\, n\neq 0 ,\, m\neq 0 \hspace{2em} \Big\},
\end{alignat}
where indices should be understood as modulo $L$. 
The constraints on the indices simply ensure that no operator is counted twice. We note that the dimension of the commutant is $8L^2 -16L +9$.

In the following sections we will need some quantities related to the moment operator $\Phi_2$. In particular we will need
\begin{align}
    \sum_p \mathbb{E}_{u\sim CU_{\rm Sym}(L)} \left(u^{\dag \otimes 2} \right)_{pp,ki} \left(u^{\otimes 2}\right)_{lj,pp}&=\sum_p \Phi_2(k,i,l,j;p,p,p,p)\nonumber\\
    &=\sum_p \left[\sum_{n,m} \,\overline{v}^{(n)}_k \,\overline{v}^{(m)}_i v^{(n)}_l v^{(m)}_j \; v^{(n)}_p v^{(m)}_p \,\overline{v}^{(n)}_p \,\overline{v}^{(m)}_p \right.\nonumber \\
    &\hspace{8mm}+\sum_{n\neq 0,m} \,\overline{v}^{(n)}_k \,\overline{v}^{(m)}_i v^{(-n)}_l v^{(m)}_j \; v^{(n)}_p v^{(m)}_p \,\overline{v}^{(-n)}_p \,\overline{v}^{(m)}_p \nonumber \\
    &\hspace{8mm}+\sum_{n,m\neq 0} \,\overline{v}^{(n)}_k \,\overline{v}^{(m)}_i v^{(n)}_l v^{(-m)}_j \; v^{(n)}_p v^{(m)}_p \,\overline{v}^{(n)}_p \,\overline{v}^{(-m)}_p \nonumber \\
    &\hspace{8mm}+\sum_{n\neq 0,m\neq 0} \,\overline{v}^{(n)}_k \,\overline{v}^{(m)}_i v^{(-n)}_l v^{(-m)}_j \; v^{(n)}_p v^{(m)}_p \,\overline{v}^{(-n)}_p \,\overline{v}^{(-m)}_p \nonumber \\
    &\hspace{8mm}+\sum_{\substack{n,m \\ n\neq \pm m }} \,\overline{v}^{(n)}_k \,\overline{v}^{(m)}_i v^{(m)}_l v^{(n)}_j  \; v^{(n)}_p v^{(m)}_p \,\overline{v}^{(m)}_p \,\overline{v}^{(n)}_p \nonumber \\
    &\hspace{8mm}+\sum_{\substack{n\neq 0,m \\ n\neq \pm m }} \,\overline{v}^{(n)}_k \,\overline{v}^{(m)}_i v^{(m)}_l v^{(-n)}_j \; v^{(n)}_p v^{(m)}_p \,\overline{v}^{(m)}_p \,\overline{v}^{(-n)}_p \nonumber \\
    &\hspace{8mm}+\sum_{\substack{n,m\neq 0 \\ n\neq \pm m }} \,\overline{v}^{(n)}_k \,\overline{v}^{(m)}_i v^{(-m)}_l v^{(n)}_j \; v^{(n)}_p v^{(m)}_p \,\overline{v}^{(-m)}_p \,\overline{v}^{(n)}_p \nonumber \\
    &\hspace{8mm}+\sum_{\substack{n\neq 0,m\neq 0 \\ n\neq \pm m }} \left.\,\overline{v}^{(n)}_k \,\overline{v}^{(m)}_i v^{(-m)}_l v^{(-n)}_j \; v^{(n)}_p v^{(m)}_p \,\overline{v}^{(-m)}_p \,\overline{v}^{(-n)}_p \right]\,,
\end{align}
which we have evaluated by expanding it according to~\eqref{eq:t_moment_op_elements} and~\eqref{eq:setW}. Here, $v^{(n)}_l=\frac{1}{\sqrt{L}}\omega^{nl}$ is the $l$-th component of the vector $\boldsymbol{v}^{(n)}$. By using $\sum_p \omega^{pn}=L \delta_{n,0}$, where $\delta$ is the Kroenecker delta function on $\mathbb{Z}_L$ (\ie it constrains its arguments to be equal modulo $L$), we further have
\begin{align}
    \sum_p \mathbb{E}_{u\sim CU_{\rm Sym}(L)} &\left(u^{\dag \otimes 2} \right)_{pp,ki} \left(u^{\otimes 2}\right)_{lj,pp}= \nonumber\\[2mm]
    =& \frac{1}{L^3}\left[ \, \sum_{n,m} \omega^{n(l-k)} \,  \omega^{m(j-i)} \right.\nonumber \\
    &\hspace{8mm}+ \sum_{n\neq 0,m} \omega^{n(-l-k)} \,  \omega^{m(j-i)} \;\delta_{2n,0} \nonumber \\
    &\hspace{8mm}+\sum_{n,m\neq 0} \omega^{n(l-k)} \,  \omega^{m(-j-i)} \delta_{2m,0} \nonumber \\
    &\hspace{8mm}+\sum_{n\neq 0,m\neq 0} \omega^{n(-l-k)} \,  \omega^{m(-j-i)} \delta_{2n+2m,0}\nonumber \\
    &\hspace{8mm}+\sum_{\substack{n,m \\ n\neq \pm m }} \omega^{n(j-k)} \,  \omega^{m(l-i)} \nonumber \\
    &\hspace{8mm}+\sum_{\substack{n\neq 0,m \\ n\neq \pm m }} \omega^{n(-j-k)} \,  \omega^{m(l-i)} \delta_{2n,0} \nonumber \\
    &\hspace{8mm}+\sum_{\substack{n,m\neq 0 \\ n\neq \pm m }} \omega^{n(j-k)} \,  \omega^{m(-l-i)} \delta_{2m,0} \nonumber \\
    &\hspace{8mm}+  \sum_{\substack{n\neq 0,m\neq 0 \\ n\neq \pm m }} \omega^{n(-j-k)} \,  \omega^{m(-l-i)} \delta_{2n+2m,0} \biggr]\\ \nonumber\\
    =& \frac{1}{L^3}\left[ \, \sum_{n,m} \omega^{n(l-k)} \,  \omega^{m(j-i)} +\sum_{n\neq 0,m\neq 0} \omega^{n(-l-k)} \,  \omega^{m(-j-i)} \delta_{2n+2m,0}+\sum_{\substack{n,m \\ n\neq \pm m }} \omega^{n(j-k)} \,  \omega^{m(l-i)} \right] \label{eq:2moment_CUs_intermediate}\\ \nonumber \\
    =& \frac{1}{L^3}\left[ \, \sum_{n,m} \omega^{n(l-k)} \,  \omega^{m(j-i)} +\sum_{n\neq 0} \omega^{n(i+j-l-k)}+\sum_{\substack{n,m \\ n\neq m }} \omega^{n(j-k)} \,  \omega^{m(l-i)} -  \sum_{\substack{n \neq 0}} \omega^{n(j-k-l+i)} \right] \\ \nonumber \\
    =& \frac{1}{L^3}\left[ \, \sum_{n,m} \omega^{n(l-k)} \,  \omega^{m(j-i)} +\sum_{n,m } \omega^{n(j-k)} \,  \omega^{m(l-i)} -\sum_{n} \omega^{n(j-k+l-i)}  \right] \\ \nonumber \\
    =&\frac{1}{L} \left[ \delta_{l,k} \delta_{i,j} + \delta_{j,k} \delta_{l,i} - \frac{1}{L} \,\delta_{j-k,i-l} \right] \,,
    \label{eq:matrixelementmoment}   
\end{align}
In \eqref{eq:2moment_CUs_intermediate}, we have used the assumption that $L$ is odd (which implies that $\delta_{2n,0}$ is non-vanishing only for $n=0$) to eliminate some terms.

\subsubsection{Fourth moments}
\label{sec:fourthmomentop}

The commutant of $\{u^{\otimes 4}\}_{u\in CU_{\rm Sym}(L)}$ is given by the span of the operators $\left( \boldsymbol{v}^{(i_1)}\!\otimes\!\cdots\!\otimes\!\boldsymbol{v}^{(i_4)}\right) {\left( \boldsymbol{v}^{(j_1)}\!\otimes\!\cdots\!\otimes\!\boldsymbol{v}^{(j_4)}\right)}^\dag$, for all $\boldsymbol{i}\equiv(i_1,\dots,i_4)$ and $\boldsymbol{j}\equiv(j_1,\dots,j_4)$ such that $\boldsymbol{i}\sim\boldsymbol{j}$, like in Proposition~\ref{prop:irreps_of_ut}. Therefore, the fourth moment operator can be expressed as
\begin{align}
    \label{eqn:fourth_moment_operator}
    \mathbb{E}_{u \sim CU(L)} {u^\dag}^{\otimes 4} (\cdot ) u^{\otimes 4} &= \sum_{\substack{(\boldsymbol{i},\boldsymbol{j} )\\ \boldsymbol{i}\sim \boldsymbol{j} }} \Tr [{\left( \boldsymbol{v}^{(j_1)}\!\otimes\!\cdots\!\otimes\!\boldsymbol{v}^{(j_4)}\right)} {\left( \boldsymbol{v}^{(i_1)}\!\otimes\!\cdots\!\otimes\!\boldsymbol{v}^{(i_4)}\right)}^\dag (\cdot)] \nonumber\\
    &\hspace{60mm}\times\left( \boldsymbol{v}^{(i_1)}\!\otimes\!\cdots\!\otimes\!\boldsymbol{v}^{(i_4)}\right) {\left( \boldsymbol{v}^{(j_1)}\!\otimes\!\cdots\!\otimes\!\boldsymbol{v}^{(j_4)}\right)}^\dag \,.
\end{align}
As for second moments, we will only need expressions for specific quantities related to this fourth moment operator in subsequent derivations, namely
\begin{equation}
\label{eqn:fourth_moment_quantity}
\sum_{p_1,p_2}\Phi_4(\boldsymbol{k},\boldsymbol{m},\boldsymbol{l},\boldsymbol{n};\boldsymbol{p},\boldsymbol{p},\boldsymbol{p},\boldsymbol{p})=\sum_{p_1,p_2} \mathbb{E}_u \left({u^{\dag \otimes 4}}\right)_{p_1 p_2 p_1 p_2, k_1 k_2 m_1 m_2} \left(u^{\otimes 4} \right)_{l_1 l_2 n_1 n_2, p_1 p_2 p_1 p_2}.
\end{equation}
However, since it may be of independent value, we compute the expression for a general entry of the fourth moment operator, giving \ref{eqn:fourth_moment_quantity} as a simple corollary. Since all expressions will be anti-symmetrized later on, see Eq.~ (\ref{eq:coefficients_Mk}), we will compute entries up to an additive constant. This constant can be easily determined by evaluating $\Phi_4(\boldsymbol{0},\boldsymbol{0},\boldsymbol{0},\boldsymbol{0}; \boldsymbol{0},\boldsymbol{0}, \boldsymbol{0}, \boldsymbol{0})$. As all unitaries in the ensemble are circulant, such entry will only depend on differences of indices, i.e.,
\begin{equation}
\Phi_4(\boldsymbol{k},\boldsymbol{l},\boldsymbol{m},\boldsymbol{n};\boldsymbol{p},\boldsymbol{q},\boldsymbol{r},\boldsymbol{s}) = \Phi_4(\boldsymbol{k}\rm{-}\boldsymbol{p},\boldsymbol{l}\rm{-}\boldsymbol{q},\boldsymbol{m}\rm{-}\boldsymbol{r},\boldsymbol{n}\rm{-}\boldsymbol{s};\boldsymbol{0},\boldsymbol{0},\boldsymbol{0},\boldsymbol{0})=\Phi_4((a_1, a_2), (a_3, a_4), (b_1, b_2), (b_3, b_4), \boldsymbol{0}, \boldsymbol{0}, \boldsymbol{0}, \boldsymbol{0}),
\end{equation}
where negative indices are understood modulo $L$, and the newly-defined variables $a_i$, $b_i$ reflect the true number of degrees of freedom. From
Eq.~(\ref{eqn:fourth_moment_operator}), it follows that
\begin{align}
    & \Phi_4(\boldsymbol{k},\boldsymbol{l},\boldsymbol{m},\boldsymbol{n};\boldsymbol{p},\boldsymbol{q},\boldsymbol{r},\boldsymbol{s}) = \frac{1}{L^8} \sum_{\substack{(\boldsymbol{i},\boldsymbol{j} )\\ \boldsymbol{i}\sim \boldsymbol{j} }} \omega^{-\boldsymbol{i}\cdot \boldsymbol{a}+\boldsymbol{j}\cdot \boldsymbol{b}}.
\end{align}

To proceed, one needs to analyse all distinct choices of $\boldsymbol{i}$ and $\boldsymbol{j}$ such that $\boldsymbol{j} = \pi(s_1i_1,s_2i_2,s_3i_3,s_4i_4)$, for an arbitrary permutation $\pi\in S_4$ and collection of signs $\boldsymbol{s}\in\{\pm 1\}^4$. Special care has to be taken to not overcount any tuple $(\boldsymbol{i},\boldsymbol{j})$, since for a fixed $\boldsymbol{i}$ several choices of $\pi$, $\boldsymbol{s}$ may give rise to the same $\boldsymbol{j}$. For instance, a string $\boldsymbol{i}$ may contain repeated indices, in which case the total number of permutations acting non-trivially on it is reduced - adding signs to the mix may further decrease the number of such permutations. In addition, zero entries are of course unaffected by any sign they are multiplied by. To conveniently count all unique tuples, we employ the following strategy:

We may associate any $\boldsymbol{i}\in \{0,\dots, L-1\}^4$ with the corresponding sequence of "positive" indices $\Tilde{\boldsymbol{i}}\in \{0,\dots, \lfloor \frac{L}{2}\rfloor\}$. One obtains $\Tilde{\boldsymbol{i}}$ from $\boldsymbol{i}$ by leaving all indices $0\leq \boldsymbol{i}_k \leq \lfloor \frac{L}{2} \rfloor$ unchanged, and mapping any $\boldsymbol{i}_k > \lfloor \frac{L}{2} \rfloor$ to $-\boldsymbol{i}_k \equiv L-\boldsymbol{i}_k$. The resulting sequence $\Tilde{\boldsymbol{i}}$ will have a multiplicity pattern given by some partition $\boldsymbol{\lambda}\vdash4$. A partition of a positive integer $n$ is an ordered sequence of integers for which the sum of all elements is $n$.  For example, the partition $(3,1)$ captures $\Tilde{\boldsymbol{i}}$ that have two distinct indices, one of which appears with multiplicity $3$. By acting with a suitable permutation on $\Tilde{\boldsymbol{i}}$, we may order the sequence of indices according to $\lambda$ (different indices that appear with the same multiplicity are put in some arbitrary, but fixed order). As such, it is clear that every $\boldsymbol{i}$ may be obtained from an ordered sequence $\boldsymbol{i}^*$, corresponding to a $\boldsymbol{\lambda}\vdash 4$, by acting on it with suitable sign and permutation transformations. By extension, any tuple $(\boldsymbol{i},\boldsymbol{j})$ for which $\boldsymbol{i} \sim \boldsymbol{j}$ may be obtained from $(\boldsymbol{i}^*,\boldsymbol{i}^*)$ in this way, if we allow for independent signs and permutations. This insight provides us with a convenient way of addressing the problem of overcounting, and allows us to write

\begin{equation}
\label{eqn:monster}
     \sum_{\substack{(\boldsymbol{i},\boldsymbol{j} )\\ \boldsymbol{i}\sim \boldsymbol{j} }} \omega^{-\boldsymbol{i}\cdot \boldsymbol{a}+\boldsymbol{j}\cdot \boldsymbol{b}}=\sum_{\boldsymbol{\lambda}\vdash 4}\frac{1}{m(\boldsymbol{\lambda})}\sum_{\substack{i_1\neq \dots \neq i_{\abs{\boldsymbol{\lambda}}}, \\ 0\leq i_p \leq \lfloor \frac{L}{2}\rfloor}} \sum_{\sigma, \sigma' \in \{\pm 1\}^4}\frac{1}{\left( \boldsymbol{\lambda}!\right)^2}\sum_{\pi, \pi'\in S_4} \omega^{\boldsymbol{i}^{*}\cdot\left(\sigma\pi(\boldsymbol{a})+\sigma'\pi'(\boldsymbol{b})\right)} \prod_{k=1}^{\abs{\boldsymbol{\lambda}}}\left(1-\delta_{i_k,0}\frac{4^{\boldsymbol{\lambda}_k}-1}{4^{\boldsymbol{\lambda}_k}}\right),
\end{equation}
where $\boldsymbol{\lambda}!=\prod_{k}\boldsymbol{\lambda}_k!$,\, and $\boldsymbol{i}^{*}$ is the ordered sequence of integers belonging to the partition $\boldsymbol{\lambda}$ (subscript omitted for readability) and with corresponding entries $i_p$, as detailed in the paragraph above. The function $m(\boldsymbol{\lambda})$ is defined as
\begin{equation}
m(\boldsymbol{\lambda}) =   
    \begin{cases}
    4!, & \rm if \; \boldsymbol{\lambda} = (1,1,1,1) ,\\
    2, & \rm if \; \boldsymbol{\lambda} = (2,1,1), (2,2) ,\\
    1 , & \rm otherwise    .    
    \end{cases}
\end{equation}
Instead of summing over all valid "positive", ordered sequences $\boldsymbol{i}^{*}$ belonging to $\boldsymbol{\lambda}$, we perform a summation over distinct entries $i_p$ of $\boldsymbol{i}^{*}$. However, relabelling sub-blocks of equal multiplicity leaves $\boldsymbol{i}^{*}$ invariant and results in overcounting. Hence, we divide by $m(\lambda)$, which is defined to account for exactly this. Furthermore, we divide by $4$ for every entry of $\boldsymbol{i}^{*}$ that is zero, since both $\boldsymbol{i}$ and $\boldsymbol{j}$ will have an entry that is invariant under assignment of signs . This is taken care of by the product on the very right-hand side of Eq.~(\ref{eqn:monster}) -- each term acts non-trivially only if $i_k = 0$. Lastly, the only remaining source of overcounting is due to permutations $\pi, \pi'$ acting trivially on $\boldsymbol{i}^{*}$ by swapping duplicate entries. The number of trivial permutations depends on the multiplicity pattern of $\boldsymbol{i}^{*}$, and is given by $(\boldsymbol{\lambda})!$\,.

Equation (\ref{eqn:monster}) can be re-written as a linear combination of products of delta functions on $\mathbb{Z}_L$. To do so, we write the summation over distinct indices in terms of independent summations, where individual indices may be identified, for example $\sum_{i_1 \neq i_2} = \sum_{i_1, i_2} - \sum_{i_1 = i_2}$. Again, one has to pay special attention to not overcount any sequence of indices. 
Before we may replace sums of powers of $\omega$ with delta functions, we extend all limits of summation from $\lfloor \frac{L}{2} \rfloor$ to $L-1$. Since every term involves a summation over all assignments of signs $\sigma, \sigma'$, mapping $i_p$ to $-i_p$ is a symmetry transformation. Moreover, all indices in $\{0,\dots, L-1\}$ may be reached by starting from a "positive" index $0\leq i_p\leq \lfloor\frac{L}{2}\rfloor$ and assigning an appropriate sign to it (for indices that are zero this assignment will not be unique). Making use of these symmetry properties, we extend limits of summation by averaging over all possible assignments of signs to the $i_p$, taking special care of indices that are zero. For example, we may write 
\begin{equation}    
\sum_{\substack{i_1, i_2 \\ 0\leq i_p \leq \lfloor \frac{L}{2}\rfloor}}\left(\cdot\right) = \frac{1}{4}\sum_{\substack{i_1, i_2 \\ 0\leq i_p \leq L-1}} \left(\cdot\right)\left(1+\delta_{i_1,0}\right)\left(1+\delta_{i_2,0}\right).
\end{equation}
After some combinatorial effort one finds, up to an additive constant,

\begin{align}
\label{eqn:beast}
& \sum_{\substack{(\boldsymbol{i},\boldsymbol{j} )\\ \boldsymbol{i}\sim \boldsymbol{j} }} \omega^{-\boldsymbol{i}\cdot \boldsymbol{a}+\boldsymbol{j}\cdot \boldsymbol{b}} \equiv \sum_{\sigma, \sigma' \in \{\pm 1\}^4} \left[\frac{L^4}{2^4}\sum_{A^{(1,1,1,1)}}\delta(A^{(1,1,1,1)}_1) \delta(A^{(1,1,1,1)}_2) \delta(A^{(1,1,1,1)}_3) \delta(A^{(1,1,1,1)}_4) \right. \nonumber \\
&\hspace{50mm} -\frac{L^3}{2^5}\sum_{A^{(1,1,1)}}\delta(A^{(1,1,1)}_1) \delta(A^{(1,1,1)}_2) \delta(A^{(1,1,1)}_3) \nonumber \\
&\hspace{50mm} -\frac{L^3}{2^3} \sum_{A^{(2,1,1)}}\delta(A^{(2,1,1)}_1)\delta(A^{(2,1,1)}_2)\delta(A^{(2,1,1)}_3) \nonumber \\
&\hspace{50mm} +\frac{L^2}{2^4}\sum_{A^{(2,1)}}\delta(A^{(2,1)}_1)\delta(A^{(2,1)}_2) + \frac{9 L^2}{2^6}\sum_{A^{(1,1)}}\delta(A^{(1,1)}_1)\delta(A^{(1,1)}_2) \nonumber \\
&\hspace{50mm}+\frac{L^2}{2^2}\sum_{A^{(2,2)}}\delta(A^{(2,2)}_1)\delta(A^{(2,2)}_2) - \frac{9L^2}{2^5}\sum_{A^{(2)}}\delta(A^{(2)}_1) \nonumber \\
& \hspace{50mm} + L^2\sum_{A^{(3,1)}}\delta(A^{(3,1)}_1) \delta(A^{(3,1)}_2) -\frac{L}{2} \sum_{A^{(3)}}\delta(A^{(3)}_1) \nonumber \\
& \hspace{50mm} \left. -\frac{193L}{2^7}\sum_{A^{(1)}}\delta(A^{(1)}_1)  + L\sum_{A^{(4)}}\delta(A^{(4)}_1)\hspace{40mm}\right],
\end{align}
where $\delta$ is the Kroenecker delta function that is non-vanishing only if its argument is equal to $0$ modulo $L$. The sums over $A^{(p,q,\dots)}$ represent sums over all the possible ways of doing the following operation: consider the two sets of indices $\{a_1, a_2, a_3, a_4\}$ and $\{b_1, b_2, b_3, b_4\}$ and then consider all the possible ways of grouping them in such a way that the first group contains $p$ indices from the first set and $p$ from the second set, the second group contains $q$ indices from the first set and $q$ from the second set, and so on. The sums of the indices in each group are the components of $A^{(p,q,\dots)}_i$. For instance, one possible valid choice for $A^{(2,1,1)}$ would be
\begin{align}
    A^{(2,1,1)}_1&=a_1+a_4+b_1+b_2,  \\
    A^{(2,1,1)}_2&=a_2+b_4, \\
    A^{(2,1,1)}_3&=a_3+b_3 \,.
\end{align}
The total number of terms in (\ref{eqn:beast}) which arises from these sums over all possible choices of $A^{(p,q,\dots)}$ is a fixed combinatorial quantity that does not depend on $L$ and is of order $\sim 100$. Summation with respect to the variables $\sigma, \sigma'$ indicates that one must include all possible assignments of signs to the sets $\{a_1,a_2,a_3,a_4\}$ and $\{b_1,b_2,b_3,b_4\}$.

\section{Constructing unbiased estimators}
In this section we will review the classical shadows formalism \cite{huang_predicting_2020} and tailor it to the fermionic case we are considering.
The classical shadows formalism allows for estimating expectation values $\Tr(\rho O_i)$ of some operators $\{O_i\}_i$, with respect to a given quantum state $\rho$.
One round of the estimation procedure consists of drawing a random unitary $U$ from a suitable ensemble $\mathcal{U}$ (mostly uniformly from some unitary subgroup), evolving $\rho$ accordingly, and measuring the system. In our case, the $O_i$ are correlation operators (see Definition \ref{def_CorrOp}), evolution is restricted to free, particle number preserving, translationally invariant, and symmetric mode space transformations, and measurements are performed with respect to the occupation number basis.

The outcome $(U,\boldsymbol{n})$ of each round, with occupation numbers $\boldsymbol{n}$, is stored in a classical memory. The memory to store the outcome of a single shot must only grow linearly in the number of modes.
After repeating this prepare-and-measurement cycle, a suitable number of times, the expectation values $\Tr(\rho O_i)$ are estimated in post-processing as we explain below.
Consider the so-called 
\emph{measurement channel} \cite{huang_predicting_2020}
\begin{align}
    \mathcal{M}(\cdot):=\sum_{\boldsymbol{n}\in \{0,1\}^L}\underset{U \sim \mathcal{U} }{\mathbb{E}}\,\bra{\boldsymbol{n}}U (\cdot) U^\dagger\ket{\boldsymbol{n}} U^\dagger\ketbra{\boldsymbol{n}}{\boldsymbol{n}}U .
\end{align}
One can verify that this linear operator is self-adjoint with respect to the Hilbert-Schmidt inner product (i.e., 
$\Tr\left[\mathcal{M}(A)^\dag B\right]=\Tr\left[A^\dag \mathcal{M}(B)\right]$ for all operators $A$ and $B$).
For our choices of unitary ensemble, $\mathcal{M}$ is not invertible. However, we may consider its pseudo-inverse $\mathcal{M^+}$. The Moore-Penrose pseudo-inverse $\mathcal{M^+}$ of a self-adjoint channel $\mathcal{M}$ acts like an inverse on the image of $\mathcal{M}$ and like the zero operator on the kernel $\ker{(\mathcal{M})}$, i.e., 
\begin{equation}
    \mathcal{M}\circ \mathcal{M^+}=\mathcal{M^+}\circ \mathcal{M}=\mathcal{P}_{\Im(\mathcal{M})}\,,
\end{equation}
where $\mathcal{P}_{\Im(\mathcal{M})}$ is the orthogonal projector onto the image $\Im(\mathcal{M})$, which is the orthogonal complement of the kernel, as $\mathcal{M}$ is self adjoint.
Note that in 
Ref.~\cite{huang_predicting_2020}, all measurement channels are invertible, i.e., $\mathcal{P}_{\Im(\mathcal{M})}=\mathcal{I}$, the identity channel. 

Similarly to the steps taken in Ref.~\cite{huang_predicting_2020}, we use the measurement channel $\mathcal{M}$ to construct unbiased estimators for correlation functions of interest. Importantly, note that for a non-invertibile $\mathcal{M}$, only operators in the image of $\mathcal{M}$ can be recovered. Indeed, the part orthogonal to $\Im(\mathcal{M})$ gets projected out of the expectation value:

\begin{proposition}[Constructing unbiased estimators]
\label{prop:estimator}
    Let $\rho$ be the density matrix of a given quantum state and $O$ an operator whose expectation value with respect to $\rho$ we wish to estimate. Sampling $U$ uniformly from $\mathcal{U}$, performing a computational basis measurement with outcome $\boldsymbol{n}$, and computing the single-shot estimate $X_{U,\boldsymbol{n}} := \braket{\boldsymbol{n}| U \mathcal{M}^+(O) U^\dag |\boldsymbol{n}}$ produces the desired quantity $\Tr\left[\mathcal{P}_{\Im(\mathcal{M})}(O) \rho \right]$ in expectation. That is, we have
    \begin{equation} 
        \underset{U\sim\mathcal{U},\boldsymbol{n}}{\mathbb{E}} X_{U,\boldsymbol{n}} = \Tr\left[\mathcal{P}_{\Im(\mathcal{M})}(O) \rho \right] \,,
    \end{equation}
\end{proposition}
\begin{proof}
We have
    \begin{align}
    \nonumber
    \underset{U\sim \mathcal{U},\boldsymbol{n}}{\mathbb{E}} X_{U,\boldsymbol{n}} &= \sum_{\boldsymbol{n} \in \{0,1\}^L}\underset{U \sim \mathcal{U} }{\mathbb{E}}\,\braket{\boldsymbol{n}| U \rho U^\dagger |\boldsymbol{n}} \,  \braket{\boldsymbol{n}| U \mathcal{M}^+(O) U^\dag |\boldsymbol{n}}\\[0.8em]
    \nonumber 
    &=\Tr \left( \sum_{\boldsymbol{n} \in \{0,1\}^L}\underset{U \sim \mathcal{U} }{\mathbb{E}}   \,  \braket{\boldsymbol{n}| U \mathcal{M}^+(O) U^\dag |\boldsymbol{n}} \,U^\dagger \ket{\boldsymbol{n}}\bra{\boldsymbol{n}} U \rho \right) \\[1em]
    \nonumber
    &=\Tr\left( \mathcal{M} (\mathcal{M}^+(O)) \rho \right)\\[1em]
    &=\Tr\left(\mathcal{P}_{\Im(\mathcal{M})}(O) \rho \right).
\end{align}
\end{proof}

Let us now study measurement operators for unitary subgroups $\mathcal{U}$ that contain only free, particle number preserving fermionic evolution. In this case, the measurement channel takes the following form: 
\begin{proposition}[Block diagonal form of $\mathcal{M}$ for free, particle number preserving unitaries]\label{propo:blocks_M}\ \\
    If the unitaries in the considered subgroup $\mathcal{U}$ are free and particle number preserving, then the measurement channel $\mathcal{M}$ is block diagonal with respect to the $2k$-point operator subspaces $W_{2k}$ defined in Definition~\ref{def_CorrOpSubspace}. That is, the action of $\mathcal{M}$ on a $2k$-point operator gives a linear combination of other $2k$-point operators of the same order:
    \begin{equation}
        \mathcal{M}(a^\dag_{i_1}\cdots a^\dag_{i_k} a_{j_1}\cdots a_{j_k})= \sum_{\boldsymbol{l},\boldsymbol{m}} M^{(2k)}_{\boldsymbol{i} \boldsymbol{j},\boldsymbol{l} \boldsymbol{m}} a^\dag_{l_1}\cdots \: a^\dag_{l_k} a_{m_1}\cdots a_{m_k}\,,
    \end{equation}
    where bold letters indicate a collection of $k$ indices, such as $\boldsymbol{i}\equiv(i_1,\dots ,i_k)$. The coefficient matrices $M^{(2k)}$ are given by
    \begin{align}
        M^{(2k)}_{\boldsymbol{i} \boldsymbol{j},\boldsymbol{l} \boldsymbol{m}}&= \frac{1}{(k!)^2}\sum_{\substack{\boldsymbol{p}\\
        \pi,\sigma,\lambda\in\mathcal{S}_k}} \!\mathrm{sgn}(\pi\sigma\lambda) \:\underset{U \sim \mathcal{U} }{\mathbb{E}} (u^\dag)_{p_1,i_1} \!\cdots (u^\dag)_{p_k,i_k} \, (u^\dag)_{p_1,\sigma m_1} \!\cdots (u^\dag)_{p_k,\sigma m_k} \cdot \nonumber\\
        &\hspace{80mm}\cdot (u)_{\pi j_1, p_1} \!\cdots (u)_{\pi j_k, p_k} \, (u)_{\lambda l_1, p_1} \!\cdots (u)_{\lambda l_k, p_k}, 
        \nonumber 
        \\
        &=\frac{1}{(k!)^2}\sum_{\substack{\boldsymbol{p}\\ \pi,\sigma,\lambda\in\mathcal{S}_k}} \!\mathrm{sgn}(\pi\sigma\lambda) \: \Phi_{2k}(\boldsymbol{i},\sigma\boldsymbol{m},\pi\boldsymbol{j},\lambda\boldsymbol{l};\boldsymbol{p},\boldsymbol{p},\boldsymbol{p},\boldsymbol{p})\label{eq:coefficients_Mk}
    \end{align}
    where $\mathcal{S}_k$ is the group of permutations of $k$ elements and $\mathrm{sgn}(\pi)$ is the sign of the permutation $\pi$. Here $u$ is the mode space unitary matrix associated to the free fermions particle number preserving evolution $U$, as in Equation \eqref{evolutionfreeHapp}, and $\Phi_{2k}$ is the moment operator defined as in~\eqref{eq:t_moment_op_elements}. Notice that $M^{(2k)}_{\boldsymbol{i} \boldsymbol{j},\boldsymbol{l} \boldsymbol{m}}$ is anti-symmetric under any permutation of the indices within $\boldsymbol{i}$, $\boldsymbol{j}$, $\boldsymbol{l}$ and $\boldsymbol{m}$, and Hermitian under exchanging $\boldsymbol{i}$, $\boldsymbol{j}$ with $\boldsymbol{l}$ and $\boldsymbol{m}$, that is $M^{(2k)}_{\boldsymbol{i} \boldsymbol{j},\boldsymbol{l} \boldsymbol{m}}=\overline{M}^{(2k)}_{\boldsymbol{l} \boldsymbol{m},\boldsymbol{i} \boldsymbol{j}}$.
\end{proposition}
\begin{proof}
    We have
    \begin{align}
        \mathcal{M}&(a^\dag_{i_1}\cdots a^\dag_{i_k} a_{j_1}\cdots a_{j_k})=\\
        \nonumber
        &=\sum_{\boldsymbol{n}} \underset{U}{\mathbb{E}} \braket{\boldsymbol{n}|U a^\dag_{i_1}\cdots a^\dag_{i_k} a_{j_1}\cdots a_{j_k} U^\dag |\boldsymbol{n}} U^\dag \ket{\boldsymbol{n}} \bra{\boldsymbol{n}} U \\
        \nonumber
        &=\sum_{\boldsymbol{p},\boldsymbol{q}} \underset{U}{\mathbb{E}} \:\overline{u}_{i_1,p_1} \cdots \overline{u}_{i_k,p_k} u_{j_1,q_1} \cdots u_{j_k,q_k} \sum_{\boldsymbol{n}} \!\braket{\boldsymbol{n}| a^\dag_{p_1}\cdots a^\dag_{p_k} a_{q_1}\cdots a_{q_k}|\boldsymbol{n}} U^\dag \ket{\boldsymbol{n}} \bra{\boldsymbol{n}} U \\
        \nonumber
        &=\sum_{\boldsymbol{p},\pi}  \underset{U}{\mathbb{E}} \:\overline{u}_{i_1,p_1} \cdots \overline{u}_{i_k,p_k} u_{j_1,\pi p_1} \cdots u_{j_k,\pi p_k} \sum_{\boldsymbol{n}} \braket{\boldsymbol{n}| a^\dag_{p_1}\cdots a^\dag_{p_k} a_{\pi p_1}\cdots a_{\pi p_k}|\boldsymbol{n}} U^\dag \ket{\boldsymbol{n}} \bra{\boldsymbol{n}} U \label{eq:Mblock-proof-line3}\\
        &=\sum_{\boldsymbol{p},\pi}  \mathrm{sgn}(\pi) \: \underset{U}{\mathbb{E}} \:\overline{u}_{i_1,p_1} \cdots \overline{u}_{i_k,p_k} u_{j_1,\pi p_1} \cdots u_{j_k,\pi p_k} \: U^\dag\left(\sum_{\boldsymbol{n}} \ket{\boldsymbol{n}}\braket{\boldsymbol{n}| a^\dag_{p_1}\cdots a^\dag_{p_k} a_{p_1}\cdots a_{p_k}|\boldsymbol{n}}   \bra{\boldsymbol{n}} \right)U \\
        \nonumber
        &=\sum_{\boldsymbol{p},\pi} \mathrm{sgn}(\pi) \: \underset{U}{\mathbb{E}} \:\overline{u}_{i_1,p_1} \cdots \overline{u}_{i_k,p_k} u_{j_1,\pi p_1} \cdots u_{j_k,\pi p_k} \: U^\dag a^\dag_{p_1}\cdots a^\dag_{p_k} a_{p_1}\cdots a_{p_k} U \label{eq:Mblock-proof-line5}\\
        &=\sum_{\boldsymbol{p},\boldsymbol{l},\boldsymbol{m},\pi}  \mathrm{sgn}(\pi) \: \underset{U}{\mathbb{E}} \:\overline{u}_{i_1,p_1} \cdots \overline{u}_{i_k,p_k} u_{\pi^{-1} \pi j_1,\pi p_1} \cdots u_{\pi^{-1} \pi j_k,\pi p_k} \: u_{l_1,p_1} \cdots u_{l_k,p_k} \overline{u}_{m_1,p_1} \cdots \overline{u}_{m_k,p_k}  \cdot \nonumber \\
        &\hspace{120mm}\cdot a^\dag_{l_1}\cdots a^\dag_{l_k} a_{m_1}\cdots a_{m_k} \\
        &=\sum_{\boldsymbol{p},\boldsymbol{l},\boldsymbol{m},\pi}  \mathrm{sgn}(\pi) \: \underset{U}{\mathbb{E}} \:(u^\dag)_{p_1,i_1} \!\cdots (u^\dag)_{p_k,i_k} \, (u^\dag)_{p_1,m_1} \!\cdots (u^\dag)_{p_k,m_k} \, (u)_{\pi^{-1} j_1, p_1} \!\cdots (u)_{\pi^{-1} j_k, p_k} \, (u)_{l_1, p_1} \!\cdots (u)_{l_k, p_k} \cdot \nonumber \\
        &\hspace{120mm}\cdot a^\dag_{l_1}\cdots a^\dag_{l_k} a_{m_1}\cdots a_{m_k}\,.
    \end{align}
    In step \eqref{eq:Mblock-proof-line3} we have used the fact that $\ket{\boldsymbol{n}}$ is a Fock state, therefore, we only have a non-vanishing expression if the indices $\boldsymbol{q}$ are a permutation of the indices $\boldsymbol{p}$ such that the operators $a_q$ annihilate exactly the same fermions that the operators $a_p^\dag$ create. In step \eqref{eq:Mblock-proof-line5} we have observed that the operator $a^\dag_{p_1}\cdots a^\dag_{p_k} a_{p_1}\cdots a_{p_k}$ is diagonal in the Fock basis and, therefore, $\sum_{\boldsymbol{n}} \ket{\boldsymbol{n}}\braket{\boldsymbol{n}| a^\dag_{p_1}\cdots a^\dag_{p_k} a_{p_1}\cdots a_{p_k}|\boldsymbol{n}}   \bra{\boldsymbol{n}}$ is just its regular expansion with respect to these diagonal elements. Finally, in the last step we have reordered the unitaries $u_{\pi^{-1} \pi j_1,\pi p_1} \cdots u_{\pi^{-1} \pi j_k,\pi p_k}$ so as to "undo" the effect of the permutation $\pi$.  To reach the result \eqref{eq:coefficients_Mk}, notice that as the permutations form a group, summing over $\pi$ or its inverse is the same, and further notice that the tensor $a^\dag_{l_1}\cdots a^\dag_{l_k} a_{m_1}\cdots a_{m_k}$ is anti-symmetric, therefore the indices $\boldsymbol{l}$ and $\boldsymbol{m}$ that are contracted with it can be anti-symmetrised. 
    
    To show Hermicity, first compute $\overline{M}^{(2k)}_{\boldsymbol{i} \boldsymbol{j},\boldsymbol{l} \boldsymbol{m}}$. To verify that this expression is equal to ${M}^{(2k)}_{\boldsymbol{l} \boldsymbol{m},\boldsymbol{i} \boldsymbol{j}}$, we again make use of the fact that we can arbitrarily reorder the unitary entries in each term, giving us
    \begin{align}
        \overline{M}^{(2k)}_{\boldsymbol{i} \boldsymbol{j},\boldsymbol{l} \boldsymbol{m}} & =  \frac{1}{(k!)^2}\sum_{\substack{\boldsymbol{p}\\
        \pi,\sigma,\lambda\in\mathcal{S}_k}} \!\mathrm{sgn}(\pi\sigma\lambda) \:\underset{U \sim \mathcal{U} }{\mathbb{E}} (u)_{i_1,p_1} \!\cdots (u)_{i_k,p_k} \, (u)_{\sigma m_1,p_1} \!\cdots (u)_{\sigma m_k, p_k} \cdot \nonumber\\
        &\hspace{60mm}\cdot (u^\dag)_{\pi^{-1} p_1, j_1} \!\cdots (u^\dag)_{\pi^{-1} p_k, j_k} \, (u^\dag)_{\lambda^{-1} p_1, l_1} \!\cdots (u^\dag)_{\lambda^{-1} p_k, l_k}.
    \end{align}
    Hermicity then follows from replacing the sum over $\boldsymbol{p}$ with one over $\lambda^{-1}\boldsymbol{p}$, as well as summing instead over the independent permutations $\tilde{\pi}=\sigma\lambda^{-1}$, $\tilde{\sigma}=\lambda^{-1}\pi$, and $\tilde\lambda = \lambda^{-1}$. 
\end{proof}

This result implies that each of the $2k$-point operator subspaces constitutes an invariant subspace, and the measurement channel is block-diagonal with respect to these. Depending on which type of correlation functions one is interested in estimating, it is sufficient to study $\mathcal{M}$ on the corresponding subspace. Similarly, the pseudo-inverse $\mathcal{M}^+$ will also be block diagonal with respect to the subspaces $W_{2k}$, so it will be sufficient to "pseudo-invert" $\mathcal{M}$ within the subspace of interest. Within each subspace $W_{2k}$, the action of the measurement channel is defined by the matrix $M^{(2k)}$, while $\mathcal{M}^+\vert_{W_{2k}}$ can be constructed by computing the pseudo-inverse of $M^{(2k)}$.

\subsection{$2$-point operator sector} 

Due to the block-diagonal form of $\mathcal{M}$, let us now focus on the two-point sector $W_2$. Assuming that $L$ is odd, the corresponding matrix $M^{(2)}$ is as follows.

\begin{proposition}[Form of $M^{(2)}$ for $CU_{\rm Sym}(L)$]\label{prop:Channel2point}\ 
    Let $U$ be drawn such that the corresponding mode space unitary $u$ is uniformly distributed in $\mathcal{U}=CU_{\rm Sym}(L)$. If $L$ is odd, then
    \begin{equation}
        M^{(2)}_{i j, lm} = \frac{1}{L}\left( \delta_{i,j}\delta_{l,m}+\delta_{i,l}\delta_{j,m}-\dfrac{1}{L}\delta_{i-j,l-m}\right).
    \end{equation}
\end{proposition}
\begin{proof}
    Starting from equation \eqref{eq:coefficients_Mk} for $k=1$, we get the desired result using Eq.~(\ref{eq:matrixelementmoment}) for the second moment operator. 
\end{proof}

Now that we have an explicit expression for the matrix $M^{(2)}$, we will characterise the image $\Im(\mathcal{M}\vert_{W_2})$ and pseudo-inverse $\mathcal{M}^+\vert_{W_2}$, as this will give us valuable insights into what correlation functions can be estimated, as well as the necessary post-processing. 

\begin{proposition}[Kernel of $\mathcal{M}\vert_{W_2}$ for $CU_{\rm Sym}(L)$] 
\label{Prop:kernel}
The kernel of $\mathcal{M}\vert_{W_2}$, for Haar random $u\in CU_{\rm Sym}(L)$, is given by
\begin{align}
    &\ker(\mathcal{M}\vert_{W_2})=\mathrm{span} \left\{\sum^{L-1}_{i=0}a_i^\dagger a_{i+k}\right\}_{k=1}^{L-1},
\end{align}
where indices are understood modulo L. The given elements are orthogonal with respect to the Hilbert-Schmidt inner product.
\end{proposition}
\begin{proof}
A general operator $O=\sum^{L-1}_{i,j=0} c_{i,j} a_i^\dagger a_j $, for $c_{i,j}\in \mathbb{C}$, is contained in $ \ker(\mathcal{M}\vert_{W_2})$ iff
\begin{align}
\boldsymbol{0}=\mathcal{M}\vert_{W_2}\left(O\right)=\sum^{L-1}_{p,q=0}a_p^\dagger a_q\sum^{L-1}_{i,j=0}\frac{ c_{i,j} }{L}\left( \delta_{i,j}\delta_{p,q}+\delta_{i,p}\delta_{j,q}-\dfrac{1}{L}\delta_{i-j,p-q}\right) .
\end{align}
Since the $a_p^\dagger a_q$ are linearly independent operators, all coefficients have to be zero, giving
\begin{equation}
O\in \ker(\mathcal{M}\vert_{W_2}) \iff c_{p,q} =\begin{cases} 0 & \textrm{, if } p=q \\ f_{p-q} & \textrm{, if } p \neq q \end{cases},
\end{equation}
for arbitrary coefficients $f_k$, i.e., $c_{p,q}$ can only depend on the difference $p-q$.
\end{proof}

\begin{proposition}[Image of $\mathcal{M}\vert_{W_2}$ for $CU_{\rm Sym}(L)$]
    The image of $\mathcal{M}\vert_{W_2}$, for Haar random $u\in CU_{\rm Sym}(L)$, is
    \begin{align}
    &\Im(\mathcal{M}\vert_{W_2})=\left\{ O=\sum^{L-1}_{l,m=0} b_{l,m} a_l^\dagger a_m\,\,\Big\vert  \, \sum^{L-1}_{i=0}b_{i,i+k}=0 ,\, k \in \{1,\dots,L-1\} \right\}.
\end{align}
A basis of $\Im(\mathcal{M}\vert_{W_2})$ is given by the set of operators
\begin{align}
            \hspace{8em}\left\{a_i^\dagger a_j - a_{L-1}^\dagger a_{j-i-1}\right\}_{i \in \{0,\dots, L-2\} , \, j\in \{0,\dots, L-1\}  :  j\neq i} \,\cup\,  \left\{a_i^\dagger a_i\right\}^{L-1}_{i=0}.
            \label{eq:BasisImage}
        \end{align}
\label{Prop:image2}
\end{proposition}
\begin{proof}
     As $\mathcal{M}$ is self-adjoint with respect to the Hilbert-Schmidt inner product, $\Im(\mathcal{M}\vert_{W_2})=\ker(\mathcal{M}\vert_{W_2})^\perp$. Using Proposition \ref{Prop:kernel}, we have that $O=\sum^{L-1}_{l,m=0} b_{l,m} a_l^\dagger a_m  \in \Im(\mathcal{M}\vert_{W_2})$ iff
    \begin{align}
        0 &=\Tr\left( \sum^{L-1}_{i=0}a_{i+k}^\dagger a_{i} \, O \right)
        =\sum^{L-1}_{l,m,i=0} b_{l,m}\Tr\left(a_{i+k}^\dagger a_{i}\,a_l^\dagger a_m \right)\propto \sum^{L-1}_{i=0} b_{i,i+k}  \hspace{5mm} \forall \, k\in \{1,\dots,L-1\}.
\end{align}
\end{proof}

We now identify operators that are eigenvectors of the measurement channel $\mathcal{M}$. These correspond to quantities that we can straightforwardly estimate with our procedure, since applying the pseudo-inverse of $\mathcal{M}$ as in Proposition \ref{prop:estimator} can be done by simply multiplying with the inverse eigenvalue.

\begin{proposition}[Eigenvectors of $\mathcal{M}^+\vert_{W_2}$ for $CU_{\rm Sym}(L)$]\label{prop:ev_of_M2} \ 
    The non-diagonal operators in the basis given in Proposition ~\ref{Prop:image2} are eigenvectors of the measurement channel $\mathcal{M}$ for Haar random $u\in CU_{\rm Sym}(L)$, all with eigenvalue $1/L$: Let $i,j$ be like in ~\ref{eq:BasisImage}, then
        \begin{align}
            &\mathcal{M}\left(a_i^\dagger a_j- a_{L-1}^\dagger a_{j-i-1}\right)=\frac{1}{L}\left(a_i^\dagger a_j- a_{L-1}^\dagger a_{j-i-1}\right) .
    \end{align} 
\label{Prop:eigenstatesChannel}
\end{proposition}
\begin{proof}\
    This can be explicitly verified using Proposition ~\ref{prop:Channel2point}.
\end{proof}

 \subsection{$4$-point operator sector}
We now move on to the sector of four point correlation operators.
We are interested in studying the measurement channel restricted to the $k=2$ block, that is $\mathcal{M}\vert_{W_4}$, which is characterised by the coefficient matrix $M^{(4)}$. For this matrix we have the following expression:

\begin{proposition}[Form of $M^{(4)}$ for $CU_{\rm Sym}(L)$]\label{prop:Channel4point}\ 
    Consider the ensemble of free, particle number preserving evolution, where
    mode space unitaries are distributed according to the Haar measure on $CU_{\rm Sym}(L)$. Then
    \begin{equation}
         M^{(4)}_{i_1 i_2 j_1 j_2, l_1 l_2 m_1 m_2}=\frac{1}{4} \sum_{\pi,\sigma,\lambda\in\mathcal{S}_2} \mathrm{sgn}(\pi\sigma\lambda) \; \tilde{M}^{(4)}_{i_1 i_2 \pi(j_1) \pi(j_2), \sigma(l_1) \sigma(l_2) \lambda(m_1) \lambda(m_2)},
    \end{equation}
    where
    \begin{align}
        \tilde{M}^{(4)}_{i_1 i_2 j_1 j_2, l_1 l_2 m_1 m_2} = \sum_{p_1,p_2} \frac{1}{L^8} &\sum_{\sigma, \sigma' \in \{\pm 1\}^4} \left[\frac{L^4}{2^4}\sum_{A^{(1,1,1,1)}}\delta(A^{(1,1,1,1)}_1) \delta(A^{(1,1,1,1)}_2) \delta(A^{(1,1,1,1)}_3) \delta(A^{(1,1,1,1)}_4) \right. \nonumber \\
&\hspace{30mm}-\frac{L^3}{2^5}\sum_{A^{(1,1,1)}}\delta(A^{(1,1,1)}_1) \delta(A^{(1,1,1)}_2) \delta(A^{(1,1,1)}_3) \nonumber \\
&\hspace{30mm} -\frac{L^3}{2^3} \sum_{A^{(2,1,1)}}\delta(A^{(2,1,1)}_1)\delta(A^{(2,1,1)}_2)\delta(A^{(2,1,1)}_3) \nonumber \\
&\hspace{30mm} +\frac{L^2}{2^4}\sum_{A^{(2,1)}}\delta(A^{(2,1)}_1)\delta(A^{(2,1)}_2) + \frac{9 L^2}{2^6}\sum_{A^{(1,1)}}\delta(A^{(1,1)}_1)\delta(A^{(1,1)}_2) \nonumber \\
&\hspace{30mm}+\frac{L^2}{2^2}\sum_{A^{(2,2)}}\delta(A^{(2,2)}_1)\delta(A^{(2,2)}_2) - \frac{9L^2}{2^5}\sum_{A^{(2)}}\delta(A^{(2)}_1) \nonumber \\
& \hspace{30mm} + L^2\sum_{A^{(3,1)}}\delta(A^{(3,1)}_1) \delta(A^{(3,1)}_2) -\frac{L}{2} \sum_{A^{(3)}}\delta(A^{(3)}_1) \nonumber \\
& \hspace{30mm} \left. -\frac{193L}{2^7}\sum_{A^{(1)}}\delta(A^{(1)}_1)  + L\sum_{A^{(4)}}\delta(A^{(4)}_1)\hspace{40mm}\right].
    \end{align}
    The objects $A^{(a,b,\dots)}$ are defined similarly to after equation \eqref{eqn:beast}. That is, the sums over $A^{(a,b,\dots)}$ represent summing over all the possible ways of doing the following operation: consider the two sets of indices $\mathcal{A}=\{(p_1-i_1), (p_2-i_2), (p_1-m_1), (p_2-m_2)\}$ and $\mathcal{B}=\{(j_1-p_1), (j_2-p_2), (l_1-p_1), (l_2-p_2)\}$, and then consider all the possible ways of grouping them in such a way that the first group contains $a$ indices from the first set and $a$ from the second set, the second group contains $b$ indices from the first set and $b$ from the second set, and so on. The sums of the indices in each group are the components of $A^{(a,b,\dots  )}_i$. Again, summation with respect to the variables $\sigma, \sigma'$ indicates that one must include all possible assignments of signs to the sets $\mathcal{A}$ and $\mathcal{B}$.
\end{proposition}
\begin{proof}
    This form for the matrix $M^{(4)}$ follows from taking  Eq.~\eqref{eq:coefficients_Mk} for the $k=2$ case and using result \eqref{eqn:beast} to evaluate the resulting fourth moment expression.
\end{proof}

The total number of independent elements of $M^{(4)}$ that one would need to compute is $L^4 (L-1)^4/16$ (since we must consider only $i_1<i_2$, $j_1<j_2$, $l_1<l_2$, $m_1<m_2$, as the remaining elements are related by a sign change).  
Taking into account the polynomial scaling with $L$ of all the relevant quantities, we can conclude that it is practical to evaluate and invert $M^{(4)}$ numerically for reasonably large $L$. Note moreover that this only needs to be done once in our whole post-processing procedure (as well as future ones). 

\section{Estimating individual correlation functions}
Equipped with the unbiased estimators from the previous section, our procedure is able to estimate two-point expectation values for all $O\in \Im{\mathcal{M}\vert_{W_2}}$ and four-point expectation values for all $O\in \Im{\mathcal{M}\vert_{W_4}}$. However, due to the non-invertibility of $\mathcal{M}$, it may seem that individual correlation functions $\Tr(a_i^\dagger a_j \rho)$ or $\Tr(a_i^\dagger a^\dag_j a_k a_l \rho)$ are inaccessible with operations in $CU_{\rm Sym}(L)$ alone.
However, by extending our fermion lattice by some auxiliary modes, we may circumvent this limitation.

\subsection{Estimating 2-point correlation functions}
\label{sec:samplecomplexity}
Suppose $\rho$ is a fermionic state on a one-dimensional lattice with an even number of modes $L$. Further, assume that we can prepare the product state $\rho^\prime:=\rho \otimes \ketbra{0}{0}$, i.e., an embedding of $\rho$ into an $L+1$-mode lattice with the additional, auxiliary mode unoccupied. By Proposition \ref{Prop:image2}, translationally invariant evolution on the extended lattice allows for estimating the expectation values of $a_i^\dagger a_j - a_{(L+1)-1}^\dagger a_{j-i-1} $, for all $i,j\in\{0,\dots,L-1\}$. This turns out to be sufficient for our purpose. Notice indeed that 
\begin{equation}
    \Tr\left((a_i^\dagger a_j - a_{(L+1)-1}^\dagger a_{j-i-1} )\rho^\prime \right) = \Tr\left(a_i^\dagger a_j \rho \right)\,.
\end{equation}
Following Propositions \ref{prop:estimator} and \ref{prop:ev_of_M2}, the estimator for this quantity is 
\begin{align}
    X^{(i,j)}_{u,\boldsymbol{n}}&= \braket{\boldsymbol{n}|U\mathcal{M}^+(a_i^\dagger a_j - a_{(L+1)-1}^\dagger a_{j-i-1})U^\dag|\boldsymbol{n}}\\[1em]
    \nonumber
    &=(L+1) \braket{\boldsymbol{n}|U(a_i^\dagger a_j - a_{(L+1)-1}^\dagger a_{j-i-1})U^\dag|\boldsymbol{n}}\\[0.5em]
    \nonumber&=(L+1)\sum^{L}_{k=0} \left(\,\overline{u}_{i,k}u_{j,k} - \overline{u}_{L,k}u_{j-i-1,k}\right) n_k,
    \nonumber
\end{align}
where $n_k$ is the $k$-th entry of $\boldsymbol{n}$.
For this estimator we can further prove an upper bound on the required sample complexity. 

\begin{proposition}[Sample complexity upper bound]\label{prop:sample_complexity_2}
Let $\epsilon, \delta > 0$, and 
\begin{equation}
N\ge \frac{16}{\epsilon^{ 2}}\left(L+1\right)^2\log\left(\frac{2L\left(L-1\right)}{\delta}\right). 
\end{equation}
Further, let $\hat{X}^{(i,j)}=\frac{1}{N}\sum^{N}_{m=1}X^{(i,j)}_{u^{(m)},\boldsymbol{n}^{(m)}}$ be the arithmetic mean of $N$ single-round estimates $X^{(i,j)}_{u^{(m)},\boldsymbol{n}^{(m)}}$, where $m$ denotes which measurement cycle's data is used to compute the estimate. Then
\begin{align}
    \operatorname{max}_{\, i>j }\abs{ \hat{X}^{(i,j)} - \Tr\left(\rho a_i^\dagger a_j\right) } < \epsilon \quad \text{with probability} \, \geq 1-\delta.
\end{align}
Moreover, $\hat{X}^{i,j}$ can be efficiently computed in time $\mathcal{O}\left(L \times N \right)$.  
\end{proposition}
\begin{proof}
Let $\epsilon^\prime=\frac{\epsilon}{\sqrt{2}}$. For $N$ samples, we prove that imaginary and real parts of $\Tr\left(\rho a_i^\dagger a_j\right)$ are recovered up to additive error $< \epsilon^\prime$ and failure probability $\leq\frac{\delta}{2}$. A union bound then gives us the desired result.
Since the i.i.d.~random variables $X^{(i,j)}_{u^{(m)},\boldsymbol{n}^{(m)}}$ can be easily bounded by $\abs{ X^{(i,j)}_{u^{(m)},\boldsymbol{n}^{(m)}} } \le 2\left(L+1\right)$, we use Hoeffding's inequality to lower bound the probability 
\begin{align}
    &P\left(\operatorname{max}_{i>j}\abs{ \,\mathfrak{Re}\,\hat{X}^{(i,j)} - \mathfrak{Re}\,\Tr\left(\rho a_i^\dagger a_j\right) } < \epsilon^\prime \right)\\[0.5em]
    \ge \, & 1-\sum_{i>j} P\left( \, \abs{ \mathfrak{Re}\,\hat{X}^{(i,j)} - \mathfrak{Re}\, \Tr\left(\rho a_i^\dagger a_j\right) } \ge \epsilon^\prime \right).\nonumber
\end{align}
In particular, by Hoeffding's inequality, 
\begin{align}
    P\left( \abs{ \,\, \frac{1}{N}\sum^{N}_{m=1} \mathfrak{Re} \, X^{(i,j)}_{u^{(m)},\boldsymbol{n}^{(m)}}  - \mathfrak{Re} \, \Tr\left(\rho a_i^\dagger a_j\right) } \ge \epsilon^\prime \right)\le2 \exp\left(-\frac{N\epsilon^{\prime 2}}{8\left(L+1\right)^2}\right).
\end{align}
Plugging in $N$ gives an upper bound of $\frac{\delta}{2}$. The analysis for the imaginary part is analogous.

\end{proof}

\subsection{Estimating $k$-point correlation functions}
The procedure for estimating individual 2-point correlation functions thanks to the addition of an auxiliary mode as in the previous section can be generalized to higher order correlation functions.
Our aim is to estimate $\Tr\left(\rho A_{\boldsymbol{i},\boldsymbol{j}}\right)$ for an individual correlation operator $A_{\boldsymbol{i},\boldsymbol{j}}=a_{i_1}^\dagger\dots a_{i_k}^\dagger a_{j_1}\dots a_{j_k}$. In general, however, we find that for our distribution of unitaries $\mathcal{P}_{\Im(\mathcal{M})}(A_{\boldsymbol{i},\boldsymbol{j}})\neq A_{\boldsymbol{i},\boldsymbol{j}}$, making the direct use of Proposition~\ref{prop:estimator} unfeasible. 
Nevertheless, we can consider the state $\rho'=\rho \otimes \ketbra{\boldsymbol{0}}{\boldsymbol{0}}_{\mathrm{anc}}$, where a certain number $L_{\mathrm{anc}}$ of auxiliary modes, initially in the unoccupied state, have been added to our system.
We will denote by $W_{2k}'$ the $2k$-point operator space on the enlarged system of $L_{\mathrm{tot}}=L+L_{\mathrm{anc}}$ modes.

Then, to estimate the correlation function $A_{\boldsymbol{i},\boldsymbol{j}}$ it is sufficient to find an operator $O_{\boldsymbol{i},\boldsymbol{j}}'\in W_{2k}'$ such that $\Tr \rho' O_{\boldsymbol{i},\boldsymbol{j}}'  = \Tr \rho A_{\boldsymbol{i},\boldsymbol{j}} $ and $O_{\boldsymbol{i},\boldsymbol{j}}'\in \Im(\mathcal{M}')$, where $\mathcal{M}'$ is the measurement channel on the enlarged system. The latter condition in particular implies that there must exist an operator $X^{(\boldsymbol{i},\boldsymbol{j})}\in W_{2k}'$ such that $O_{\boldsymbol{i},\boldsymbol{j}}' = \mathcal{M}'(X^{(\boldsymbol{i},\boldsymbol{j})})$.
Expanding the relevant quantities on a basis of correlation operators, we have
\begin{align}
    X^{(\boldsymbol{i},\boldsymbol{j})}&:=
    \sum_{l,m=0}^{L_\mathrm{tot}-1} X^{(\boldsymbol{i},\boldsymbol{j})}_{\boldsymbol{l},\boldsymbol{m}} a_{l_1}^\dagger\dots a_{l_k}^\dagger a_{m_1}\dots a_{m_k} ,\\
    O_{\boldsymbol{i},\boldsymbol{j}}' &:= \sum_{l,m,r,s=0}^{L_\mathrm{tot}-1} X^{(\boldsymbol{i},\boldsymbol{j})}_{\boldsymbol{l},\boldsymbol{m}} M'^{(2k)}_{\boldsymbol{l},\boldsymbol{m};\boldsymbol{r},\boldsymbol{s}} a_{r_1}^\dagger\dots a_{r_k}^\dagger a_{s_1}\dots a_{s_k} \,.
\end{align}
The condition $\Tr \rho' O_{\boldsymbol{i},\boldsymbol{j}}'  = \Tr \rho A_{\boldsymbol{i},\boldsymbol{j}} $ is satisfied if the expansion of $O_{\boldsymbol{i},\boldsymbol{j}}'$ does not contain correlation operators of the original system $W_{2k}$, other than $A_{\boldsymbol{i},\boldsymbol{j}}$, that is
\begin{equation}
    \sum_{l,m=0}^{L_\mathrm{tot}-1} X^{(\boldsymbol{i},\boldsymbol{j})}_{\boldsymbol{l},\boldsymbol{m}} M'^{(2k)}_{\boldsymbol{l},\boldsymbol{m};\boldsymbol{r},\boldsymbol{s}} = \delta_{\boldsymbol{r},\boldsymbol{i}} \delta_{\boldsymbol{s},\boldsymbol{j}}\hspace{10mm}\mbox{ for all }\boldsymbol{r},\boldsymbol{s}\in \{0,\dots,L-1\}^k\,. \label{eq:system_for_X}
\end{equation}
In conclusion, identifying a suitable $X^{(\boldsymbol{i},\boldsymbol{j})}$ amounts to solving the linear system $\sum_{l,m} X^{(\boldsymbol{i},\boldsymbol{j})}_{\boldsymbol{l},\boldsymbol{m}} \widetilde{M}^{(2k)}_{\boldsymbol{l},\boldsymbol{m};\boldsymbol{r},\boldsymbol{s}} = B^{(\boldsymbol{i},\boldsymbol{j})}_{\boldsymbol{r},\boldsymbol{s}}$, where $\widetilde{M}^{(2k)}$ is a rectangular matrix obtained by removing from $M'^{(2k)}$ the columns corresponding to correlation operators acting on at least one auxiliary mode, and $B^{(\boldsymbol{i},\boldsymbol{j})}_{\boldsymbol{r},\boldsymbol{s}}=\delta_{\boldsymbol{r},\boldsymbol{i}} \delta_{\boldsymbol{s},\boldsymbol{j}}$. Adding a sufficient number of auxiliary modes will lead to $\widetilde{M}^{(2k)}$ having maximal rank, and thus the system will have a solution.

Once $X^{(\boldsymbol{i},\boldsymbol{j})}$ has been identified, we need to compute the corresponding estimator. According to Proposition~\ref{prop:estimator}, this is achieved by evaluating 
\begin{equation}
    {\mathcal{M}'}^{+}(O_{\boldsymbol{i},\boldsymbol{j}}')={\mathcal{M}'}^{+} \circ \mathcal{M}' (X^{(\boldsymbol{i},\boldsymbol{j})})= \mathcal{P}_{\Im(\mathcal{M}')} (X^{(\boldsymbol{i},\boldsymbol{j})})\,. \label{eq:projection_of_X}
\end{equation}

Let us point out that, as $M'^{(2k)}$ is positive semi-definte, both solving equation~\eqref{eq:system_for_X} and performing the projection~\eqref{eq:projection_of_X} numerically only require the implementation of a function that returns the value of $\sum_{l,m} X_{\boldsymbol{l},\boldsymbol{m}} M'^{(2k)}_{\boldsymbol{l},\boldsymbol{m};\boldsymbol{r},\boldsymbol{s}}$ for an arbitrary $X$. Storing the full matrix $M'^{(2k)}$ is thus not required, which can be helpful given that it can reach a considerable size already for $k=2$ and moderately large $L$. This can be achieved for example using a \emph{conjugate gradient} algorithm for equation~\eqref{eq:system_for_X}. The projection~\eqref{eq:projection_of_X} can be implemented by iteratively solving the differential equation $d/dt X_{\boldsymbol{l},\boldsymbol{m}}(t)=-\sum_{r,s} M'^{(2k)}_{\boldsymbol{l},\boldsymbol{m};\boldsymbol{r},\boldsymbol{s}} (X_{\boldsymbol{r},\boldsymbol{s}}(t)- X^{(\boldsymbol{i},\boldsymbol{j})}_{\boldsymbol{r},\boldsymbol{s}})$, with initial condition $X_{\boldsymbol{l},\boldsymbol{m}}(0)=0$, which has solution $X(t)=(\openone-e^{-M't})X^{(\boldsymbol{i},\boldsymbol{j})} $ that converges to $\mathcal{P}_{\Im(\mathcal{M}')} (X^{(\boldsymbol{i},\boldsymbol{j})})$ in the long-time limit.

\section{Recovery with nearest-neighbour Hamiltonians} \label{sec:nn-hamiltonians}
We have seen so far that by sampling unitaries uniformly from the symmetric circulant subgroup $CU_{\text{Sym}}(L)$, which only requires free fermionic, translationally invariant evolution, we can recover all 2-point and 4-point correlation functions. However, it is difficult to implement these random unitaries in a physical system due to the complex (long-range) interactions required and the limited control we currently have, as well as possibly short coherence times. A more realistic proposal for fermionic devices would be to consider nearest-neighbour Hamiltonians, with only \emph{real} hopping amplitudes. The simplest Hamiltonian of this form looks like
\begin{equation}
    H_{\rm n.n} = - J \sum^{L-1}\limits_{i=0} \left(a_i^\dagger a_{i+1} + a_{i+1}^\dagger a_i\right)\enspace , 
    \label{Eq.nnAPP}
\end{equation}
where $J\in\mathbb{R}$ is a coupling term that dictates the strength of hopping between them. Time-evolution under these nearest-neighbour Hamiltonians with real hopping amplitudes, gives rise to the unitaries:
\begin{equation}
    U_{\rm n.n}(\alpha) = \exp\left(i\alpha\sum^{L-1}\limits_{i=0} \left(a_i^\dagger a_{i+1} + a_{i+1}^\dagger a_i\right)\right), \ \alpha \in \mathbb{R} \enspace ,
\end{equation}
where $\alpha= Jt$ and $t$ is the quench time. These nearest-neighbour unitaries form the one-parameter subgroup $CU_{\rm n.n}(L) = \{U_{\rm n.n}(\alpha)\,, \ \forall \alpha \in \mathbb{R}\}$ within the symmetric circulant group:
\begin{equation}
    CU_{\rm n.n}(L) \subset CU_{\rm Sym}(L) \enspace .
\end{equation}
As we have seen in Proposition~\ref{circulant_prop}, all unitaries contained in $CU(L)$ are defined by their eigenvalues in the Fourier basis, which must lie on the unit circle. For $CU(L)$, the phases $\varphi_k$ can be any number in $[0, 2\pi)$, and for the symmetric circulant group $\varphi_k = \varphi_{L-k}$. For  the nearest-neighbour circulant unitaries, the phases are constrained as
\begin{equation}
     \varphi_k = 2\alpha\cos(\frac{2\pi}{L} k) \ \text{mod} \ 2\pi \hspace{10mm}\forall U \in CU_{\rm n.n}(L)\enspace .
\end{equation}
This follows from observing that $H_{\rm n.n}$ can be written in the form \eqref{freeHapp} with $h_{i,j}= \delta_{j,i+1}+\delta_{j,i-1}$ and computing its eigenvalues as in Proposition~\ref{circulant_prop}. 

\subsection{Approximating the Haar distribution}\label{sec:nn-hamiltonians_one}
An ensemble of unitaries is said to be a unitary $t$-design for some group of unitaries if all moment operators up to order $t$ (see Definition~\ref{def:momop}) match those of the uniform measure on said group. We relax this notion somewhat, and study whether ensembles of nearest-neighbour unitaries, with support contained in $CU_{\rm n.n}(L)$, can approximate the uniform measure on the ``full group'' $CU_{\rm Sym}(L)$. 

The elements of the subgroup $CU_{\rm Sym}(L)$ are defined by the vector $\boldsymbol{\varphi} = (\varphi_0, \varphi_1, \dots , \varphi_{\lfloor L/2\rfloor})$, which is an element of the $(\lfloor L/2\rfloor + 1)$-torus. Here, the entries of $\boldsymbol{\varphi}$ can be arbitrary phases in $\left[0,2\pi\right)$. Sampling unitaries uniformly according to the Haar measure on $CU_{\rm Sym}(L)$ coincides with sampling $\boldsymbol{\varphi}$ uniformly from the $(\lfloor L/2\rfloor + 1)$-torus, as this gives a right-invariant distribution. In the case of $CU_{\rm n.n}(L)$, we instead have
\begin{equation}
    \boldsymbol{\varphi} = \alpha \boldsymbol{c} \ \text{mod} \ 2\pi \enspace , 
\end{equation}
where $\boldsymbol{c} = (2, 2\cos (\frac{2\pi}{L}), \dots , 2\cos (\frac{2\pi}{L}\lfloor L/2\rfloor))$. If, for suitable odd $L$, the elements of $\boldsymbol{c}$ were all rationally independent, then $\boldsymbol{\varphi}$ would be an irrational winding of the $([L/2] + 1)$-torus, which is known to be dense~\cite{pontrjagin1939topological}. Given this, there should be a way of sampling $\alpha$ such that the corresponding elements of $CU_{\rm n.n}(L)$ cover $CU_{\rm Sym}(L)$ uniformly, approximating even higher order moments accurately. Using well-known results from number theory, it is not hard to show rational independence when $L$ is prime:

\begin{lemma}[Linear independence of roots of unity \cite{johnsen1985lineare}]
\label{lemm:roots_of_unity}
Let $L\in\mathbb{N}$ be a positive integer. Then\\

\hspace{20mm} L is square-free if and only if the primitive L-th roots of unity are linearly independent over $\mathbb{Q}$.
\end{lemma}

A positive integer is said to be square-free if its factorization does not contain duplicate primes, and the primitive $L$-th roots of unity are $\{e^{i\frac{2\pi}{L}k}\,\vert \, 1\leq k \leq L,\, \operatorname{gcd}(k,L)=1\}$, i.e., those that are not also roots of unity of lower order. If $L$ is prime, it is in particular square-free, and the entries of $\boldsymbol{c}$ will be linearly independent over $\mathbb{Q}$.

However, rational independence seems to be a much too strong requirement if one is just interested in approximating low order moments. For a more pragmatic analysis, let us bound the difference between the two moment operators of order $t=2$. The bounds below can be straightforwardly generalized to higher $t$.
In the following, $\boldsymbol{p}$ denotes an arbitrary distribution on $\mathbb{R}$ according to which $\alpha$ is sampled, which we will specify later on. 
As both $CU_{\rm Sym}(L)$ and $CU_{\rm n.n}(L)$ are circulant, we only care about a fixed column
\begin{align}
    &\norm{\mathbb{E}_{u\sim CU_{\rm Sym}(L)} u^{\otimes 2}\otimes u^{\dag \otimes 2}\ket{0}^{\otimes 4}-\mathbb{E}_{\alpha\sim\boldsymbol{p}} u(\alpha)^{\otimes 2}\otimes u(\alpha)^{\dag \otimes 2}\ket{0}^{\otimes 4}}_2^2 \leq  \nonumber\\[0.5em]
    &\hspace{20mm} \leq \max_{0\leq k,l,m,n \leq \lfloor\frac{L}{2}\rfloor}\abs{\, \mathbb{E}_{\phi_i \sim \left[0,2\pi\right)}e^{i(\phi_k+\phi_l - \phi_m - \phi_n)}-\mathbb{E}_{\alpha \sim \boldsymbol{p}}e^{i\alpha\left(\cos\frac{2\pi k}{L} + \cos\frac{2\pi l}{L} - \cos\frac{2\pi m}{L} - \cos\frac{2\pi n}{L} \right)}}^2,  \label{eqn:design_bound} \\[0.5em]
    &\hspace{20mm}\leq \max_{\substack{0\leq k,l,m,n\leq \lfloor\frac{L}{2}\rfloor\\ \{k,l\}\neq \{m,n\}}} \abs{\, \mathbb{E}_{\alpha \sim \boldsymbol{p}}e^{i\alpha\left(\cos\frac{2\pi k}{L} + \cos\frac{2\pi l}{L} - \cos\frac{2\pi m}{L} - \cos\frac{2\pi n}{L} \right)}}^2. \label{eqn:design_bound2}
\end{align}
Here, in the first step we used that all unitaries in the two ensembles are diagonal with respect to the Fourier basis, see Proposition \ref{circulant_prop}, and the properties of the corresponding eigenvalues. In the second step, we exploited the fact that for $\{k,l\}=\{m,n\}$ the right-hand side of \eqref{eqn:design_bound} is zero, and the first term vanishes whenever $\{k,l\}\neq\{m,n\}$. 
Let us define \begin{equation}
\kappa^{(2)}(L) = \min_{\substack{0\leq k,l,m,n \leq \lfloor \frac{L}{2} \rfloor\\ \{k,l\}\neq \{m,n\}}}\abs{\cos\frac{2\pi k}{L} + \cos\frac{2\pi l}{L} - \cos\frac{2\pi m}{L} - \cos\frac{2\pi n}{L}},\end{equation}
which will be the relevant quantity to (lower-)bound. The superscript $(2)$ references $t=2$; for $t>2$, one can straightforwardly extend the definition above. Note that $\kappa^{(t)}(L)$ contains a very specific linear combination of cosines, reinforcing that rational independence is indeed too stringent for our purposes.

We will now examine concrete distributions $\alpha\sim \boldsymbol{p}$. If we choose $\boldsymbol{p}$ to be the uniform distribution on $[0,\alpha_{\rm max}]$, one finds
\begin{equation} \eqref{eqn:design_bound2}\leq \frac{2}{\alpha_{\rm max}\kappa^{(2)}(L)}.
\end{equation}

Alternatively, if $\boldsymbol{p}$ is the normal distribution with mean $\mu$ and variance $\sigma$ (in practice, one would of course truncate this distribution suitably),
\begin{equation}
    \eqref{eqn:design_bound2} \leq e^{-\sigma^2 \kappa^{(2)}(L)^2}, \label{eq:normal-dist-bound}
\end{equation}
where we may replace $\kappa^{(2)}(L)$ with $\kappa^{(t)}(L)$ in both bounds for moments $t > 2$.

For both distributions $\boldsymbol{p}$, whenever $\kappa^{(t)}(L)$ is non-zero, we can approximate $t$-th order moments of $CU_{\rm Sym}(L)$ arbitrarily well with just nearest-neighbour evolution, by increasing $\alpha_{\rm max}$ or $\sigma$ accordingly. This would require implementing larger hopping rates and/or evolution times. Due to Lemma \ref{lemm:roots_of_unity}, $\kappa^{(t)}(L)$ will of course be non-zero if $L$ is prime. More interestingly, we observe numerically that $\kappa^{(2)}(L)$ behaves roughly like $\sim \frac{500}{L^4}$, and the decay is not monotonous in $L$.

The inverse-exponential scaling of the approximation error~\eqref{eq:normal-dist-bound} in $\sigma$ gives the normal distribution a very favourable behaviour, at least theoretically. 
On the other hand, sampling $\alpha$ from a uniform distribution supported entirely on a positive or negative interval may be preferable in those cases where it is hard to change the sign of $J$ for the underlying physical system (for instance by exchanging fermions with holes). In these cases, one may, for example, want to fix the nearest-neighbour Hamiltonian and in particular $J$, and just perform measurements for different evolution times. 

These considerations support our intuition that for odd $L_{\text{tot}} = L + L_{\text{anc}}$ there are feasible distributions of coefficients $\alpha$ that give rise to ensembles sufficiently dense in the symmetric circulant group $CU_{\rm Sym}(L)$. It is our hope that further research will deepen our understanding of which distributions are most favourable. Overall, these results are encouraging and indicate that we can expect to find experimentally driven ensembles that our recovery procedure can be applied to. Below, in Sec.~\ref{sec:nn-hamiltonians_two} we determine numerically  that with uniformly distributed $\alpha$ we can recover $2$- and $4$-point correlation functions with our recovery procedure. 

When implementing these recovery procedures in practice there will be two main sources of error. One will be given by the fact that we can repeat the protocol for only a finite number $N$ of unitary evolutions. This is a statistical error, in the sense that by suitably increasing the number of samples $N$ we can get with high probability closer and closer to the expected value of the sampling procedure. The second source of error descends from the fact that our chosen distribution on $CU_{\rm n.n}(L)$ approximates the Haar measure $CU_{\rm Sym}(L)$ only up to a certain finite precision, as computed for example in~\eqref{eqn:design_bound2}. This will give a systematic bias to our recovery procedure, in the sense that this error will persist also in the limit of infinitely large sample number $N$. The only way to reduce this error is to improve the distribution on $CU_{\rm n.n}(L)$, for instance by choosing larger $\alpha_{\rm max}$ or $\sigma$ in the previous examples.

In the following propositions we give more detailed guarantees on the impact of these errors. In Proposition~\ref{prop:sample_compl_bound}, we give a bound on the number of samples needed to achieve a certain sampling accuracy. Then, in Proposition~\ref{prop:bias_from_approx}, we bound how well the chosen ensemble must approximate the Haar measure on $CU_{\rm Sym}$ to achieve a given recovery bias.

\begin{proposition}[Sample complexity upper bound]
\label{prop:sample_compl_bound}
Let $\epsilon, \delta > 0$. Assume that we want to estimate the $k$-point correlation functions $\Tr\left(\rho a_{{\bbb{i}}}^\dagger a_{{\bbb{j}}}\right)\equiv \Tr\left(\rho a_{i_1}^\dagger \cdots a_{i_k}^\dagger a_{j_1}\cdots a_{j_k}\right)$ for a set of $M_{\rm{tot}}$ distinct indices $(\bbb{i},\bbb{j})$. For each $(\bbb{i},\bbb{j})$ we then consider $\hat{X}^{({\bbb{i}},{\bbb{j}})} =\frac{1}{N}\sum^{N}_{m=1}X^{({\bbb{i}},{\bbb{j}})} _{u^{(m)},\boldsymbol{n}^{(m)}}$ , which represents the arithmetic mean of $N$ single-round estimates 
\begin{equation}
X^{({\bbb{i}},{\bbb{j}})} _{u^{(m)},\boldsymbol{n}^{(m)}}=\braket{\boldsymbol{n}^{(m)}|U^{(m)}\mathcal{M}^+(O^{({\bbb{i}},{\bbb{j}})} ){U^{(m)}}^\dag|\boldsymbol{n}^{(m)}}, 
\end{equation}
where $O^{({\bbb{i}},{\bbb{j}})}$ denotes the operator on the enlarged system used to form the unbiased estimator for the correlation function $\Tr\left(\rho a_{{\bbb{i}}}^\dagger a_{{\bbb{j}}}\right)$ and $m$ labels the data from a specific measurement cycle. To achieve the estimation error
\begin{align}
\max_{({\bbb{i}},{\bbb{j}})}\vert \hat{X}^{({\bbb{i}},{\bbb{j}})}  - \Tr\left(\rho a_{{\bbb{i}}}^\dagger a_{{\bbb{j}}}\right) \vert < \epsilon \quad \text{with probability } \geq 1-\delta,
\end{align}
it suffices to choose $N$ such that
\begin{align}
    N\ge \frac{4}{\epsilon^{ 2}}\log\left(\frac{4 M_{\rm tot}}{\delta}\right)\max_{({\bbb{i}},{\bbb{j}})} f(O^{({\bbb{i}},{\bbb{j}})})^2 ,
\end{align}
where $M_{\rm {tot}}$ is the total number of correlation functions to estimate. Here $f(O^{({\bbb{i}},{\bbb{j}})})$ is a function dependent on $O^{({\bbb{i}},{\bbb{j}})}$:
\begin{align}    f\left(O^{({\bbb{i}},{\bbb{j}})}\right):=\sum^{\mathrm{dim}\left(W_{2k}\right)}_{s=1} \lvert \sum^{\mathrm{dim}\left(W_{2k}\right)}_{l=1} M_{s,l}^{+} O^{({\bbb{i}},{\bbb{j}})}_l \rvert,
\end{align} where $O^{({\bbb{i}},{\bbb{j}})}_l$ is defined such that  $O^{({\bbb{i}},{\bbb{j}})}=\sum^{\mathrm{dim}\left(W_{2k}\right)}_{l=1}O^{({\bbb{i}},{\bbb{j}})}_l A_l $, where  $\{A_l\}^{\mathrm{dim}\left(W_{2k}\right)}_{l=1}$ represent the elements of the basis of $W_{2k}$ given by
\begin{equation}
\Big\{a_{i_1}^\dagger\dots a_{i_k}^\dagger a_{j_1}\dots a_{j_k} \, : \, 0 \le i_1 < \dots < i_k\le L-1 \text{ and } 0\le j_1< \dots < j_k\le L-1 \Big\},\end{equation}
and $M_{s,l}^{+}$ are the matrix elements associated to the pseudo-inverse channel $\mathcal{M}^+$ in the basis $\{A_l\}^{\mathrm{dim}\left(W_{2k}\right)}_{l=1}$.
It is worth noting that the function $f\left(O^{({\bbb{i}},{\bbb{j}})}\right)$ can be efficiently estimated during classical pre-processing.
\end{proposition}

\begin{proof}
Let $\epsilon'=\frac{\epsilon}{\sqrt{2}}$. We aim to show that for $N$ samples, the real and imaginary parts of $\Tr\left(\rho a_{{\boldsymbol{i}}}^\dagger a_{{\boldsymbol{j}}}\right)$ can be recovered within an additive error $\epsilon'$ and a failure probability bounded by $\frac{\delta}{2}$.  By applying a union bound, we obtain the desired result. To do so, we use Hoeffding's inequality, whose concentration bound depends on the range of the random variable $X _{u^{(m)},\boldsymbol{n}^{(m)}}$.
We have
\begin{align}
\Big\vert X^{({\boldsymbol{i}},{\boldsymbol{j}})}_{u^{(m)},\boldsymbol{n}^{(m)}} \Big\vert   &= \Big\vert \sum^{\mathrm{dim}\left(W_{2k}\right)}_{l=1}O^{({\boldsymbol{i}},{\boldsymbol{j}})}_l \braket{\boldsymbol{n}|U\mathcal{M}^+(A_l)U^\dag|\boldsymbol{n}}\Big\vert = \Big\vert \sum^{\mathrm{dim}\left(W_{2k}\right)}_{s,l=1} M_{s,l}^{+} O^{({\boldsymbol{i}},{\boldsymbol{j}})}_l \braket{\boldsymbol{n}|A_s |\boldsymbol{n}}\Big\vert \\ &\le  \sum^{\mathrm{dim}\left(W_{2k}\right)}_{s=1} \Big\vert \sum^{\mathrm{dim}\left(W_{2k}\right)}_{l=1} M_{s,l}^{+} O^{({\boldsymbol{i}},{\boldsymbol{j}})}_l \Big\vert \, \Big\vert \braket{\boldsymbol{n}|A_s |\boldsymbol{n}}\Big\vert  \le  \sum^{\mathrm{dim}\left(W_{2k}\right)}_{s=1} \Big\vert \sum^{\mathrm{dim}\left(W_{2k}\right)}_{l=1} M_{s,l}^{+} O^{({\boldsymbol{i}},{\boldsymbol{j}})}_l \Big\vert  ,
\nonumber
\end{align}
where in the last two steps, we used the triangle inequality and the fact that $\lvert \braket{\boldsymbol{n}|A_l |\boldsymbol{n}}\rvert \le \norm{A_l}_\infty \le 1 $, which follows from the sub-multiplicativity of the $\infty$-norm and the fact that the $\infty$-norm of an annihilation or creation operator is upper bounded by $1$. Therefore, $\lvert X^{({\boldsymbol{i}},{\boldsymbol{j}})}_{u^{(m)},\boldsymbol{n}^{(m)}} \rvert \le f(O^{({\boldsymbol{i}},{\boldsymbol{j}})})$. Also, it holds that
\begin{align}
    &\quad \,\, P\left(\operatorname{max}_{( {\boldsymbol{i}}, {\boldsymbol{j}})}\Big\vert \mathfrak{Re}\, \hat{X}^{({\boldsymbol{i}}, {\boldsymbol{j}})} - \mathfrak{Re}\, \Tr\left(\rho a_{ {\boldsymbol{i}}}^\dagger a_{ {\boldsymbol{j}}}\right) \Big\vert < \epsilon^\prime \right)
    \nonumber
    \\
    &\ge \, 1-\sum_{( {\boldsymbol{i}}, {\boldsymbol{j}})} P\left( \Big\vert \mathfrak{Re}\, \hat{X}^{{( {\boldsymbol{i}}, {\boldsymbol{j}})}} - \mathfrak{Re} \, \Tr\left(\rho a_{ {\boldsymbol{i}}}^\dagger a_{ {\boldsymbol{j}}}\right) \Big\vert \ge \epsilon^\prime \right),
\end{align}
where the sum over $( {\boldsymbol{i}}, {\boldsymbol{j}})$ runs over all the $M_{\rm tot}$ observables of interest.
\noindent
In particular, by Hoeffding's inequality, 
\begin{align}
    P\left( \Big\vert \, \frac{1}{N}\sum^{N}_{m=1} \mathfrak{Re}\, X^{( {\boldsymbol{i}}, {\boldsymbol{j}})}_{u^{(m)},\boldsymbol{n}^{(m)}}  - \mathfrak{Re}\, \Tr\left(\rho a_{ {\boldsymbol{i}}}^\dagger a_{ {\boldsymbol{j}}} \right) \Big\vert \ge \epsilon^\prime \right)\le2 \exp\left(- \frac{2N\epsilon^{\prime 2}}{\left(2f(O^{( {\boldsymbol{i}}, {\boldsymbol{j}})})\right)^2}\right).
\end{align}
Plugging in $N$ gives an upper bound of $\frac{\delta}{2}$. The analysis for the imaginary part is analogous.
\end{proof}

\begin{proposition}
\label{prop:bias_from_approx}
    Let $\tilde{\epsilon}$, $\delta>0$ . Consider $\hat{X}^{( {\boldsymbol{i}}, {\boldsymbol{j}})}$ and $N$ defined as in the previous proposition to give an estimation error $\epsilon$. 
    If each unitary is now sampled from an ensemble $\nu$ (in our case on $CU_{\rm n.n}$), satisfying
    \begin{align}
        \norm{\mathbb{E}_{u\sim \nu} (u^{\otimes 2k } \otimes u^{* \otimes 2k }) - \mathbb{E}_{u\sim CU_{\rm Sym}} (u^{\otimes 2k } \otimes u^{* \otimes 2k }) }_\infty\le \left(k!\, L^{3k}  \max_{ ({\boldsymbol{i}}, {\boldsymbol{j}})}f(O^{( {\boldsymbol{i}}, {\boldsymbol{j}})})\right)^{-1} \tilde{\epsilon},
    \end{align}
    where $f(O^{( {\boldsymbol{i}}, {\boldsymbol{j}})})$, $O^{({\boldsymbol{i}}, {\boldsymbol{j}})}$, $O_l^{( {\boldsymbol{i}}, {\boldsymbol{j}})}$ and $M^+$ are defined as in the previous proposition,
    then we have that
\begin{align}
    \max_{ ({\boldsymbol{i}}, {\boldsymbol{j}})}\vert \hat{X}^{( {\boldsymbol{i}}, {\boldsymbol{j}})}  - \Tr\left(\rho a_{ {\boldsymbol{i}}}^\dagger a_{ {\boldsymbol{j}}}\right) \vert < \epsilon + \tilde{\epsilon} \,,\text{ with probability } \geq 1-\delta.
\end{align}
\end{proposition}
\begin{proof}
    To simplify the notation, let us denote $O:=O^{( {\boldsymbol{i}}, {\boldsymbol{j}})}$ and $\hat{X}:=\hat{X}^{( {\boldsymbol{i}}, {\boldsymbol{j}})}$. The proof strategy is to separate the error arising from a finite number of measurements and the error resulting from using an approximate unitary design. 
    Furthermore, we define $X_\nu$ as the expected value of the random variable $\hat{X}$ when the unitaries are sampled from the unitary distribution $\nu$, i.e.:
    \begin{align}
         X_\nu := \underset{U\sim \mathcal{\nu},\boldsymbol{n}}{\mathbb{E}} \braket{\boldsymbol{n}|U\mathcal{M}^+(O)U^\dag|\boldsymbol{n}}.
    \end{align}
Denoting $\rho'$ as the quantum state $\rho$ coupled with the auxiliary system used to recover the correlation functions, we have the following:  
    \begin{align}
        \vert \hat{X}  - \Tr\left(\rho' \mathcal{M}(\mathcal{M}^+(O)) \right) \vert  \le \epsilon + \vert X_\nu- \Tr\left(\rho' \mathcal{M}(\mathcal{M}^+(O)) \right)\vert,
        \nonumber
    \end{align}
where we have employed the triangle inequality and the results of the previous proposition, which apply analogously since $\hat{X}$ is an unbiased estimator for $X_\nu$.
\noindent
Now, our goal is to bound the quantity $\vert X_\nu- \Tr\left(\rho' \mathcal{M}(\mathcal{M}^+(O)) \right) \vert$. 

\noindent
First, let us observe that $X_\nu$ can be expressed as
    \begin{align}
    \underset{U\sim \mathcal{\nu},\boldsymbol{n}}{\mathbb{E}}  \braket{\boldsymbol{n}|U\mathcal{M}^+(O)U^\dag|\boldsymbol{n}} &= \sum_{\boldsymbol{n} \in \{0,1\}^L}\underset{U \sim \nu }{\mathbb{E}}\,\braket{\boldsymbol{n}| U \rho' U^\dagger |\boldsymbol{n}} \,  \braket{\boldsymbol{n}| U \mathcal{M}^+(O) U^\dag |\boldsymbol{n}}\\[0.5em]
    \nonumber
    &=\Tr \left( \sum_{\boldsymbol{n} \in \{0,1\}^L}\underset{U \sim \nu }{\mathbb{E}}   \,  \braket{\boldsymbol{n}| U \mathcal{M}^+(O) U^\dag |\boldsymbol{n}} \,U^\dagger \ket{\boldsymbol{n}}\bra{\boldsymbol{n}} U \rho' \right) \\[1em]
     \nonumber
    &= \Tr\left(\rho' \mathcal{M}_{(\nu)} (\mathcal{M}^+(O))  \right),
     \nonumber
\end{align}
where we defined the measurement channel $\mathcal{M}_{(\nu)}$ with respect to the distribution $\nu$ as
\begin{align}
    \mathcal{M}_{(\nu)}(\cdot):=\sum_{\boldsymbol{n} \in \{0,1\}^L}\underset{U \sim \nu }{\mathbb{E}}   \,  \braket{\boldsymbol{n}| U (\cdot) U^\dag |\boldsymbol{n}} \,U^\dagger \ket{\boldsymbol{n}}\bra{\boldsymbol{n}} U.
\end{align}
As in the previous proposition, we express $O$ as  $O=\sum^{\mathrm{dim}\left(W_{2k}\right)}_{l=1}O_l A_l$ and $\mathcal{M}^+(A_l)= \sum^{\mathrm{dim}\left(W_{2k}\right)}_{s=1} M_{s,l}^{+} A_s$.
Consequently,\vspace{0.5em}
\begin{align}
    \vert X_\nu- \Tr\left(\rho' \mathcal{M}(\mathcal{M}^+(O)) \right) \vert=\vert \Tr\left(\rho' \mathcal{M}_{(\nu)}(\mathcal{M}^+(O)) \right) - \Tr\left(\rho' \mathcal{M}(\mathcal{M}^+(O)) \right) \vert \\[0.5em]
    \le \sum^{\mathrm{dim}\left(W_{2k}\right)}_{s=1} \Big\vert \sum^{\mathrm{dim}\left(W_{2k}\right)}_{l=1}O_l M_{s,l}^{+} \Big\vert \, \Big\vert  \Tr\left(\rho' \mathcal{M}_{(\nu)}(A_s) \right) - \Tr\left(\rho' \mathcal{M}(A_s) \right) \Big\vert 
\end{align}
Using Proposition~\ref{propo:blocks_M}, we have
\begin{align}
    \mathcal{M}(A_s)=\sum^{\mathrm{dim}\left(W_{2k}\right)}_{h=1} M^{(2k)}_{h,s} A_h
\end{align}
where $M^{(2k)}_{h,s}$ are defined as the coefficients in Proposition~\ref{propo:blocks_M} but with a factor $(k! ) ^2$ to account for the summation over the independent terms $\{A_l\}^{\mathrm{dim}\left(W_{2k}\right)}_{l=1}$ (utilizing the antisymmetry property of the coefficient in Proposition~\ref{propo:blocks_M} and the antisymmetry property of the correlation operators). 
Similarly, we find
\begin{align}
    \mathcal{M}_{(\nu)}(A_s)=\sum^{\mathrm{dim}\left(W_{2k}\right)}_{h=1} M^{(2k)}_{(\nu )\,h,s} A_h,
\end{align}
where $ M^{(2k)}_{(\nu )\,h,s}$ is defined analogously to before, with the average taken with respect to $\nu$ instead of the Haar measure.

\noindent
Hence, we obtain
\begin{align}
    \vert  \Tr\left(\rho' \mathcal{M}_{(\nu)}(A_s) \right) - \Tr\left(\rho' \mathcal{M}(A_s) \right) \vert &\le \sum^{\mathrm{dim}\left(W_{2k}\right)}_{h=1} \vert M^{(2k)}_{(\nu )\,h,s} - M^{(2k)}_{h,s}\vert \, \vert\Tr\left(\rho' A_h \right)\vert\\
    & \le \sum^{\mathrm{dim}\left(W_{2k}\right)}_{h=1} \vert M^{(2k)}_{(\nu )\,h,s} - M^{(2k)}_{h,s}\vert ,
    \nonumber
\end{align}
where in the last step we have used the fact that, by the Hölder inequality, $\vert\Tr\left(\rho' A_h \right)\vert \le \norm{\rho'}_1 \norm{ A_h }_{\infty} \le 1$.
We can then bound the term $\vert M^{(2k)}_{(\nu )\,h,s} - M^{(2k)}_{h,s}\vert$ as 
\begin{align}
    \vert M^{(2k)}_{(\nu )\,h,s} - M^{(2k)}_{h,s}\vert \le  (k!)^3 \, L^{k}  \norm{\mathbb{E}_{u\sim \nu} (u^{\otimes 2k } \otimes u^{* \otimes 2k }) - \mathbb{E}_{u\sim \mathrm{Haar}} (u^{\otimes 2k } \otimes u^{* \otimes 2k }) }_\infty,
\end{align}
where we applied the triangle inequality and the fact that the absolute value of each entry of a matrix is always smaller than its largest singular value. 
We arrive at
\begin{align}
    &\quad \,\, \vert X_\nu- \Tr\left(\rho' \mathcal{M}(\mathcal{M}^+(O)) \right) \vert \le \\&\le  \sum^{\mathrm{dim}\left(W_{2k}\right)}_{s=1} \Big\vert \sum^{\mathrm{dim}\left(W_{2k}\right)}_{l=1}O_l M_{s,l}^{+} \Big\vert \, \Big\vert  \Tr\left(\rho' \mathcal{M}_{\nu}(A_s) \right) - \Tr\left(\rho' \mathcal{M}(A_s) \right) \Big\vert \\
    \nonumber
    &\le \sum^{\mathrm{dim}\left(W_{2k}\right)}_{s,h=1} \Big\vert \sum^{\mathrm{dim}\left(W_{2k}\right)}_{l=1}O_l M_{s,l}^{+} \Big\vert \, \Big\vert M^{(2k)}_{(\nu )\,h,s} - M^{(2k)}_{h,s}\Big\vert \\
    \nonumber
    &\le  (k!)^3 \, L^{k} \mathrm{dim}\left(W_{2k}\right) \left( \sum^{\mathrm{dim}\left(W_{2k}\right)}_{s=1} \Big\vert \sum^{\mathrm{dim}\left(W_{2k}\right)}_{l=1}O_l M_{s,l}^{+} \Big\vert \right) \norm{\mathbb{E}_{u\sim \nu} (u^{\otimes 2k } \otimes u^{* \otimes 2k }) - \mathbb{E}_{u\sim \mathrm{Haar}} (u^{\otimes 2k } \otimes u^{* \otimes 2k }) }_\infty.
    \nonumber
\end{align}
By utilizing our assumption regarding the bound on the infinity norm difference of the vectorized moment operators, we can conclude that the latter quantity is bounded from above by $\tilde{\epsilon}$. This concludes the proof.
\end{proof}
Note that the previous proposition about the bias introduced by an approximate design makes use of only the non-interacting property of the evolution, but not of translational invariance, so it can be applied to any non-interacting protocol involving only approximate designs.

\subsection{Numerical simulations of recovery}\label{sec:nn-hamiltonians_two}
The above analytical treatment supports our intuition that nearest-neighbour interactions with real hoppings are enough to recover $2$- and $4$-point correlation functions for $L_{\text{tot}}$ odd. However, numerical exploration is required in order to get a meaningful understanding of the real numbers involved. Recall that we can approach the Haar distribution over the group $CU_{\text{Sym}}(L)$ by choosing a specific distribution from which to sample $\alpha = Jt$. We expect there to be many ways to do this, but here we focus on the simple case of uniformly distributed variables. 

Since it is currently not experimentally feasible to change the global sign of the Hamiltonian in many state-of-the-art optical lattices, we define the Hamiltonian $H_{n.n.}$ (see Eq. \eqref{Eq.nnAPP} and Fig.~\ref{fig: appendix_fig}) such that it has a global negative sign in front of the hopping coefficient $J$ that we set to be greater than $0$, which is a natural condition for many devices (see e.g.~\cite{Kuhr}). We set a maximum $\alpha_{\text{max}}$ and sample uniformly $\alpha \sim [0, \alpha_{\text{max}}]$.

Our modelled initial state is a pure state, a charge-density wave $|\psi_0\rangle = |0, 1, 0, \dots , 1, 0\rangle$ that is evolved under a Fermi-Hubbard Hamiltonian with a random field term
\begin{equation}
        H_{\text{Hub}} = \sum\limits_{i=0}^{L-2} (a_i^\dagger a_{i+1} + h.c.) + \sum\limits_{i=0}^{L-2} n_i n_{i+1} + \sum\limits_{i=0}^{L-1} g_i n_i \enspace ,
\end{equation}
for a time $t_0 = 1.5$, where $g_i$ are a fixed set of random variables sampled uniformly from $[0.2, 0.7]$. Therefore, the initial state is
\begin{equation}
    |\psi\rangle = e^{(-iH_{\text{Hub}}t_0)}|\psi_0\rangle \enspace .
\end{equation}
The initial state is embedded in an $L_{\text{tot}} = L + L_{\text{anc}}$ system, and the auxiliary modes are initialized in the state vector  $|0\rangle$. We classically simulate experiment (a) in Fig.~\ref{fig: appendix_fig} and store the values $\{\alpha_j, |\boldsymbol{n}_j\rangle\}_{j=1}^{N}$ for $N=150000$ measurement shots. 
\begin{figure}[ht]
    \centering
    \includegraphics[width=0.8\textwidth]{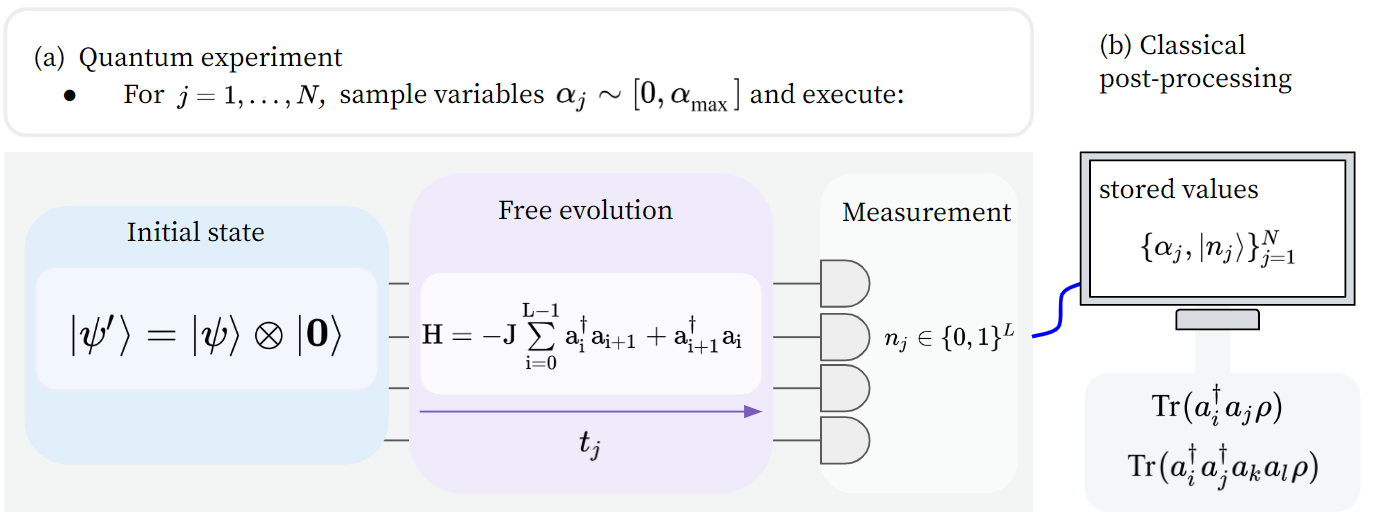}
    \caption{Schematic of the simulated experiment, with $\alpha_j = Jt_j$ ($J$ fixed)}
    \label{fig: appendix_fig}
\end{figure}

Using these stored values, we estimate the $2$-point correlation functions $\text{Tr}(a_i^\dagger a_j)$, for each $\ (i, j)\in \{0,\dots, L-1\}^{2}$ by averaging the unbiased estimator $X_{u, \boldsymbol{n}}^{(i, j)}$ in Eq.~(\ref{eq:2pointest}) for $L \in \{3, 4, 5, 6, 7, 8\}$ and, correspondingly, $L_{\text{tot}} \in \{5, 5, 7, 7, 9, 9\}$. In this equation, $n_k$ represents the $k$-th element of the vector $|\boldsymbol{n}_j\rangle$. The matrix elements of $u_j$ can be efficiently evaluated using an analytical formula that depends on $\alpha_j$, taking advantage of the diagonal nature of $u_j$ in the Fourier basis. In order to estimate all $4$-point correlation functions $\text{Tr}(a_i^\dagger a_j^\dagger a_k a_l)$, for each$ \ (i, j, k, l) \in \{0,\dots, L-1\}^4$, we numerically determine the form of the co-efficient matrix $M_{i_1 i_2, j_1 j_2, l_1 l_2, m_1 m_2}$, find its pseudo-inverse, and then calculate the unbiased estimator using Eq.~\ref{eq:system_for_X} and Eq.~\ref{eq:projection_of_X}. We do this for  $L \in \{3, 5, 7\}$ and, correspondingly, $L_{\text{tot}} \in \{5, 7, 9\}$. 

In this regime of system sizes, we can classically simulate the exact values of the $2$- and $4$-point correlation functions and use them to evaluate the accuracy of our method. In Figure~\ref{fig:appendixAlphaScalefour} we show a similar plot to Figure~\ref{fig:VarScaling} in the main text, but for the $4$-point estimates. This time, however, aligning with the colour plot of Figure~\ref{fig:colourplots}, we compute the average over a subset of all correlation functions, where the first two indices are fixed, \ie $\langle a_0^\dagger a_1^\dagger a_k a_l\rangle$ and therefore $\text{Ave}(\Delta O_{01kl}) = L^{-2} \sum_{k, l} |\langle a_0^\dagger a_1^\dagger a_k a_l\rangle - \langle a_0^\dagger a_1^\dagger a_k a_l\rangle_{\text{est}}|$. We observe, for $L \in \{3,5,7\}$, a similar trend to the $2$-point case, where the larger $\alpha_{\text{max}}$ gets the less steeply the error increases with $L$. In the main text, the inset of Figure~\ref{fig:VarScaling} shows that the average variance (over all 2-point correlation functions) scales only with the number of lattice sites $L (L_{\text{tot}})$ and not with $\alpha_{\text{max}}$. We repeat these plots below in Figures~\ref{fig:appendixAlphaScale} and ~\ref{fig:appendixVarScale} for clarity. In Figure~\ref{fig:VarScalingfour}, we see an equivalent plot for the scaling of the average variance in the $4$-point case, again taking the average over a subset of all correlation functions with $(i,j,k,l) = (0,1,k,l)$, and for $L \in \{3,5,7\}$. The variance in the $2$-point case is calculated as $\text{Var}(\langle a_i^\dagger a_j\rangle_{\rm est}) = \text{Var}(\mathfrak{Re}\,\langle a_i^\dagger a_j\rangle_{\rm est}) + \text{Var}(\mathfrak{Im}\,\langle a_i^\dagger a_j\rangle_{\rm est})$, and in the $4$-point case as $\text{Var}(\langle a_0^\dagger a_1^\dagger a_k a_l\rangle_{\rm est}) = \text{Var}(\mathfrak{Re}\, \langle a_0^\dagger a_1^\dagger a_k a_l\rangle_{\rm est}) + \text{Var}(\mathfrak{Im}\,\langle a_0^\dagger a_1^\dagger a_k a_l\rangle_{\rm est})$. We see that the average variance appears to scale sub-linearly in $L(L_{\text{tot}})$ for the $2$-point correlation function estimates. In the $4$-point case, since we have plotted three data points it is difficult to comment on concrete trends, however it appears that the average variance seems to grow sub-quadratically with $L(L_{\text{tot}})$. As mentioned, this indicates that the number of samples required for a given accuracy, in each case, will scale accordingly.
\begin{figure}[ht]
    \centering
    \subfloat[\centering \label{fig:appendixAlphaScale}]{{\includegraphics[width=0.40\textwidth]{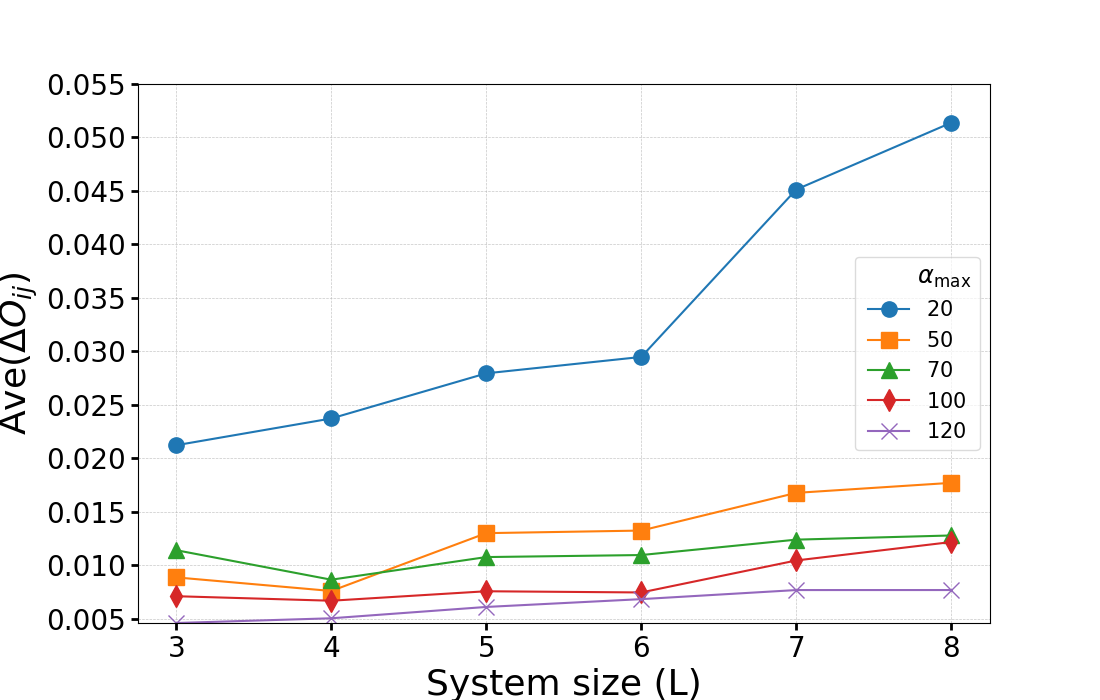} }}
    \qquad 
    \subfloat[\centering \label{fig:appendixVarScale}]{{\includegraphics[width=0.40\textwidth]{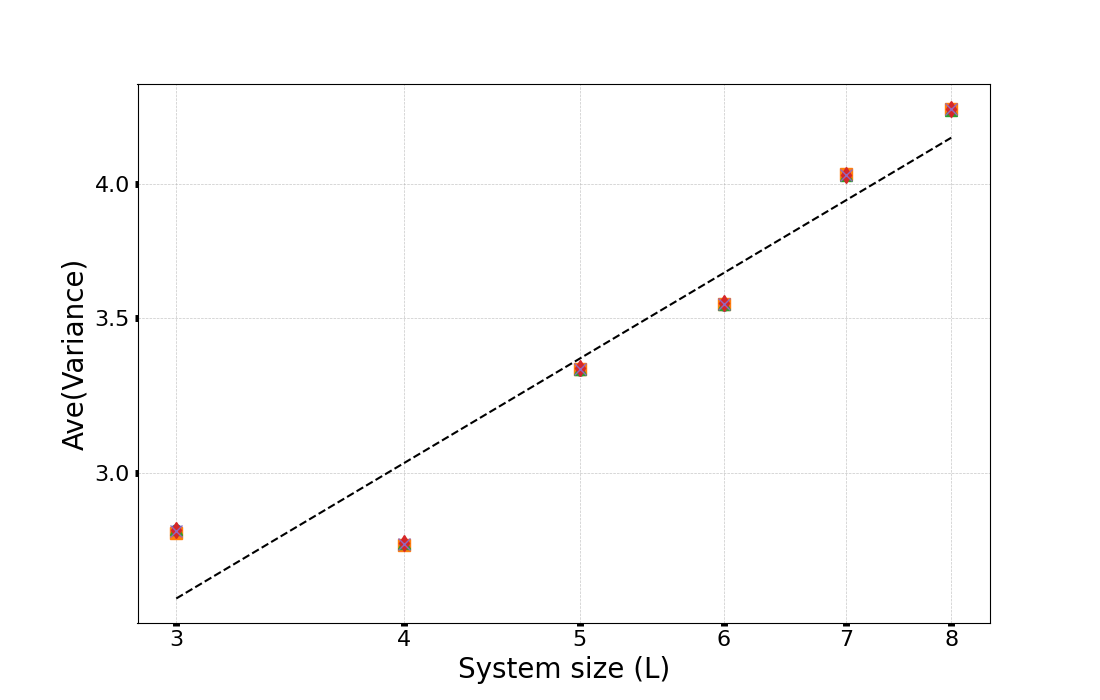}}}
    \caption{These figures show (a) the change in the average recovery error $\text{Ave}(\Delta O_{i,j})$ as a function of $L \leq 8$ for varying $\alpha_{\text{max}}$ and (b) a log-log plot illustrating the average variance of 2-point estimators for the same $L$, displaying a linear fit with a slope of $\sim 0.47$. Estimates are based on $N = 150000$ samples, chosen for high accuracy in estimating the true variance of the protocol. In practice, a significantly smaller sample size suffices for recovery, as we provide evidence for.}
    \label{fig:VarAlphScale2point}
\end{figure}
\begin{figure}[ht]
    \centering
    \subfloat[\centering \label{fig:appendixAlphaScalefour}]{{\includegraphics[width=0.40\textwidth]{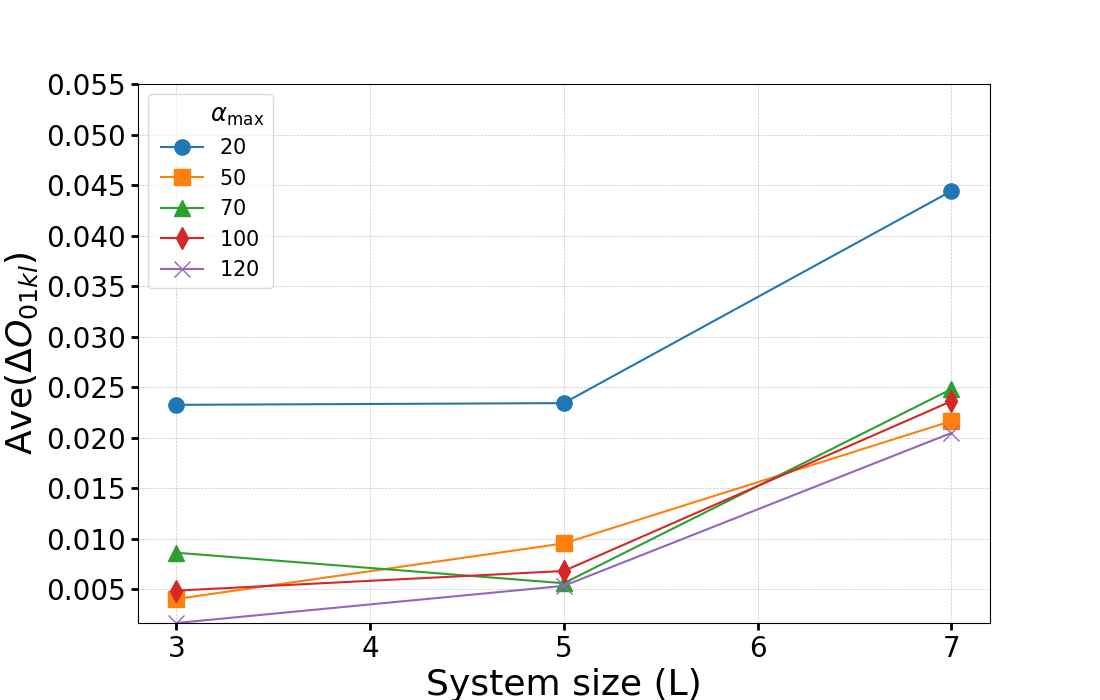} }}
    \qquad 
    \subfloat[\centering \label{fig:VarScalingfour}]{{\includegraphics[width=0.34\textwidth]{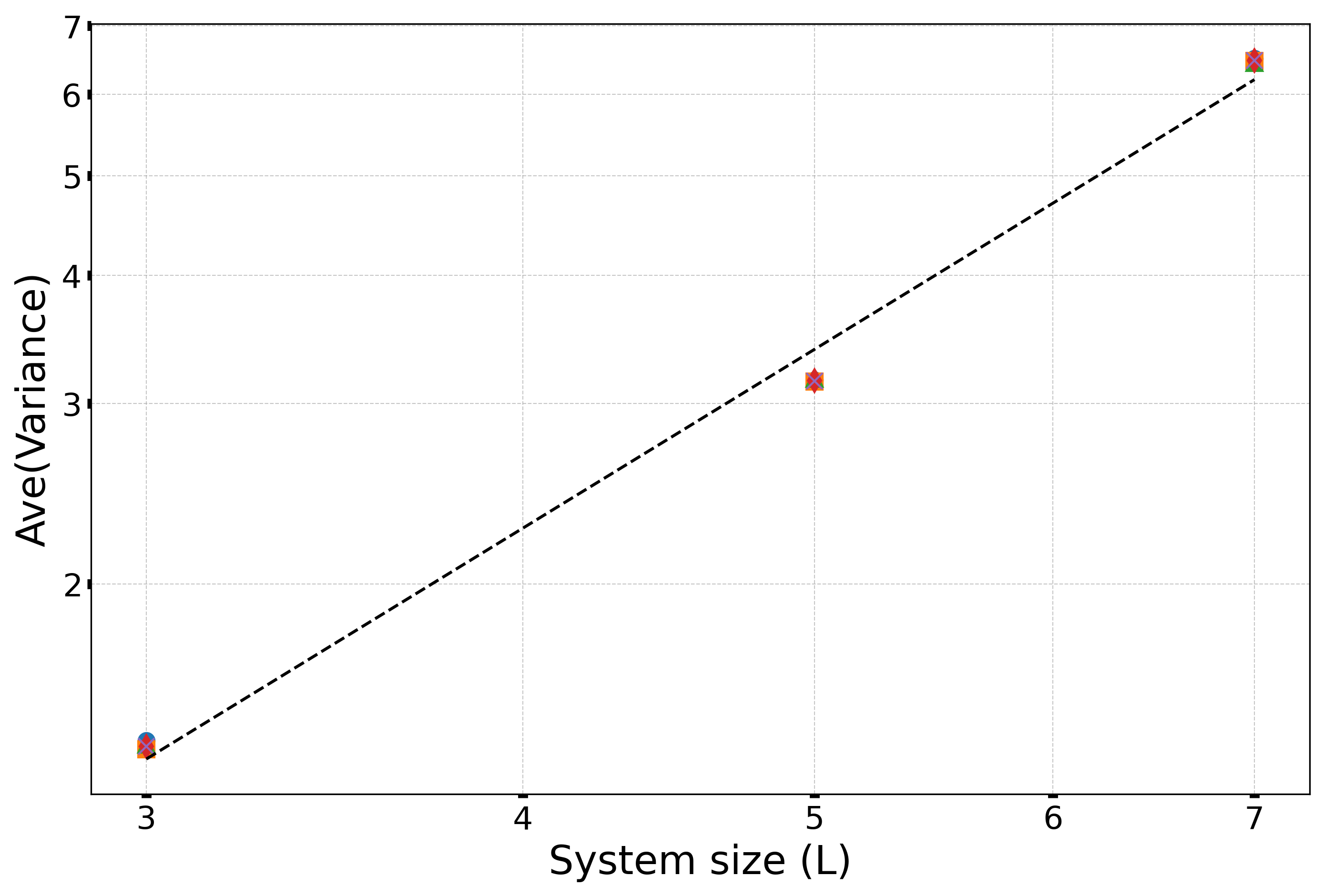}}}
    \caption{These figures demonstrate (a) the change in the average recovery error $\text{Ave}(\Delta O_{0, 1, k, l})$ as a function of $L = 3, 5, 7$ for varying $\alpha_{\text{max}}$ and (b) a log-log plot illustrating the average variance of the same subset of $4$-point estimators for the same $L$, displaying a linear fit with a slope of $\sim 1.79$. Estimates are based on $N = 150000$ samples, chosen for high accuracy in estimating the true variance of the protocol.}
    \label{fig:VarAlphScale4point}
\end{figure}

\begin{figure}[ht]
    \centering
    \subfloat[\centering \label{fig:avediff_alpha70}]{{\includegraphics[width=0.40\textwidth]{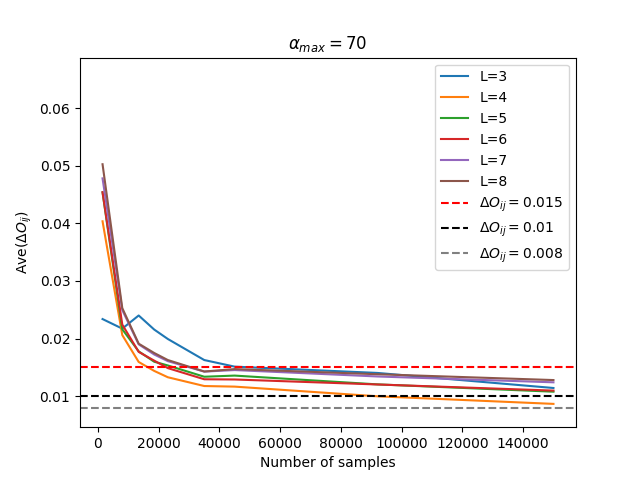} }}
    \qquad 
    \subfloat[\centering \label{fig:avediff_alpha120}]{{\includegraphics[width=0.40\textwidth]{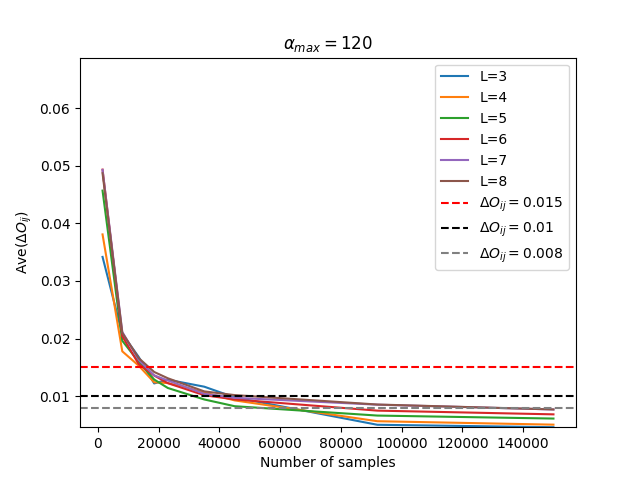}}}
    \caption{These figures demonstrate the change in $\text{Ave}(\Delta O_{i,j}) = L^{-2}\sum_{i,j}|\langle a_i^\dagger a_j \rangle - \langle a_i^\dagger a_j\rangle_{\text{est}}|$ with the number of samples used to recover the estimates $\langle a_i^\dagger a_j\rangle_{\text{est}}$, for $L\leq 8$, and for uniformly distributed $\alpha \sim [0, \alpha_{\text{max}}]$, where (a) $\alpha_{\text{max}} = 70$ and (b) $\alpha_{\text{max}} = 120$.}
    \label{fig:comparison_alpha70_alpha120}
\end{figure}
\begin{figure}[tb]
    \centering
    \subfloat[\centering \label{fig:avediff_alpha20}]{{\includegraphics[width=0.40\textwidth]{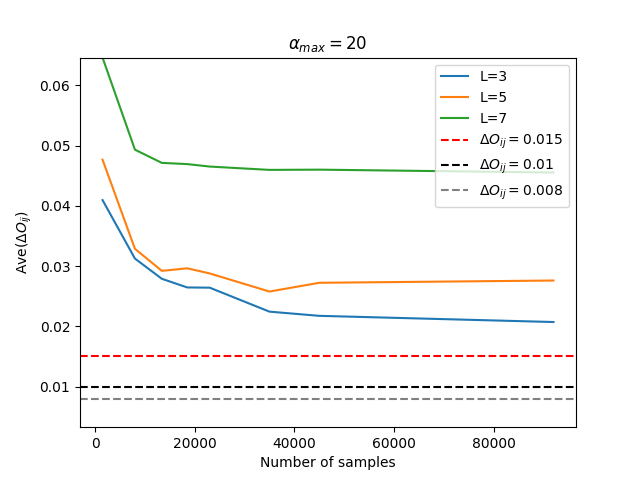} }}
    \qquad 
    \subfloat[\centering \label{fig:avediff_alpha1500}]{{\includegraphics[width=0.40\textwidth]{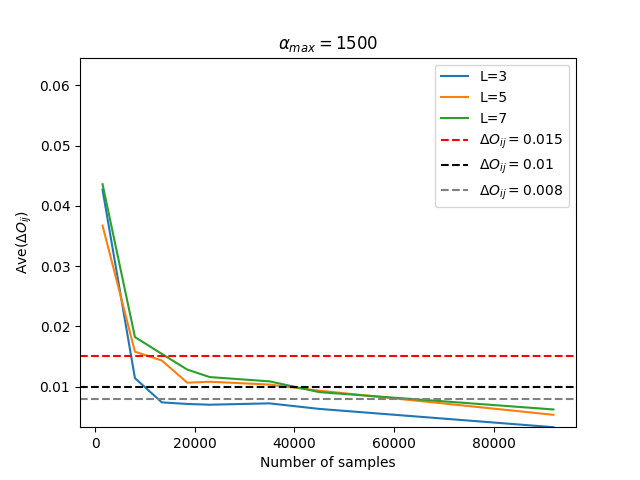}}}
    \caption{These figures demonstrate the change in $\text{Ave}(\Delta O_{i,j}) = L^{-2}\sum_{i,j}|\langle a_i^\dagger a_j \rangle - \langle a_i^\dagger a_j\rangle_{\text{est}}|$ with the number of samples used to recover the estimates $\langle a_i^\dagger a_j\rangle_{\text{est}}$, for $L = 3, 5, 7$, and for uniformly distributed $\alpha \sim [0, \alpha_{\text{max}}]$, where (a) $\alpha_{\text{max}} = 20$ and (b) $\alpha_{\text{max}} = 1500$.}
    \label{fig:comparison_alpha20_alpha1500}
\end{figure}

Below, in Figures~\ref{fig:comparison_alpha70_alpha120} and ~\ref{fig:comparison_alpha20_alpha1500} we look at the actual numbers required to recover all $2$-point correlation functions (on average) for a given error. In Fig.~\ref{fig:comparison_alpha70_alpha120} we compare the required number of samples when we set $\alpha_{\text{max}}=70$ and $\alpha_{\text{max}}=120$. We can see from Fig.~\ref{fig:avediff_alpha70} that the average errors struggle to get below $10^{-2}$ when $\alpha_{\text{max}} = 70$, whilst increasing the maximum by only a small amount to $\alpha_{\text{max}}=120$ gives us a much better chance of reaching lower errors within reasonable sample limits, as we can see in Fig.~\ref{fig:avediff_alpha120}, where for $N\sim 70000$ the averages errors are all below $10^{-2}$, reducing to $\mathcal{O}(10^{-3})$. 

Figure~\ref{fig:comparison_alpha20_alpha1500} then presents the same results, but for a sample of values $L=3, 5, 7$, for two extremal values of $\alpha_{\text{max}} \in \{20, 1500\}$ in order (i) to explore whether low (and therefore easily implemented) values of $\alpha$ can yield good estimates in the $2$-point sector and (ii) to determine whether unreasonably large (and therefore not so easily implemented) values of $\alpha$ show a significant improvement in recovering estimates (in the $2$-point sector) over more reasonable values of $\alpha$ (such as those in Fig.~\ref{fig:comparison_alpha70_alpha120}). Unsurprisingly, in Fig.~\ref{fig:avediff_alpha20} we see that $\alpha_{\text{max}}=20$ is not sufficient to yield accurate estimates, even for $L=3$, which matches our hypothesis in the earlier Sec.~\ref{sec:nn-hamiltonians_one} that $\alpha$ should be sufficiently large to approximate the uniform distribution over $CU_{\rm Sym}(L)$. However, interestingly  Fig.~\ref{fig:avediff_alpha1500} shows that even for $\alpha_{\text{max}} = 1500$ the number of samples required to get an average error below $10^{-2}$ is still $N = \mathcal{O}(10^4)$ which is the same as Fig.~\ref{fig:avediff_alpha120}, there being only a difference of $20000$ samples between them. This suggests that, at the observed system sizes, there is a not too significant trade-off between increasing $\alpha_{\text{max}}$ and the required number of samples $N$. Our conclusion is that smaller values of $\alpha = \mathcal{O}(10^{2})$ are sufficient to reach accurate estimates in times comparable to larger $\alpha$ and therefore there is not much to be gained from ramping up $\alpha$. These figures also support the information gained from the average variance plot in Figure~\ref{fig:appendixVarScale}.

In Figure~\ref{fig:fourpoint_comparison_alpha70_alpha120} below, we look at the samples required to recover the subset of $4$-point correlation functions (on average) for a given error. We compare the required number of samples when we set $\alpha_{\text{max}}=70$ and $\alpha_{\text{max}}=120$. We can see from Fig.~\ref{fig:fourpoint_avediff_alpha70} that the average errors only get below $10^{-2}$ for $L=3, 5$, when $\alpha_{\text{max}} = 70$. Increasing the maximum by only a small amount to $\alpha_{\text{max}}=120$ reduces the errors for these system sizes, and as we can see in Fig.~\ref{fig:avediff_alpha120} for $N\sim 30000$ the average errors for $L = 3, 5$ are below $10^{-2}$ reducing to $\mathcal{O}(10^{-3})$. Unfortunately, for $L=7$, the maximum system size we have explored for the $4$-point case, $\alpha_{\text{max}} = 120$ is not large enough for the average error to be $< 10^{-2}$. It is, however, worth noting that since we are not looking at the canonical average the trends do not reflect the average behaviour of all the $4$-point correlation functions. However, they still demonstrate that it is feasible to get estimates for this subset of $4$-point correlation functions with a reasonable number of samples that is comparable to the $2$-point case. Since all samples collected in the shadow estimation scheme can be re-used, this is a positive result. 

All of the figures so far have shown averages over all $2$-point correlation functions, or our chosen subset ($ijkl = 01kl$) of $4$-point correlation functions. Our final Figure~\ref{fig:individualobservables} below indicates that the behaviour of individual observable estimates follows a similar pattern to the average case. Interestingly, this is even evident for the $4$-point case when comparing to the behaviour of individual observables outside the subset that we averaged over. However, we still require more data for the $4$-point case to be confident in the behavioural trends we currently observe.
\begin{figure}[ht]
    \centering
    \subfloat[\centering \label{fig:fourpoint_avediff_alpha70}]{{\includegraphics[width=0.40\textwidth]{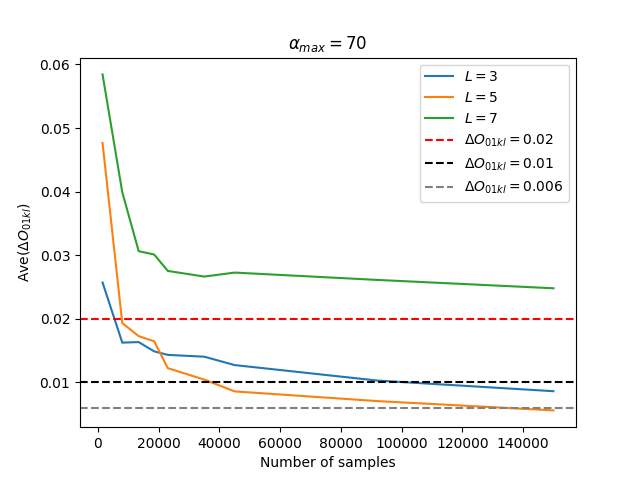} }}
    \qquad 
    \subfloat[\centering \label{fig:fourpoint_avediff_alpha120}]{{\includegraphics[width=0.40\textwidth]{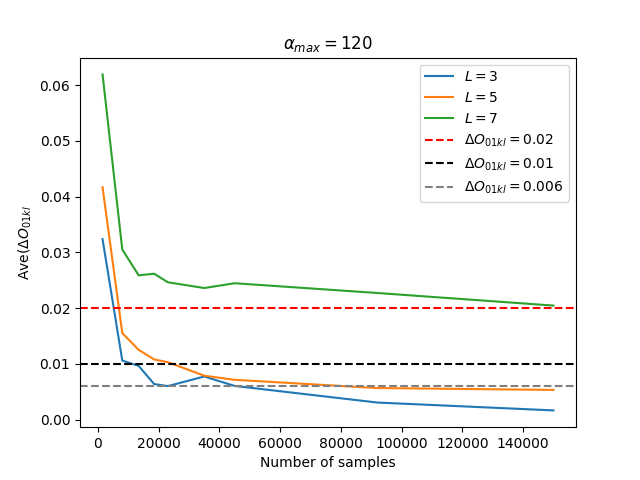}}}
    \caption{These figures demonstrate the change in $\text{Ave}(\Delta O_{0, 1, k, l}) = L^{-2}\sum_{k,l}|\langle a_0^\dagger a_1^\dagger a_k a_l\rangle - \langle a_0^\dagger a_1^\dagger a_k a_l \rangle_{\text{est}}|$ with the number of samples used to recover the estimate $\langle a_0^\dagger a_1^\dagger a_k a_l\rangle_{\text{est}}$, for $L = 3, 5, 7$, and for uniformly distributed $\alpha \sim [0, \alpha_{\text{max}}]$, where (a) $\alpha_{\text{max}} = 70$ and (b) $\alpha_{\text{max}} = 120$.}
    \label{fig:fourpoint_comparison_alpha70_alpha120}
\end{figure}

\begin{figure}[ht]
    \centering
    \subfloat[\centering]{{\includegraphics[width=0.40\textwidth]{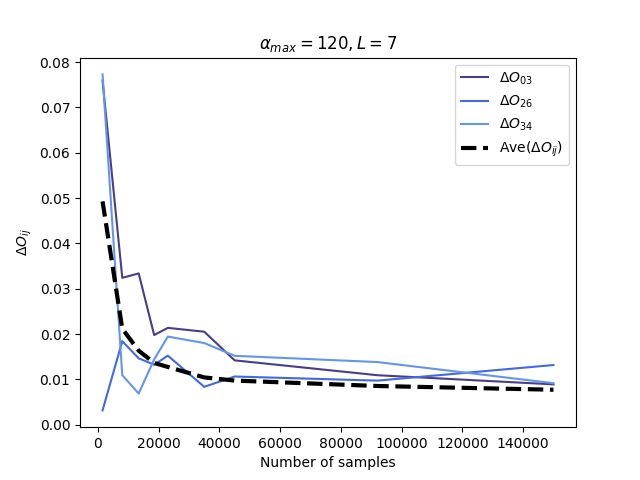} }}
    \qquad 
    \subfloat[\centering]{{\includegraphics[width=0.40\textwidth]{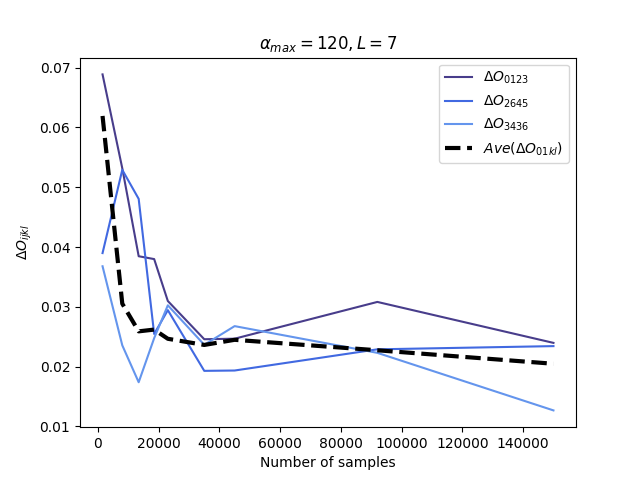}}}
    \caption{These figures show the change in the recovery error for the (a) $2$-point, $\Delta O_{i,j} = |\langle a_i^\dagger a_j\rangle - \langle a_i^\dagger a_j\rangle_{\rm est}|$, and (b) $4$-point case, $\Delta O_{i,j,k,l} = |\langle a_i^\dagger a_j^\dagger a_k a_l\rangle - \langle a_i^\dagger a_j^\dagger a_k a_l\rangle_{\rm est}|$, with the number of samples used to recover the estimates $\langle a_i^\dagger a_j\rangle_{\rm est}$ and $\langle a_i^\dagger a_j^\dagger a_k a_l\rangle_{\rm est}$ for $L=7$, for uniformly distributed $\alpha \sim [0, \alpha_{\text{max}}]$ where $\alpha_{\text{max}} = 120$. The average recovery error (a) over all $2$-point correlation functions and (b) over the subset of $4$-point correlation functions where $ijkl = 01kl$, is also plotted. }
    \label{fig:individualobservables}
\end{figure}
\end{document}

%% file: commands.tex
\providecommand{\myvec}[1]{\ensuremath{\boldsymbol{#1}}}

\providecommand{\ss}{\ensuremath{\myvec{s}}}





%% file: new_paper.bbl
\begin{thebibliography}{40}%
\makeatletter
\providecommand \@ifxundefined [1]{%
 \@ifx{#1\undefined}
}%
\providecommand \@ifnum [1]{%
 \ifnum #1\expandafter \@firstoftwo
 \else \expandafter \@secondoftwo
 \fi
}%
\providecommand \@ifx [1]{%
 \ifx #1\expandafter \@firstoftwo
 \else \expandafter \@secondoftwo
 \fi
}%
\providecommand \natexlab [1]{#1}%
\providecommand \enquote  [1]{``#1''}%
\providecommand \bibnamefont  [1]{#1}%
\providecommand \bibfnamefont [1]{#1}%
\providecommand \citenamefont [1]{#1}%
\providecommand \href@noop [0]{\@secondoftwo}%
\providecommand \href [0]{\begingroup \@sanitize@url \@href}%
\providecommand \@href[1]{\@@startlink{#1}\@@href}%
\providecommand \@@href[1]{\endgroup#1\@@endlink}%
\providecommand \@sanitize@url [0]{\catcode `\\12\catcode `\$12\catcode
  `\&12\catcode `\#12\catcode `\^12\catcode `\_12\catcode `\%12\relax}%
\providecommand \@@startlink[1]{}%
\providecommand \@@endlink[0]{}%
\providecommand \url  [0]{\begingroup\@sanitize@url \@url }%
\providecommand \@url [1]{\endgroup\@href {#1}{\urlprefix }}%
\providecommand \urlprefix  [0]{URL }%
\providecommand \Eprint [0]{\href }%
\providecommand \doibase [0]{https://doi.org/}%
\providecommand \selectlanguage [0]{\@gobble}%
\providecommand \bibinfo  [0]{\@secondoftwo}%
\providecommand \bibfield  [0]{\@secondoftwo}%
\providecommand \translation [1]{[#1]}%
\providecommand \BibitemOpen [0]{}%
\providecommand \bibitemStop [0]{}%
\providecommand \bibitemNoStop [0]{.\EOS\space}%
\providecommand \EOS [0]{\spacefactor3000\relax}%
\providecommand \BibitemShut  [1]{\csname bibitem#1\endcsname}%
\let\auto@bib@innerbib\@empty
\bibitem [{\citenamefont {Cirac}\ and\ \citenamefont
  {Zoller}(2012)}]{CiracZollerSimulation}%
  \BibitemOpen
  \bibfield  {author} {\bibinfo {author} {\bibfnamefont {J.~I.}\ \bibnamefont
  {Cirac}}\ and\ \bibinfo {author} {\bibfnamefont {P.}~\bibnamefont {Zoller}},\
  }\href {https://www.nature.com/articles/nphys2275} {\bibfield  {journal}
  {\bibinfo  {journal} {Nature Phys.}\ }\textbf {\bibinfo {volume} {8}},\
  \bibinfo {pages} {264} (\bibinfo {year} {2012})}\BibitemShut {NoStop}%
\bibitem [{\citenamefont {Trotzky}\ \emph {et~al.}(2012)\citenamefont
  {Trotzky}, \citenamefont {Chen}, \citenamefont {Flesch}, \citenamefont
  {McCulloch}, \citenamefont {Schollw\"ock}, \citenamefont {Eisert},\ and\
  \citenamefont {Bloch}}]{Trotzky}%
  \BibitemOpen
  \bibfield  {author} {\bibinfo {author} {\bibfnamefont {S.}~\bibnamefont
  {Trotzky}}, \bibinfo {author} {\bibfnamefont {Y.-A.}\ \bibnamefont {Chen}},
  \bibinfo {author} {\bibfnamefont {A.}~\bibnamefont {Flesch}}, \bibinfo
  {author} {\bibfnamefont {I.~P.}\ \bibnamefont {McCulloch}}, \bibinfo {author}
  {\bibfnamefont {U.}~\bibnamefont {Schollw\"ock}}, \bibinfo {author}
  {\bibfnamefont {J.}~\bibnamefont {Eisert}},\ and\ \bibinfo {author}
  {\bibfnamefont {I.}~\bibnamefont {Bloch}},\ }\href
  {https://doi.org/doi:10.1038/nphys2232} {\bibfield  {journal} {\bibinfo
  {journal} {Nature Phys.}\ }\textbf {\bibinfo {volume} {8}},\ \bibinfo {pages}
  {325} (\bibinfo {year} {2012})},\ \Eprint {https://arxiv.org/abs/1101.2659}
  {arXiv:1101.2659} \BibitemShut {NoStop}%
\bibitem [{\citenamefont {Schneider}\ \emph {et~al.}(2012)\citenamefont
  {Schneider}, \citenamefont {Hackerm{\"u}ller}, \citenamefont {Ronzheimer},
  \citenamefont {Will}, \citenamefont {S.~Braun}, \citenamefont {Bloch},
  \citenamefont {Demler}, \citenamefont {Mandt}, \citenamefont {Rasch},\ and\
  \citenamefont {Rosch}}]{RoschTransport}%
  \BibitemOpen
  \bibfield  {author} {\bibinfo {author} {\bibfnamefont {U.}~\bibnamefont
  {Schneider}}, \bibinfo {author} {\bibfnamefont {L.}~\bibnamefont
  {Hackerm{\"u}ller}}, \bibinfo {author} {\bibfnamefont {J.~P.}\ \bibnamefont
  {Ronzheimer}}, \bibinfo {author} {\bibfnamefont {S.}~\bibnamefont {Will}},
  \bibinfo {author} {\bibfnamefont {T.~B.}\ \bibnamefont {S.~Braun}}, \bibinfo
  {author} {\bibfnamefont {I.}~\bibnamefont {Bloch}}, \bibinfo {author}
  {\bibfnamefont {E.}~\bibnamefont {Demler}}, \bibinfo {author} {\bibfnamefont
  {S.}~\bibnamefont {Mandt}}, \bibinfo {author} {\bibfnamefont
  {D.}~\bibnamefont {Rasch}},\ and\ \bibinfo {author} {\bibfnamefont
  {A.}~\bibnamefont {Rosch}},\ }\href {https://doi.org/10.1038/nphys2205}
  {\bibfield  {journal} {\bibinfo  {journal} {Nature Phys.}\ }\textbf {\bibinfo
  {volume} {8}},\ \bibinfo {pages} {213} (\bibinfo {year} {2012})}\BibitemShut
  {NoStop}%
\bibitem [{\citenamefont {Vijayan}\ \emph {et~al.}(2020)\citenamefont
  {Vijayan}, \citenamefont {Sompet}, \citenamefont {Salomon}, \citenamefont
  {Koepsell}, \citenamefont {Hirthe}, \citenamefont {Bohrdt}, \citenamefont
  {Grusdt}, \citenamefont {Bloch},\ and\ \citenamefont
  {Gross}}]{GrossFermions}%
  \BibitemOpen
  \bibfield  {author} {\bibinfo {author} {\bibfnamefont {J.}~\bibnamefont
  {Vijayan}}, \bibinfo {author} {\bibfnamefont {P.}~\bibnamefont {Sompet}},
  \bibinfo {author} {\bibfnamefont {G.}~\bibnamefont {Salomon}}, \bibinfo
  {author} {\bibfnamefont {J.}~\bibnamefont {Koepsell}}, \bibinfo {author}
  {\bibfnamefont {S.}~\bibnamefont {Hirthe}}, \bibinfo {author} {\bibfnamefont
  {A.}~\bibnamefont {Bohrdt}}, \bibinfo {author} {\bibfnamefont
  {F.}~\bibnamefont {Grusdt}}, \bibinfo {author} {\bibfnamefont
  {I.}~\bibnamefont {Bloch}},\ and\ \bibinfo {author} {\bibfnamefont
  {C.}~\bibnamefont {Gross}},\ }\href {https://doi.org/10.1126/science.aay2354}
  {\bibfield  {journal} {\bibinfo  {journal} {Science}\ }\textbf {\bibinfo
  {volume} {367}},\ \bibinfo {pages} {186} (\bibinfo {year}
  {2020})}\BibitemShut {NoStop}%
\bibitem [{\citenamefont {Mazurenko}\ \emph {et~al.}(2017)\citenamefont
  {Mazurenko}, \citenamefont {Chiu}, \citenamefont {Ji}, \citenamefont
  {Parsons}, \citenamefont {Kanasz-Nagy}, \citenamefont {Schmidt},
  \citenamefont {Grusdt}, \citenamefont {Demler}, \citenamefont {Greif},\ and\
  \citenamefont {Greiner}}]{Mazurenko}%
  \BibitemOpen
  \bibfield  {author} {\bibinfo {author} {\bibfnamefont {A.}~\bibnamefont
  {Mazurenko}}, \bibinfo {author} {\bibfnamefont {C.~S.}\ \bibnamefont {Chiu}},
  \bibinfo {author} {\bibfnamefont {G.}~\bibnamefont {Ji}}, \bibinfo {author}
  {\bibfnamefont {M.~F.}\ \bibnamefont {Parsons}}, \bibinfo {author}
  {\bibfnamefont {M.}~\bibnamefont {Kanasz-Nagy}}, \bibinfo {author}
  {\bibfnamefont {R.}~\bibnamefont {Schmidt}}, \bibinfo {author} {\bibfnamefont
  {F.}~\bibnamefont {Grusdt}}, \bibinfo {author} {\bibfnamefont
  {E.}~\bibnamefont {Demler}}, \bibinfo {author} {\bibfnamefont
  {D.}~\bibnamefont {Greif}},\ and\ \bibinfo {author} {\bibfnamefont
  {M.}~\bibnamefont {Greiner}},\ }\href {https://doi.org/10.1038/nature22362}
  {\bibfield  {journal} {\bibinfo  {journal} {Nature}\ }\textbf {\bibinfo
  {volume} {545}},\ \bibinfo {pages} {462} (\bibinfo {year}
  {2017})}\BibitemShut {NoStop}%
\bibitem [{\citenamefont {Esslinger}(2010)}]{EsslingerReview}%
  \BibitemOpen
  \bibfield  {author} {\bibinfo {author} {\bibfnamefont {T.}~\bibnamefont
  {Esslinger}},\ }\href
  {https://doi.org/10.1146/annurev-conmatphys-070909-104059} {\bibfield
  {journal} {\bibinfo  {journal} {Ann. Rev. Con and Mat. Phys.}\ }\textbf
  {\bibinfo {volume} {1}},\ \bibinfo {pages} {129} (\bibinfo {year}
  {2010})}\BibitemShut {NoStop}%
\bibitem [{\citenamefont {Dallaire-Demers}\ \emph {et~al.}(2019)\citenamefont
  {Dallaire-Demers}, \citenamefont {Romero}, \citenamefont {Veis},
  \citenamefont {Sim},\ and\ \citenamefont
  {Aspuru-Guzik}}]{AspuruGuzikFermions}%
  \BibitemOpen
  \bibfield  {author} {\bibinfo {author} {\bibfnamefont {P.-L.}\ \bibnamefont
  {Dallaire-Demers}}, \bibinfo {author} {\bibfnamefont {J.}~\bibnamefont
  {Romero}}, \bibinfo {author} {\bibfnamefont {L.}~\bibnamefont {Veis}},
  \bibinfo {author} {\bibfnamefont {S.}~\bibnamefont {Sim}},\ and\ \bibinfo
  {author} {\bibfnamefont {A.}~\bibnamefont {Aspuru-Guzik}},\ }\href
  {https://doi.org/10.1088/2058-9565/ab3951} {\bibfield  {journal} {\bibinfo
  {journal} {Quant. Sc. Tech.}\ }\textbf {\bibinfo {volume} {4}},\ \bibinfo
  {pages} {045005} (\bibinfo {year} {2019})}\BibitemShut {NoStop}%
\bibitem [{\citenamefont {Eisert}\ \emph {et~al.}(2020)\citenamefont {Eisert},
  \citenamefont {Hangleiter}, \citenamefont {Walk}, \citenamefont {Roth},
  \citenamefont {Markham}, \citenamefont {Parekh}, \citenamefont {Chabaud},\
  and\ \citenamefont {Kashefi}}]{eisert_quantum_2020}%
  \BibitemOpen
  \bibfield  {author} {\bibinfo {author} {\bibfnamefont {J.}~\bibnamefont
  {Eisert}}, \bibinfo {author} {\bibfnamefont {D.}~\bibnamefont {Hangleiter}},
  \bibinfo {author} {\bibfnamefont {N.}~\bibnamefont {Walk}}, \bibinfo {author}
  {\bibfnamefont {I.}~\bibnamefont {Roth}}, \bibinfo {author} {\bibfnamefont
  {D.}~\bibnamefont {Markham}}, \bibinfo {author} {\bibfnamefont
  {R.}~\bibnamefont {Parekh}}, \bibinfo {author} {\bibfnamefont
  {U.}~\bibnamefont {Chabaud}},\ and\ \bibinfo {author} {\bibfnamefont
  {E.}~\bibnamefont {Kashefi}},\ }\href
  {https://doi.org/10.1038/s42254-020-0186-4} {\bibfield  {journal} {\bibinfo
  {journal} {Nature Rev. Phys.}\ }\textbf {\bibinfo {volume} {2}},\ \bibinfo
  {pages} {382} (\bibinfo {year} {2020})}\BibitemShut {NoStop}%
\bibitem [{\citenamefont {Kliesch}\ and\ \citenamefont
  {Roth}(2021)}]{kliesch_theory_2021}%
  \BibitemOpen
  \bibfield  {author} {\bibinfo {author} {\bibfnamefont {M.}~\bibnamefont
  {Kliesch}}\ and\ \bibinfo {author} {\bibfnamefont {I.}~\bibnamefont {Roth}},\
  }\href {https://doi.org/10.1103/PRXQuantum.2.010201} {\bibfield  {journal}
  {\bibinfo  {journal} {PRX Quantum}\ }\textbf {\bibinfo {volume} {2}},\
  \bibinfo {pages} {010201} (\bibinfo {year} {2021})}\BibitemShut {NoStop}%
\bibitem [{\citenamefont {Elben}\ \emph {et~al.}(2023)\citenamefont {Elben},
  \citenamefont {Flammia}, \citenamefont {Huang}, \citenamefont {Kueng},
  \citenamefont {Preskill}, \citenamefont {Vermersch},\ and\ \citenamefont
  {Zoller}}]{Toolbox}%
  \BibitemOpen
  \bibfield  {author} {\bibinfo {author} {\bibfnamefont {A.}~\bibnamefont
  {Elben}}, \bibinfo {author} {\bibfnamefont {S.~T.}\ \bibnamefont {Flammia}},
  \bibinfo {author} {\bibfnamefont {H.-Y.}\ \bibnamefont {Huang}}, \bibinfo
  {author} {\bibfnamefont {R.}~\bibnamefont {Kueng}}, \bibinfo {author}
  {\bibfnamefont {J.}~\bibnamefont {Preskill}}, \bibinfo {author}
  {\bibfnamefont {B.}~\bibnamefont {Vermersch}},\ and\ \bibinfo {author}
  {\bibfnamefont {P.}~\bibnamefont {Zoller}},\ }\href
  {https://doi.org/10.1038/s42254-022-00535-2} {\bibfield  {journal} {\bibinfo
  {journal} {Nature Rev. Phys.}\ }\textbf {\bibinfo {volume} {5}},\ \bibinfo
  {pages} {9} (\bibinfo {year} {2023})}\BibitemShut {NoStop}%
\bibitem [{\citenamefont {Magesan}\ \emph {et~al.}(2012)\citenamefont
  {Magesan}, \citenamefont {Gambetta},\ and\ \citenamefont
  {Emerson}}]{MagGamEmer}%
  \BibitemOpen
  \bibfield  {author} {\bibinfo {author} {\bibfnamefont {E.}~\bibnamefont
  {Magesan}}, \bibinfo {author} {\bibfnamefont {J.~M.}\ \bibnamefont
  {Gambetta}},\ and\ \bibinfo {author} {\bibfnamefont {J.}~\bibnamefont
  {Emerson}},\ }\href {https://doi.org/10.1103/PhysRevA.85.042311} {\bibfield
  {journal} {\bibinfo  {journal} {Phys. Rev. Lett.}\ }\textbf {\bibinfo
  {volume} {85}},\ \bibinfo {pages} {042311} (\bibinfo {year}
  {2012})}\BibitemShut {NoStop}%
\bibitem [{\citenamefont {Knill}\ \emph {et~al.}(2008)\citenamefont {Knill},
  \citenamefont {Leibfried}, \citenamefont {Reichle}, \citenamefont {Britton},
  \citenamefont {Blakestad}, \citenamefont {J.}, \citenamefont {Jost},
  \citenamefont {Langer}, \citenamefont {Ozeri}, \citenamefont {Seidelin},\
  and\ \citenamefont {Wineland}}]{KnillBenchmarking}%
  \BibitemOpen
  \bibfield  {author} {\bibinfo {author} {\bibfnamefont {E.}~\bibnamefont
  {Knill}}, \bibinfo {author} {\bibnamefont {Leibfried}}, \bibinfo {author}
  {\bibfnamefont {R.}~\bibnamefont {Reichle}}, \bibinfo {author} {\bibfnamefont
  {J.}~\bibnamefont {Britton}}, \bibinfo {author} {\bibfnamefont {R.~B.}\
  \bibnamefont {Blakestad}}, \bibinfo {author} {\bibnamefont {J.}}, \bibinfo
  {author} {\bibnamefont {Jost}}, \bibinfo {author} {\bibfnamefont
  {C.}~\bibnamefont {Langer}}, \bibinfo {author} {\bibfnamefont
  {R.}~\bibnamefont {Ozeri}}, \bibinfo {author} {\bibfnamefont
  {S.}~\bibnamefont {Seidelin}},\ and\ \bibinfo {author} {\bibfnamefont
  {J.}~\bibnamefont {Wineland}},\ }\href
  {https://doi.org/10.1103/PhysRevA.77.012307} {\bibfield  {journal} {\bibinfo
  {journal} {Phys. Rev. A}\ }\textbf {\bibinfo {volume} {77}},\ \bibinfo
  {pages} {012307} (\bibinfo {year} {2008})}\BibitemShut {NoStop}%
\bibitem [{\citenamefont {Huang}\ \emph {et~al.}(2020)\citenamefont {Huang},
  \citenamefont {Kueng},\ and\ \citenamefont
  {Preskill}}]{huang_predicting_2020}%
  \BibitemOpen
  \bibfield  {author} {\bibinfo {author} {\bibfnamefont {H.-Y.}\ \bibnamefont
  {Huang}}, \bibinfo {author} {\bibfnamefont {R.}~\bibnamefont {Kueng}},\ and\
  \bibinfo {author} {\bibfnamefont {J.}~\bibnamefont {Preskill}},\ }\href
  {https://doi.org/10.1038/s41567-020-0932-7} {\bibfield  {journal} {\bibinfo
  {journal} {Nature Phys.}\ }\textbf {\bibinfo {volume} {16}},\ \bibinfo
  {pages} {1050} (\bibinfo {year} {2020})}\BibitemShut {NoStop}%
\bibitem [{\citenamefont {Huang}\ \emph {et~al.}(2022)\citenamefont {Huang},
  \citenamefont {Broughton}, \citenamefont {Cotler}, \citenamefont {Chen},
  \citenamefont {Li}, \citenamefont {Mohseni}, \citenamefont {Neven},
  \citenamefont {Babbush}, \citenamefont {Kueng}, \citenamefont {Preskill},\
  and\ \citenamefont {McClean}}]{huang_quantum_2021}%
  \BibitemOpen
  \bibfield  {author} {\bibinfo {author} {\bibfnamefont {H.-Y.}\ \bibnamefont
  {Huang}}, \bibinfo {author} {\bibfnamefont {M.}~\bibnamefont {Broughton}},
  \bibinfo {author} {\bibfnamefont {J.}~\bibnamefont {Cotler}}, \bibinfo
  {author} {\bibfnamefont {S.}~\bibnamefont {Chen}}, \bibinfo {author}
  {\bibfnamefont {J.}~\bibnamefont {Li}}, \bibinfo {author} {\bibfnamefont
  {M.}~\bibnamefont {Mohseni}}, \bibinfo {author} {\bibfnamefont
  {H.}~\bibnamefont {Neven}}, \bibinfo {author} {\bibfnamefont
  {R.}~\bibnamefont {Babbush}}, \bibinfo {author} {\bibfnamefont
  {R.}~\bibnamefont {Kueng}}, \bibinfo {author} {\bibfnamefont
  {J.}~\bibnamefont {Preskill}},\ and\ \bibinfo {author} {\bibfnamefont
  {J.~R.}\ \bibnamefont {McClean}},\ }\href
  {https://doi.org/10.1126/science.abn7293} {\bibfield  {journal} {\bibinfo
  {journal} {Science}\ }\textbf {\bibinfo {volume} {376}},\ \bibinfo {pages}
  {1182} (\bibinfo {year} {2022})}\BibitemShut {NoStop}%
\bibitem [{\citenamefont {Chen}\ \emph {et~al.}(2021)\citenamefont {Chen},
  \citenamefont {Yu}, \citenamefont {Zeng},\ and\ \citenamefont
  {Flammia}}]{chen_robust_2021}%
  \BibitemOpen
  \bibfield  {author} {\bibinfo {author} {\bibfnamefont {S.}~\bibnamefont
  {Chen}}, \bibinfo {author} {\bibfnamefont {W.}~\bibnamefont {Yu}}, \bibinfo
  {author} {\bibfnamefont {P.}~\bibnamefont {Zeng}},\ and\ \bibinfo {author}
  {\bibfnamefont {S.~T.}\ \bibnamefont {Flammia}},\ }\href
  {https://doi.org/10.1103/PRXQuantum.2.030348} {\bibfield  {journal} {\bibinfo
   {journal} {PRX Quantum}\ }\textbf {\bibinfo {volume} {2}},\ \bibinfo {pages}
  {030348} (\bibinfo {year} {2021})}\BibitemShut {NoStop}%
\bibitem [{\citenamefont {Bertoni}\ \emph {et~al.}(2023)\citenamefont
  {Bertoni}, \citenamefont {Haferkamp}, \citenamefont {Hinsche}, \citenamefont
  {Ioannou}, \citenamefont {Eisert},\ and\ \citenamefont
  {Pashayan}}]{ShallowShadows}%
  \BibitemOpen
  \bibfield  {author} {\bibinfo {author} {\bibfnamefont {C.}~\bibnamefont
  {Bertoni}}, \bibinfo {author} {\bibfnamefont {J.}~\bibnamefont {Haferkamp}},
  \bibinfo {author} {\bibfnamefont {M.}~\bibnamefont {Hinsche}}, \bibinfo
  {author} {\bibfnamefont {M.}~\bibnamefont {Ioannou}}, \bibinfo {author}
  {\bibfnamefont {J.}~\bibnamefont {Eisert}},\ and\ \bibinfo {author}
  {\bibfnamefont {H.}~\bibnamefont {Pashayan}},\ }\href@noop {} {\  (\bibinfo
  {year} {2023})},\ \Eprint {https://arxiv.org/abs/2209.12924}
  {arXiv:2209.12924} \BibitemShut {NoStop}%
\bibitem [{\citenamefont {Ohliger}\ \emph {et~al.}(2013)\citenamefont
  {Ohliger}, \citenamefont {Nesme},\ and\ \citenamefont
  {Eisert}}]{ohliger_efficient_2013}%
  \BibitemOpen
  \bibfield  {author} {\bibinfo {author} {\bibfnamefont {M.}~\bibnamefont
  {Ohliger}}, \bibinfo {author} {\bibfnamefont {V.}~\bibnamefont {Nesme}},\
  and\ \bibinfo {author} {\bibfnamefont {J.}~\bibnamefont {Eisert}},\ }\href
  {https://doi.org/10.1088/1367-2630/15/1/015024} {\bibfield  {journal}
  {\bibinfo  {journal} {New J. Phys.}\ }\textbf {\bibinfo {volume} {15}},\
  \bibinfo {pages} {015024} (\bibinfo {year} {2013})}\BibitemShut {NoStop}%
\bibitem [{\citenamefont {Zhao}\ \emph {et~al.}(2021)\citenamefont {Zhao},
  \citenamefont {Rubin},\ and\ \citenamefont {Miyake}}]{zhao_fermionic_2021}%
  \BibitemOpen
  \bibfield  {author} {\bibinfo {author} {\bibfnamefont {A.}~\bibnamefont
  {Zhao}}, \bibinfo {author} {\bibfnamefont {N.~C.}\ \bibnamefont {Rubin}},\
  and\ \bibinfo {author} {\bibfnamefont {A.}~\bibnamefont {Miyake}},\ }\href
  {https://doi.org/10.1103/PhysRevLett.127.110504} {\bibfield  {journal}
  {\bibinfo  {journal} {Phys. Rev. Lett.}\ }\textbf {\bibinfo {volume} {127}},\
  \bibinfo {pages} {110504} (\bibinfo {year} {2021})}\BibitemShut {NoStop}%
\bibitem [{\citenamefont {Wan}\ \emph {et~al.}(2022)\citenamefont {Wan},
  \citenamefont {Huggins}, \citenamefont {Lee},\ and\ \citenamefont
  {Babbush}}]{Wan_2022}%
  \BibitemOpen
  \bibfield  {author} {\bibinfo {author} {\bibfnamefont {K.}~\bibnamefont
  {Wan}}, \bibinfo {author} {\bibfnamefont {W.~J.}\ \bibnamefont {Huggins}},
  \bibinfo {author} {\bibfnamefont {J.}~\bibnamefont {Lee}},\ and\ \bibinfo
  {author} {\bibfnamefont {R.}~\bibnamefont {Babbush}},\ }\href
  {https://arxiv.org/abs/2207.13723} {\  (\bibinfo {year} {2022})},\ \Eprint
  {https://arxiv.org/abs/2207.13723} {arXiv:2207.13723} \BibitemShut {NoStop}%
\bibitem [{\citenamefont {Low}(2022)}]{GuangHao+22}%
  \BibitemOpen
  \bibfield  {author} {\bibinfo {author} {\bibfnamefont {G.~H.}\ \bibnamefont
  {Low}},\ }\href@noop {} {\  (\bibinfo {year} {2022})},\ \Eprint
  {https://arxiv.org/abs/2208.08964} {arXiv:2208.08964} \BibitemShut {NoStop}%
\bibitem [{\citenamefont {Weitenberg}\ \emph {et~al.}(2011)\citenamefont
  {Weitenberg}, \citenamefont {Endres}, \citenamefont {Sherson}, \citenamefont
  {Cheneau}, \citenamefont {Schau{\ss}}, \citenamefont {Fukuhara},
  \citenamefont {Bloch},\ and\ \citenamefont {Kuhr}}]{Kuhr}%
  \BibitemOpen
  \bibfield  {author} {\bibinfo {author} {\bibfnamefont {C.}~\bibnamefont
  {Weitenberg}}, \bibinfo {author} {\bibfnamefont {M.}~\bibnamefont {Endres}},
  \bibinfo {author} {\bibfnamefont {J.~F.}\ \bibnamefont {Sherson}}, \bibinfo
  {author} {\bibfnamefont {M.}~\bibnamefont {Cheneau}}, \bibinfo {author}
  {\bibfnamefont {P.}~\bibnamefont {Schau{\ss}}}, \bibinfo {author}
  {\bibfnamefont {T.}~\bibnamefont {Fukuhara}}, \bibinfo {author}
  {\bibfnamefont {I.}~\bibnamefont {Bloch}},\ and\ \bibinfo {author}
  {\bibfnamefont {S.}~\bibnamefont {Kuhr}},\ }\href
  {https://doi.org/10.1038/nature09827} {\bibfield  {journal} {\bibinfo
  {journal} {Nature}\ }\textbf {\bibinfo {volume} {471}},\ \bibinfo {pages}
  {319} (\bibinfo {year} {2011})}\BibitemShut {NoStop}%
\bibitem [{\citenamefont {Aidelsburger}(2016)}]{aidelsburger_artificial_2016}%
  \BibitemOpen
  \bibfield  {author} {\bibinfo {author} {\bibfnamefont {M.}~\bibnamefont
  {Aidelsburger}},\ }\href {https://doi.org/10.1007/978-3-319-25829-4} {\emph
  {\bibinfo {title} {Artificial {Gauge} {Fields} with {Ultracold} {Atoms} in
  {Optical} {Lattices}}}},\ Springer {Theses}\ (\bibinfo  {publisher} {Springer
  International Publishing},\ \bibinfo {address} {Cham},\ \bibinfo {year}
  {2016})\BibitemShut {NoStop}%
\bibitem [{\citenamefont {Bakr}\ \emph {et~al.}(2009)\citenamefont {Bakr},
  \citenamefont {Gillen}, \citenamefont {Peng}, \citenamefont {F{\"o}lling},\
  and\ \citenamefont {Greiner}}]{GreinerMicroscope}%
  \BibitemOpen
  \bibfield  {author} {\bibinfo {author} {\bibfnamefont {W.~S.}\ \bibnamefont
  {Bakr}}, \bibinfo {author} {\bibfnamefont {J.~I.}\ \bibnamefont {Gillen}},
  \bibinfo {author} {\bibfnamefont {A.}~\bibnamefont {Peng}}, \bibinfo {author}
  {\bibfnamefont {S.}~\bibnamefont {F{\"o}lling}},\ and\ \bibinfo {author}
  {\bibfnamefont {M.}~\bibnamefont {Greiner}},\ }\href
  {https://doi.org/10.1038/nature08482} {\bibfield  {journal} {\bibinfo
  {journal} {Nature}\ }\textbf {\bibinfo {volume} {462}},\ \bibinfo {pages}
  {74} (\bibinfo {year} {2009})}\BibitemShut {NoStop}%
\bibitem [{\citenamefont {Endres}\ \emph {et~al.}(2011)\citenamefont {Endres},
  \citenamefont {Cheneau}, \citenamefont {Fukuhara}, \citenamefont
  {Weitenberg}, \citenamefont {Schau{\ss}}, \citenamefont {Gross},
  \citenamefont {Mazza}, \citenamefont {Banuls}, \citenamefont {Pollet},
  \citenamefont {Bloch},\ and\ \citenamefont {Kuhr}}]{Microscope}%
  \BibitemOpen
  \bibfield  {author} {\bibinfo {author} {\bibfnamefont {M.}~\bibnamefont
  {Endres}}, \bibinfo {author} {\bibfnamefont {M.}~\bibnamefont {Cheneau}},
  \bibinfo {author} {\bibfnamefont {T.}~\bibnamefont {Fukuhara}}, \bibinfo
  {author} {\bibfnamefont {C.}~\bibnamefont {Weitenberg}}, \bibinfo {author}
  {\bibfnamefont {P.}~\bibnamefont {Schau{\ss}}}, \bibinfo {author}
  {\bibfnamefont {C.}~\bibnamefont {Gross}}, \bibinfo {author} {\bibfnamefont
  {L.}~\bibnamefont {Mazza}}, \bibinfo {author} {\bibfnamefont {M.~C.}\
  \bibnamefont {Banuls}}, \bibinfo {author} {\bibfnamefont {L.}~\bibnamefont
  {Pollet}}, \bibinfo {author} {\bibfnamefont {I.}~\bibnamefont {Bloch}},\ and\
  \bibinfo {author} {\bibfnamefont {S.}~\bibnamefont {Kuhr}},\ }\href
  {https://doi.org/10.1126/science.1209284} {\bibfield  {journal} {\bibinfo
  {journal} {Science}\ }\textbf {\bibinfo {volume} {334}},\ \bibinfo {pages}
  {200} (\bibinfo {year} {2011})}\BibitemShut {NoStop}%
\bibitem [{\citenamefont {Omran}\ \emph {et~al.}(2015)\citenamefont {Omran},
  \citenamefont {Boll}, \citenamefont {Hilker}, \citenamefont {Kleinlein},
  \citenamefont {Salomon}, \citenamefont {Bloch},\ and\ \citenamefont
  {Gross}}]{PhysRevLett.115.263001}%
  \BibitemOpen
  \bibfield  {author} {\bibinfo {author} {\bibfnamefont {A.}~\bibnamefont
  {Omran}}, \bibinfo {author} {\bibfnamefont {M.}~\bibnamefont {Boll}},
  \bibinfo {author} {\bibfnamefont {T.~A.}\ \bibnamefont {Hilker}}, \bibinfo
  {author} {\bibfnamefont {K.}~\bibnamefont {Kleinlein}}, \bibinfo {author}
  {\bibfnamefont {G.}~\bibnamefont {Salomon}}, \bibinfo {author} {\bibfnamefont
  {I.}~\bibnamefont {Bloch}},\ and\ \bibinfo {author} {\bibfnamefont
  {C.}~\bibnamefont {Gross}},\ }\href
  {https://doi.org/10.1103/PhysRevLett.115.263001} {\bibfield  {journal}
  {\bibinfo  {journal} {Phys. Rev. Lett.}\ }\textbf {\bibinfo {volume} {115}},\
  \bibinfo {pages} {263001} (\bibinfo {year} {2015})}\BibitemShut {NoStop}%
\bibitem [{\citenamefont {Gluza}\ \emph {et~al.}(2020)\citenamefont {Gluza},
  \citenamefont {Schweigler}, \citenamefont {Rauer}, \citenamefont {Krumnow},
  \citenamefont {Schmiedmayer},\ and\ \citenamefont {Eisert}}]{QuantumReadout}%
  \BibitemOpen
  \bibfield  {author} {\bibinfo {author} {\bibfnamefont {M.}~\bibnamefont
  {Gluza}}, \bibinfo {author} {\bibfnamefont {T.}~\bibnamefont {Schweigler}},
  \bibinfo {author} {\bibfnamefont {B.}~\bibnamefont {Rauer}}, \bibinfo
  {author} {\bibfnamefont {C.}~\bibnamefont {Krumnow}}, \bibinfo {author}
  {\bibfnamefont {J.}~\bibnamefont {Schmiedmayer}},\ and\ \bibinfo {author}
  {\bibfnamefont {J.}~\bibnamefont {Eisert}},\ }\href
  {https://doi.org/10.1038/s42005-019-0273-y} {\bibfield  {journal} {\bibinfo
  {journal} {Physics Comm.}\ }\textbf {\bibinfo {volume} {3}},\ \bibinfo
  {pages} {12} (\bibinfo {year} {2020})}\BibitemShut {NoStop}%
\bibitem [{\citenamefont {Gluza}\ and\ \citenamefont
  {Eisert}(2021)}]{gluza_recovering_2021}%
  \BibitemOpen
  \bibfield  {author} {\bibinfo {author} {\bibfnamefont {M.}~\bibnamefont
  {Gluza}}\ and\ \bibinfo {author} {\bibfnamefont {J.}~\bibnamefont {Eisert}},\
  }\href {https://doi.org/10.1103/PhysRevLett.127.090503} {\bibfield  {journal}
  {\bibinfo  {journal} {Phys. Rev. Lett.}\ }\textbf {\bibinfo {volume} {127}},\
  \bibinfo {pages} {090503} (\bibinfo {year} {2021})}\BibitemShut {NoStop}%
\bibitem [{\citenamefont {Naldesi}\ \emph {et~al.}(2023)\citenamefont
  {Naldesi}, \citenamefont {Elben}, \citenamefont {Minguzzi}, \citenamefont
  {Cl\'ement}, \citenamefont {Zoller},\ and\ \citenamefont
  {Vermersch}}]{NaldesiZoller}%
  \BibitemOpen
  \bibfield  {author} {\bibinfo {author} {\bibfnamefont {P.}~\bibnamefont
  {Naldesi}}, \bibinfo {author} {\bibfnamefont {A.}~\bibnamefont {Elben}},
  \bibinfo {author} {\bibfnamefont {A.}~\bibnamefont {Minguzzi}}, \bibinfo
  {author} {\bibfnamefont {D.}~\bibnamefont {Cl\'ement}}, \bibinfo {author}
  {\bibfnamefont {P.}~\bibnamefont {Zoller}},\ and\ \bibinfo {author}
  {\bibfnamefont {B.}~\bibnamefont {Vermersch}},\ }\href
  {https://doi.org/10.1103/PhysRevLett.131.060601} {\bibfield  {journal}
  {\bibinfo  {journal} {Phys. Rev. Lett.}\ }\textbf {\bibinfo {volume} {131}},\
  \bibinfo {pages} {060601} (\bibinfo {year} {2023})}\BibitemShut {NoStop}%
\bibitem [{\citenamefont {Tran}\ \emph {et~al.}(2022)\citenamefont {Tran},
  \citenamefont {Mark}, \citenamefont {Ho},\ and\ \citenamefont
  {Choi}}]{Tran+22}%
  \BibitemOpen
  \bibfield  {author} {\bibinfo {author} {\bibfnamefont {M.~C.}\ \bibnamefont
  {Tran}}, \bibinfo {author} {\bibfnamefont {D.~K.}\ \bibnamefont {Mark}},
  \bibinfo {author} {\bibfnamefont {W.~W.}\ \bibnamefont {Ho}},\ and\ \bibinfo
  {author} {\bibfnamefont {S.}~\bibnamefont {Choi}},\ }\href@noop {} {\
  (\bibinfo {year} {2022})},\ \Eprint {https://arxiv.org/abs/2212.02517}
  {arXiv:2212.02517} \BibitemShut {NoStop}%
\bibitem [{\citenamefont {Kirk}\ \emph {et~al.}(2022)\citenamefont {Kirk},
  \citenamefont {Cotler}, \citenamefont {Huang},\ and\ \citenamefont
  {Lukin}}]{vankirk2022hardwareefficient}%
  \BibitemOpen
  \bibfield  {author} {\bibinfo {author} {\bibfnamefont {K.~V.}\ \bibnamefont
  {Kirk}}, \bibinfo {author} {\bibfnamefont {J.}~\bibnamefont {Cotler}},
  \bibinfo {author} {\bibfnamefont {H.-Y.}\ \bibnamefont {Huang}},\ and\
  \bibinfo {author} {\bibfnamefont {M.~D.}\ \bibnamefont {Lukin}},\ }\href@noop
  {} {\bibinfo {title} {Hardware-efficient learning of quantum many-body
  states}} (\bibinfo {year} {2022}),\ \Eprint
  {https://arxiv.org/abs/2212.06084} {arXiv:2212.06084 [quant-ph]} \BibitemShut
  {NoStop}%
\bibitem [{sup()}]{supp}%
  \BibitemOpen
  \href@noop {} {\bibinfo {title} {See {S}upplemental {M}aterial at [url] for
  additional numerics and detailed derivations of measurement operator
  properties, as well as sample complexity bounds, which includes {R}efs.
  [36-40]}}\BibitemShut {NoStop}%
\bibitem [{\citenamefont {Cerezo}\ \emph {et~al.}(2021)\citenamefont {Cerezo},
  \citenamefont {Arrasmith}, \citenamefont {Babbush}, \citenamefont {Benjamin},
  \citenamefont {Endo}, \citenamefont {Fujii}, \citenamefont {McClean},
  \citenamefont {Mitarai}, \citenamefont {Yuan}, \citenamefont {Cincio},\ and\
  \citenamefont {Coles}}]{Variational}%
  \BibitemOpen
  \bibfield  {author} {\bibinfo {author} {\bibfnamefont {M.}~\bibnamefont
  {Cerezo}}, \bibinfo {author} {\bibfnamefont {A.}~\bibnamefont {Arrasmith}},
  \bibinfo {author} {\bibfnamefont {R.}~\bibnamefont {Babbush}}, \bibinfo
  {author} {\bibfnamefont {S.~C.}\ \bibnamefont {Benjamin}}, \bibinfo {author}
  {\bibfnamefont {S.}~\bibnamefont {Endo}}, \bibinfo {author} {\bibfnamefont
  {K.}~\bibnamefont {Fujii}}, \bibinfo {author} {\bibfnamefont {J.~R.}\
  \bibnamefont {McClean}}, \bibinfo {author} {\bibfnamefont {K.}~\bibnamefont
  {Mitarai}}, \bibinfo {author} {\bibfnamefont {X.}~\bibnamefont {Yuan}},
  \bibinfo {author} {\bibfnamefont {L.}~\bibnamefont {Cincio}},\ and\ \bibinfo
  {author} {\bibfnamefont {P.~J.}\ \bibnamefont {Coles}},\ }\href
  {https://doi.org/10.1038/s42254-021-00348-9} {\bibfield  {journal} {\bibinfo
  {journal} {Nature Rev. Phys.}\ }\textbf {\bibinfo {volume} {3}},\ \bibinfo
  {pages} {625} (\bibinfo {year} {2021})}\BibitemShut {NoStop}%
\bibitem [{\citenamefont {Eisert}\ \emph {et~al.}(2007)\citenamefont {Eisert},
  \citenamefont {Brandao},\ and\ \citenamefont {Audenaert}}]{quant-ph/0607167}%
  \BibitemOpen
  \bibfield  {author} {\bibinfo {author} {\bibfnamefont {J.}~\bibnamefont
  {Eisert}}, \bibinfo {author} {\bibfnamefont {F.~G. S.~L.}\ \bibnamefont
  {Brandao}},\ and\ \bibinfo {author} {\bibfnamefont {K.~M.}\ \bibnamefont
  {Audenaert}},\ }\href {https://doi.org/10.1088/1367-2630/9/3/046} {\bibfield
  {journal} {\bibinfo  {journal} {New J. Phys.}\ }\textbf {\bibinfo {volume}
  {9}},\ \bibinfo {pages} {46} (\bibinfo {year} {2007})}\BibitemShut {NoStop}%
\bibitem [{\citenamefont {Audenaert}\ and\ \citenamefont
  {Plenio}(2006)}]{Audenaert06}%
  \BibitemOpen
  \bibfield  {author} {\bibinfo {author} {\bibfnamefont {K.~M.~R.}\
  \bibnamefont {Audenaert}}\ and\ \bibinfo {author} {\bibfnamefont {M.~B.}\
  \bibnamefont {Plenio}},\ }\href {https://doi.org/10.1088/1367-2630/8/11/266}
  {\bibfield  {journal} {\bibinfo  {journal} {New J. Phys.}\ }\textbf {\bibinfo
  {volume} {8}},\ \bibinfo {pages} {266} (\bibinfo {year} {2006})}\BibitemShut
  {NoStop}%
\bibitem [{\citenamefont {Bennett}\ \emph {et~al.}(2020)\citenamefont
  {Bennett}, \citenamefont {Melchers},\ and\ \citenamefont
  {Proppe}}]{bennett_curta_2020}%
  \BibitemOpen
  \bibfield  {author} {\bibinfo {author} {\bibfnamefont {L.}~\bibnamefont
  {Bennett}}, \bibinfo {author} {\bibfnamefont {B.}~\bibnamefont {Melchers}},\
  and\ \bibinfo {author} {\bibfnamefont {B.}~\bibnamefont {Proppe}},\ }\href
  {https://doi.org/10.17169/refubium-26754} {\bibinfo {title} {Curta: {A}
  {General}-purpose {High}-{Performance} {Computer} at {ZEDAT}, {Freie}
  {Universität} {Berlin}}} (\bibinfo {year} {2020})\BibitemShut {NoStop}%
\bibitem [{\citenamefont {Surace}\ and\ \citenamefont
  {Tagliacozzo}(2022)}]{tagliacozzo}%
  \BibitemOpen
  \bibfield  {author} {\bibinfo {author} {\bibfnamefont {J.}~\bibnamefont
  {Surace}}\ and\ \bibinfo {author} {\bibfnamefont {L.}~\bibnamefont
  {Tagliacozzo}},\ }\href {https://doi.org/10.21468/SciPostPhysLectNotes.54}
  {\bibfield  {journal} {\bibinfo  {journal} {SciPost Phys. Lect. Notes}\ ,\
  \bibinfo {pages} {54}} (\bibinfo {year} {2022})}\BibitemShut {NoStop}%
\bibitem [{\citenamefont {Jozsa}\ and\ \citenamefont
  {Miyake}(2008)}]{Jozsa_2008}%
  \BibitemOpen
  \bibfield  {author} {\bibinfo {author} {\bibfnamefont {R.}~\bibnamefont
  {Jozsa}}\ and\ \bibinfo {author} {\bibfnamefont {A.}~\bibnamefont {Miyake}},\
  }\href {https://doi.org/10.1098/rspa.2008.0189} {\bibfield  {journal}
  {\bibinfo  {journal} {Proc. Roy. Soc. A}\ }\textbf {\bibinfo {volume}
  {464}},\ \bibinfo {pages} {3089} (\bibinfo {year} {2008})}\BibitemShut
  {NoStop}%
\bibitem [{\citenamefont {Simon}(1996)}]{Simon1995RepresentationsOF}%
  \BibitemOpen
  \bibfield  {author} {\bibinfo {author} {\bibfnamefont {B.}~\bibnamefont
  {Simon}},\ }\href@noop {} {\emph {\bibinfo {title} {Representations of finite
  and compact groups}}},\ Graduate Studies in Mathematics\ (\bibinfo
  {publisher} {AMS},\ \bibinfo {year} {1996})\BibitemShut {NoStop}%
\bibitem [{\citenamefont {Pontrjagin}(1939)}]{pontrjagin1939topological}%
  \BibitemOpen
  \bibfield  {author} {\bibinfo {author} {\bibfnamefont {L.}~\bibnamefont
  {Pontrjagin}},\ }\href {https://books.google.de/books?id=tRqrxgEACAAJ} {\emph
  {\bibinfo {title} {Topological Groups}}},\ Princeton Mathematical Series\
  (\bibinfo  {publisher} {Princeton University Press},\ \bibinfo {year}
  {1939})\ \bibinfo {note} {pag. 67}\BibitemShut {NoStop}%
\bibitem [{\citenamefont {Johnsen}\ and\ \citenamefont
  {Albrechts}(1985)}]{johnsen1985lineare}%
  \BibitemOpen
  \bibfield  {author} {\bibinfo {author} {\bibfnamefont {K.}~\bibnamefont
  {Johnsen}}\ and\ \bibinfo {author} {\bibfnamefont {C.}~\bibnamefont
  {Albrechts}},\ }\href@noop {} {\bibfield  {journal} {\bibinfo  {journal}
  {Elemente der Mathematik}\ }\textbf {\bibinfo {volume} {40}},\ \bibinfo
  {pages} {57} (\bibinfo {year} {1985})}\BibitemShut {NoStop}%
\end{thebibliography}%
